\documentclass[twocolumn]{IEEEtran}
\usepackage{cite}
\usepackage{graphicx}
\usepackage[cmex10]{amsmath}
\usepackage{amsfonts}
\usepackage{amssymb}
\usepackage{amsbsy}
\usepackage{amsthm}
\usepackage[mathscr]{eucal}
\usepackage[normal]{subfigure}
\usepackage{algorithm}
\usepackage{algorithmic}
\usepackage{comment}
\usepackage{paralist}
\usepackage{boxedminipage}
\usepackage{array}
\usepackage{multirow}
\usepackage{float}
\usepackage{color}

\newtheorem{lemma}{Lemma}
\newtheorem{theorem}{Theorem}

\DeclareMathOperator{\ex}{\mathbb{E}}

\DeclareMathOperator{\tr}{tr}
\DeclareMathOperator{\diag}{diag}
\DeclareMathOperator{\rank}{rank}

\DeclareMathOperator{\convhull}{co}

\begin{document}

\title{Theory and Experiment for Wireless-Powered Sensor Networks: How to Keep Sensors Alive}

\author{Kae Won Choi, \IEEEmembership{Senior Member, IEEE}, Phisca Aditya Rosyady, Lorenz Ginting, Arif Abdul Aziz, Dedi Setiawan, and Dong In Kim, \IEEEmembership{Senior Member, IEEE}
\thanks{K.~W.~Choi, A.~A.~Aziz, and D.~I.~Kim are with the School of Information and Communication Engineering, Sungkyunkwan University (SKKU), Suwon, Korea (email: kaewon.choi@gmail.com, arif.abdul.aziz92@gmail.com, and dikim@skku.ac.kr).}
\thanks{P.~A.~Rosyady and L.~Ginting are with the Dept.~of Computer Science and Engineering, Seoul National University of Science and Technology, Korea (email: adityaphisca@gmail.com and lorenzgins@gmail.com).}
\thanks{D.~Setiawan is with the Convergence Institute of Biomedical Engineering \& Biomaterials, Seoul National University of Science and Technology, Korea (email: morethanubabe@gmail.com).}
\thanks{This work was supported by the National Research Foundation of Korea (NRF) grant funded by the Korean government (MSIP) (2014R1A5A1011478).}
}

\maketitle

\begin{abstract}
In this paper, we investigate a multi-node multi-antenna wireless-powered sensor networks (WPSN) comprised of one power beacon and multiple sensor nodes.
We have implemented a real-life multi-node multi-antenna WPSN testbed that operates in real time.
We propose a beam-splitting beamforming technique that enables a power beacon to split microwave energy beams towards multiple nodes for simultaneous charging.
We experimentally demonstrate that the beam-splitting beamforming technique achieves the Pareto optimality.
For perpetual operation of the sensor nodes, we adopt an energy neutral control algorithm that keeps a sensor node alive by balancing the harvested and consumed power.
The joint beam-splitting and energy neutral control algorithm is designed by means of the Lyapunov optimization technique.
By experiments, we have shown that the proposed algorithm can successfully keep all sensor nodes alive by optimally splitting energy beams towards multiple sensor nodes.
\end{abstract}


\section{Introduction}\label{section:introduction}

The RF energy transfer is a type of wireless power transfer (WPT) techniques, which makes use of electromagnetic radiation for far-field power transfer \cite{Huang:2015}.
Even though the amount of delivered power steeply deteriorates over distances due to path loss, the RF energy transfer technique can transfer energy enough to power up a sensor node which requires small energy for its operation \cite{Xie:2013}.

In this paper, we study a wireless-powered sensor network (WPSN) with one power beacon that utilizes multiple transmit antennas to power up multiple sensor nodes simultaneously, as described in Fig.~\ref{fig:model}.
We have conducted experiment-driven researches on the WPSN in our previous works \cite{Choi:2016_2} and \cite{Choi:2017}.
In \cite{Choi:2016_2}, we have provided a comprehensive system model of the WPSN with a single transmit antenna based on the experimental results on a real-life testbed.
In another work \cite{Choi:2017}, we have studied the multi-antenna WPSN, in which the power beacon is able to steer the microwave energy beam towards specific directions.
In \cite{Choi:2017}, we have proposed a channel estimation algorithm for energy beamforming and an adaptive duty cycling algorithm for energy neutral operation.
In addition, the proposed algorithms have been implemented on a full-fledged WPSN testbed, and were verified by extensive experiments.

\begin{figure}
	\centering
    \includegraphics[width=8cm, bb=1in 0.2in 11.3in 9.2in]{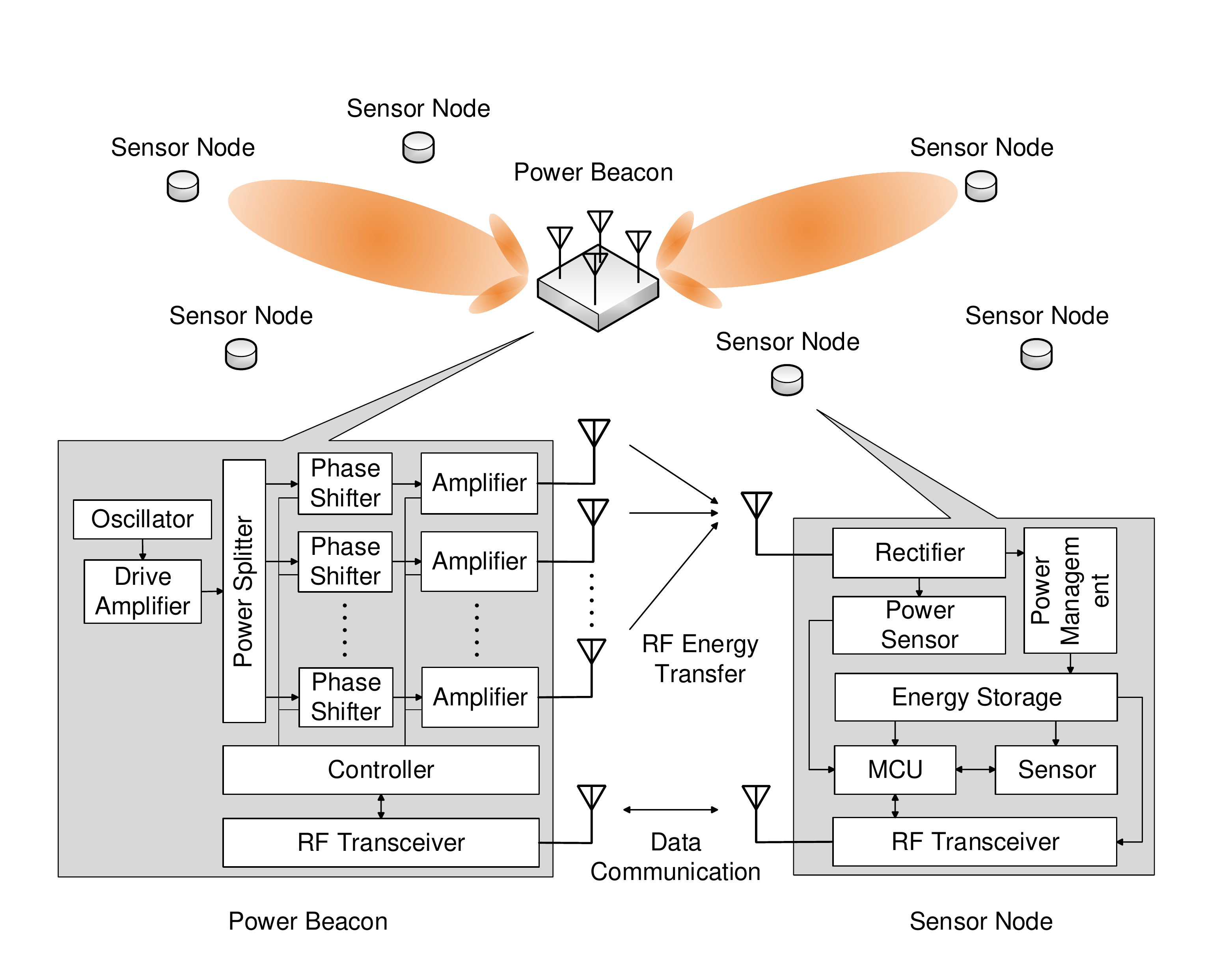}
    \caption{Multi-node multi-antenna WPSN model.}
    \label{fig:model}
\end{figure}

This paper extends our previous work \cite{Choi:2017} for the WPSN with only a single sensor node to a more generalized WPSN model accommodating multiple sensor nodes.
In the multi-node context, it is a challenging task to simultaneously charge multiple sensor nodes by dynamically forming multiple microwave energy beams.
In addition, the beamforming method should be designed tightly coupled with the energy neutral control algorithm that keeps a sensor node alive by balancing the harvested and consumed power.
Therefore, we focus on designing a joint beam-splitting and energy neutral control algorithm to support multiple sensor nodes in this paper.

We have built a multi-node multi-antenna WPSN testbed that operates in real time, by assembling commercial off-the-shelf (COTS) devices.
In this testbed, the sensor nodes operate without any other power source but the RF energy transferred from the power beacon.
Considering that the practical WPT model is not well established yet, it is of paramount importance to experimentally investigate the characteristics of the multi-node multi-antenna WPSN.
This WPSN testbed has facilitated the establishment of the realistic WPSN system model, which has helped us to identify core problems that needs to be solved in the multi-node multi-antenna WPSN design.

The WPT for simultaneously charging multiple devices has been studied by some previous works \cite{Liu:2014,Shi:2014,Son:2014,Khandaker:2014,Wang:2016,Bi:2016, Waters:2015,Mishra:20152}.
The authors of \cite{Liu:2014} have researched the multi-node multi-antenna wireless-powered communication network (WPCN), in which the information is transmitted in the uplink by using the energy transferred in the downlink.
In \cite{Shi:2014,Son:2014,Khandaker:2014}, a multiuser multiple-input single-output (MISO) downlink system is investigated for the simultaneous wireless information and power transfer (SWIPT).
Simultaneously charging multiple sensor nodes by the directional WPT in large-scale sensor networks is studied in \cite{Wang:2016}.
In \cite{Bi:2016}, a distributed transmit power allocation method is designed when multiple energy transmission and reception nodes coexist in the network.
However, all these researches \cite{Liu:2014,Shi:2014,Son:2014,Khandaker:2014,Wang:2016,Bi:2016} are theoretical works lacking experimental validation.

Experiment-driven works for the WPT with multiple nodes are very scarce.
For example, a multiple device charging WPT system is proposed and implemented in \cite{Waters:2015}, but the WPT in consideration is not the RF energy transfer but the near-field WPT via magnetic resonance.
In \cite{Mishra:20152}, the multi-hop and multi-path energy transfer is experimentally studied, but multiple device charging with a multi-antenna transmitter is not considered.
Therefore, the experimental research on the multi-node multi-antenna WPSN is required for demonstrating the viability and usefulness of such type of WPSNs.

In this paper, we first investigate the achievable receive power region of multiple sensor nodes when the linear and circular antenna arrays are used at the power beacon.
We present the actual receive power region experimentally derived in the testbed.
To achieve one point out of the receive power region, we can use time-sharing and beam-splitting beamforming techniques.
In the time-sharing beamforming technique, the power beacon concentrates an energy beam to a single node at a time, and multiple nodes are charged in a time-sharing manner.
On the other hand, the beam-splitting beamforming technique splits energy beams towards multiple nodes for charging multiple nodes at the same time.
In this paper, we propose the algorithms for realizing the time-sharing and beam-splitting beamforming techniques.
As shown in our previous work \cite{Choi:2015}, the Pareto frontier of the receive power region can be achieved by the beam-splitting beamforming technique.
We show how much gain can be achieved by using the beam-splitting beamforming technique over the time-sharing beamforming technique.

In order to realize a perpetual operation of a sensor node with ambient energy harvesting, many researchers have applied a duty cycle algorithm that autonomously adjusts wake-up and sleep cycles for controlling the power consumption.
For example, in \cite{Hsu:2006Oct}, the duty cycle algorithm takes into account the expected energy availability based on the history.
In this paper, we implement the duty cycle concept in the form of an awake frame ratio of each sensor node, and the energy neutral control algorithm controls the awake frame ratios based on stored energy levels of all sensor nodes.
Differently from the sensor networks with ambient energy harvesting, the WPSN has an intentional power source that provides RF energy by beamforming.
Therefore, the energy neutral control algorithm in the WPSN should be jointly designed with the beamforming algorithm.

We use the Lyapunov optimization technique \cite{Neely:2006} to design the optimal joint beam-splitting and energy neutral control algorithm.
We first formulate an optimization problem in which the sum of the utilities of all sensor nodes are maximized under the condition that the energy neutral operation is guaranteed.
Here, the utility is defined as a concave function of the awake frame ratio.
For dynamically solving this optimization problem, we define a quadratic Lyapunov function and design the algorithm that decides both beamforming weights and awake frame ratios minimizing the drift-plus-penalty function in each time frame.
We have conducted extensive experiments to test the performance of the proposed algorithm on the testbed.
The experimental results show that the proposed algorithm can achieve the optimality while maintaining the stored energies of all sensor nodes by adaptively controlling the energy beams and awake frame ratios.

The rest of paper is organized as follows.
We present the system model of the multi-node multi-antenna WPSN in Section \ref{section:model}.
We investigate the time-sharing and beam-splitting beamforming techniques in Section \ref{section:beamsplit}.
The joint beam-splitting and energy neutral control method is proposed in Section \ref{section:jointmethod}.
Section \ref{section:result} presents the experimental results, and Section \ref{section:conclusion} concludes the paper.

\section{System Model}\label{section:model}

\subsection{Multi-Node Multi-Antenna Wireless-Powered Sensor Network Model}

We consider a wireless-powered sensor network (WPSN) with one power beacon and multiple sensor nodes as described in Fig.~\ref{fig:model}.
The power beacon is connected to a power grid, and wirelessly supplies electrical energy to all the sensor nodes in the WPSN.
The power beacon is equipped with $N$ transmit antennas for RF energy transfer to adaptively focus energy beam towards the sensor nodes.
Henceforth, the $n$th transmit antenna will be called antenna $n$ ($=1,\ldots,N$).
There are $K$ sensor nodes, each of which has one receive antenna for harvesting the RF energy from the power beacon.
The $k$th sensor node will be called node $k$ ($=1,\ldots,K$).
A sensor node relies only on the power supplied by the power beacon without any other power source.

The detailed model of the power beacon and the sensor nodes in this paper is very similar to those for the single node WPSN model in our previous work \cite{Choi:2017}.
Therefore, we will briefly explain those models in this paper, and \cite{Choi:2017} can be referred to for more detailed model description.

In the power beacon, there are $N$ RF chains consisting of a phase shifter, a variable gain amplifier, and a transmit antenna.
The controller at the power beacon is able to decide the phase and magnitude of the RF energy signal transmitted from each transmit antenna by controlling the phase shifter and variable gain amplifier.
The RF transceiver at the power beacon is used to communicate with the sensor nodes.
A low-power communication technology is used for the RF transceiver (e.g., IEEE 802.15.4).
We assume that the frequency band for the communication is different from that for the RF energy transfer so that the RF energy transfer does not interfere with the communication.

A sensor node receives the RF energy signal from the power beacon via the receive antenna.
The rectifier converts the received RF energy signal to the DC power.
The harvested DC power is consumed by the active circuit components in the sensor node, and the remaining DC power is stored in the energy storage for future use.
The active circuit components drain energy from the energy storage in the case that the power consumption is higher than the harvested power.
The active circuit components include the micro controller unit (MCU) and the RF transceiver.
The MCU is a processor that controls the sensor node, and the RF transceiver is used for the communication with the power beacon.
A sensor node is capable of measuring the receive power of the RF energy signal and the amount of the stored energy in the energy storage.

\subsection{RF Energy Transfer Model}

We now set up the mathematical model of the energy beamforming for the RF energy transfer.
Let $w_n$ denote the beamforming weight of antenna $n$ in the power beacon, and ${\bold w} = (w_1,\ldots,w_N)^T$ denote the beamforming weight vector.
Then, the transmit power of antenna $n$ is $|w_n|^2$, and the total transmit power of the power beacon is ${\bold w}^H {\bold w}$.
We impose the following two types of the constraints on the transmit power.
The first transmit power constraint is a per-antenna power constraint such that
\begin{align}\label{eq:perantpow}
|w_n|^2\le P_\text{ant}, \text{ for all }n=1,\ldots,N,
\end{align}
where $P_\text{ant}$ is the maximum per-antenna transmit power.
This per-antenna power constraint is imposed because of the power limit of the amplifier in each RF chain.
The second transmit power constraint is a total power constraint such that
\begin{align}\label{eq:totpow}
{\bold w}^H {\bold w}\le P_\text{tot}.
\end{align}
where $P_\text{tot}$ is the maximum total transmit power.
The total power can be limited to meet the radio regulation on the radiated microwave power or to minimize the interference to other microwave devices.
Note that only the per-antenna power constraints are in effect if $P_\text{tot} \ge NP_\text{ant}$, and only the total power constraint is in effect if $P_\text{tot} \le P_\text{ant}$.

The channel gain from antenna $n$ to the receive antenna of node $k$ is denoted by $h_{k,n}$, and the channel gain vector of node $k$ is given by ${\bold h}_k=(h_{k,1},\ldots,h_{k,N})^T$.
In addition, the channel gain matrix is defined as ${\bold H}=({\bold h}_1,\ldots,{\bold h}_K)^T$.
The receive RF energy signal at node $k$ is
\begin{align}\label{eq:rxsig}
y_k = {\bold h}_k^T {\bold w}.
\end{align}
From \eqref{eq:rxsig}, the receive power at node $k$ is calculated as
\begin{align}\label{eq:rxpow}
\begin{split}
r_k &= |y_k|^2 = |{\bold h}_k^T {\bold w}|^2 = \tr({\bold w}^H{\bold h}_k^*{\bold h}_k^T {\bold w})\\
&= \tr({\bold h}_k^*{\bold h}_k^T{\bold w}{\bold w}^H)=\tr({\bold G}_k {\bold S}),
\end{split}
\end{align}
where ${\bold G}_k = {\bold h}_k^*{\bold h}_k^T$ and ${\bold S} = {\bold w}{\bold w}^H$.
Note that ${\bold G}_k$ and ${\bold S}$ are positive semidefinite matrices with rank one.
The receive power vector is defined as ${\bold r} = (r_1,\ldots,r_K)^T$.

\subsection{Channel Model for Linear and Circular Antenna Arrays}

In this subsection, we investigate the channel models when the power beacon is equipped with the linear and circular antenna arrays \cite{Balanis:2005}.
We use a spherical coordinate system to represent the location of nodes and antennas.
Let ${\boldsymbol \psi}^\text{no}_k = (d^\text{no}_k, \theta^\text{no}_k, \phi^\text{no}_k)$ denote a vector representing the location of node $k$ in a spherical coordinate system.
The radius, the elevation, and the azimuth of node $k$ are denoted by $d^\text{no}_k$, $\theta^\text{no}_k$, and $\phi^\text{no}_k$, respectively.
The reference point of the antenna array is located at the center of the spherical coordinate system.
Therefore, $d^\text{no}_k$ is equal to the distance from the reference point of the antenna array to node $k$.
Let ${\boldsymbol \psi}^\text{ant}_n = (d^\text{ant}_n, \theta^\text{ant}_n, \phi^\text{ant}_n)$ denote the vector representing the location of antenna $n$ in the spherical coordinate system.
In the case of the linear antenna array, the location of antenna $n$ is given as follows when $\zeta$ denotes an antenna spacing.
If $n<(N+1)/2$, we have $d^\text{ant}_n = ((N-1)/2-(n-1))\zeta$, $\theta^\text{ant}_n = \pi /2$, and $\phi^\text{ant}_n = \pi$.
If $n\ge (N+1)/2$, we have $d^\text{ant}_n = (n-1-(N-1)/2)\zeta$, $\theta^\text{ant}_n = \pi /2$, and $\phi^\text{ant}_n = 0$.
On the other hand, in the case of the circular antenna array, we have $d^\text{ant}_n = \xi$, $\theta^\text{ant}_n = \pi /2$, and $\phi^\text{ant}_n = (2\pi/N) n$, where $\xi$ is the radius of the circular antenna array.

From Friis equation and the antenna array equation, the receive power at node $k$ is given as
\begin{align}\label{eq:arrayrxpow}
r_k = \bigg(\frac{\lambda}{4\pi d^\text{no}_k}\bigg)^2 g_t(\theta^\text{no}_k,\phi^\text{no}_k) g_r |f(\theta^\text{no}_k,\phi^\text{no}_k)|^2,
\end{align}
where $\lambda$ is the wavelength, $g_t(\theta,\phi)$ is the antenna gain of a single transmit antenna element towards elevation $\theta$ and azimuth $\phi$, and $g_r$ is the antenna gain of the receive antenna at a sensor node.
We assume that the maximum gain direction of the receive antenna at each sensor node faces towards the antenna array, and hence the antenna gain of the receive antenna (i.e., $g_r$) is a constant.
In \eqref{eq:arrayrxpow}, $f(\theta,\phi)$ is an array factor towards elevation $\theta$ and azimuth $\phi$, which is given by
\begin{align}\label{eq:arrayfactor}
f(\theta,\phi) = \sum_{n=1}^N w_n\cdot \exp\big(j(2\pi/\lambda)\big\langle{\boldsymbol \psi}^\text{ant}_n,\widehat{\boldsymbol \psi}\big\rangle\big),
\end{align}
where $\widehat{\boldsymbol \psi} = (1,\theta,\phi)$ is a vector with a unit radius in a spherical coordinate system, `$\langle\cdot,\cdot\rangle$' is the inner product, and $j$ is an imaginary unit.

From \eqref{eq:rxpow}, \eqref{eq:arrayrxpow}, and \eqref{eq:arrayfactor}, the channel gain $h_{k,n}$ is
\begin{align}
\begin{split}
h_{k,n}& = \frac{\lambda}{4\pi d^\text{no}_k} \sqrt{g_t(\theta^\text{no}_k,\phi^\text{no}_k) g_r} \exp\big(j(2\pi/\lambda)\big\langle{\boldsymbol \psi}^\text{ant}_n,\widehat{\boldsymbol \psi}^\text{no}_k\big\rangle\big)\\
& = \frac{\lambda}{4\pi d^\text{no}_k} \sqrt{g_t(\theta^\text{no}_k,\phi^\text{no}_k) g_r} \\
&\qquad\times \exp(j(2\pi/\lambda) d^\text{ant}_n \sin(\theta^\text{no}_k)\cos(\phi^\text{no}_k-\phi^\text{ant}_n)),
\end{split}
\end{align}
where $\widehat{\boldsymbol \psi}^\text{no}_k = (1,\theta^\text{no}_k,\phi^\text{no}_k)$ and the elevation of antenna $n$ is $\theta^\text{ant}_n = \pi/2$ for $n=1,\ldots,N$.
When the linear antenna array is used, the channel gain is given as
\begin{align}
\begin{split}
&h_{k,n} = \frac{\lambda}{4\pi d^\text{no}_k} \sqrt{g_t(\theta^\text{no}_k,\phi^\text{no}_k) g_r} \\
&\quad \times \exp(j(2\pi/\lambda) (n-1-(N-1)/2)\zeta \sin(\theta^\text{no}_k)\cos(\phi^\text{no}_k)).
\end{split}
\end{align}
The channel gain in the case of the circular antenna array is
\begin{align}
\begin{split}
h_{k,n} = &\frac{\lambda}{4\pi d^\text{no}_k} \sqrt{g_t(\theta^\text{no}_k,\phi^\text{no}_k) g_r} \\
&\times \exp(j(2\pi/\lambda) \xi \sin(\theta^\text{no}_k)\cos(\phi^\text{no}_k-(2\pi/N) n)).
\end{split}
\end{align}

\subsection{Sensor Node Power Model}

In this subsection, we present the sensor node power model that describes how each sensor node harvests, consumes, and stores energy.
Since this sensor node power model is almost identical to that in our previous work \cite{Choi:2017}, only a brief explanation will be given here.

As in \cite{Choi:2017}, we consider the sensor node circuit model where the energy harvesting part, the energy storage part, and the energy consuming part are connected in parallel.
We assume that the energy storage part is a supercapacitor.
Then, the stored energy in the supercapacitor at node $k$ is $E_k = CV_k^2/2$, where $V_k$ is the voltage of the supercapacitor at node $k$ and $C$ is the capacitance of the supercapacitor.
Since all parts are connected in parallel, the voltages across the energy harvesting and energy consuming parts are the same as the supercapacitor voltage $V_k$.

First, we explain the energy harvesting part.
The energy harvesting part includes a rectifier that converts the received RF energy signal to the DC power.
Let $\rho_k$ denote the harvested power that is the amount of the rectified DC power at node $k$.
Due to the imperfection in the rectifier circuit, some amount of the power is lost in the course of the RF-to-DC conversion.
The energy harvesting efficiency is defined as the ratio of the harvested power to the receive power.
Although the energy harvesting efficiency is typically a function of the receive power \cite{Choi:2017}, we assume that the energy harvesting efficiency is a constant for simplicity.
This assumption is valid if the energy harvesting circuit is designed in such a way that the energy harvesting efficiency is flat over the receive power of interest.
Then, the harvested power at node $k$ is
\begin{align}\label{eq:rhok}
\rho_k = \eta\cdot r_k,
\end{align}
where $\eta$ is the energy harvesting efficiency.

Second, we explain the energy consuming part that consists of the MCU and RF transceiver.
To save the energy, a sensor node can be put into one of the following four modes: idle, active, receive, and transmit modes.
The set of all possible modes is ${\mathcal M}=\{\text{idle},\text{act},\text{rx},\text{tx}\}$, and each mode is indexed by $m\in {\mathcal M}$.
In the idle mode, both the MCU and RF transceiver are inactive, and very small power is consumed.
In the active mode, only the MCU is activated and the RF transceiver is inactive.
The RF transceiver is ready to receive data from the power beacon in the receive mode, and the RF transceiver sends data to the power beacon in the transmit mode.
Let $m_k$ denote the mode of node $k$.
Let $\delta(m,E)$ denote the consumed power when a sensor node is in mode $m$ and the stored energy is $E$.
According to \cite{Choi:2017}, the consumed power is given by
\begin{align}\label{eq:deltame}
\delta(m,E) = \frac{2}{C\cdot\zeta(m)} E + \sqrt{\frac{2}{C}}\cdot\xi(m)\sqrt{E},
\end{align}
where $\zeta(m)$ is the resistance of the constant resistance load in mode $m$ and $\xi(m)$ is the current of the constant current load in mode $m$.

Now, we explain the energy storage part.
The stored energy changes over time $t$ depending on the harvested and consumed power according to the following equation.
\begin{align}\label{eq:contevol}
\frac{\mathrm{d}E_k}{\mathrm{d}t} = \eta\cdot r_k - \delta(m_k,E_k) - \delta_\text{leak}(E_k).
\end{align}
The stored energy in the supercapacitor slowly leaks by itself.
In \eqref{eq:contevol}, $\delta_\text{leak}(E) = 2E/(CR_\text{leak})$ is the leakage power from the supercapacitor, where $R_\text{leak}$ is the leakage resistance,
Since the capacity of the energy storage is limited, the stored energy cannot exceed the maximum stored energy $E_\text{max}$.
The sensor node is blacked out if the stored energy goes below the minimum stored energy $E_\text{min}$.
Therefore, the stored energy should be kept above the minimum stored energy (i.e., $E_k\ge E_\text{min}$) for continuous operation.

\subsection{Testbed Implementation}

We have implemented a multi-node multi-antenna WPSN testbed with one power beacon and multiple sensor nodes as shown in Fig.~\ref{fig:testbedpic}.
This testbed is a multi-node extension of that in our previous work \cite{Choi:2017}.

\begin{figure}
    \centering
    \subfigure[Antenna array]{
        \label{fig:testbedpic1}\includegraphics[width=1.9cm, bb=0in 0in 5.3in 10in] {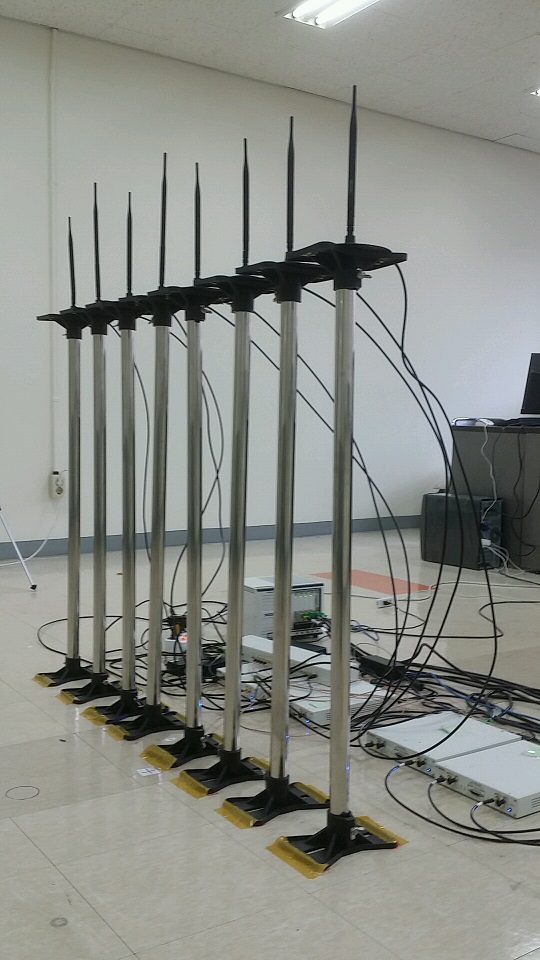}
        }~
    \subfigure[Testbed measurement setting]{
        \label{fig:testbedpic2}\includegraphics[width=7.1cm, bb=0in 0in 25in 12.5in] {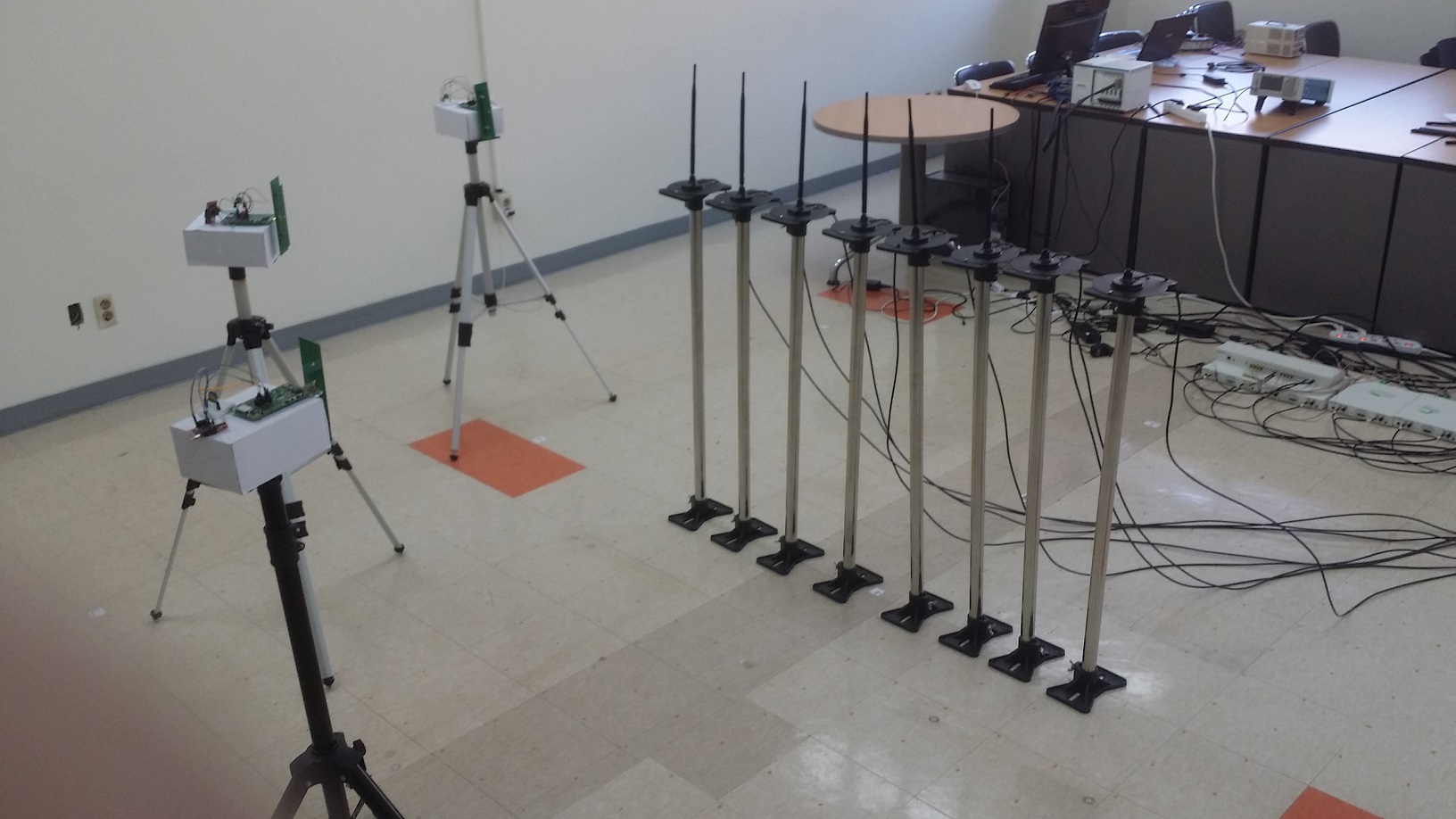}
        }
    \caption{Testbed}
    \label{fig:testbedpic}
\end{figure}

The power beacon consists of Universal Software Radio Peripherals (USRPs) and eight dipole antennas.
The laptop computer, which runs a LabVIEW software, generates eight streams of complex beamforming weights, and delivers them to the USRPs.
The USRPs convert the beamforming weights to the RF energy signal with the frequency of 920 MHz.
The maximum per-antenna transmit power is 140 mW.
This RF signal is wirelessly transmitted to the sensor nodes by means of the dipole antennas arranged in a linear or a circular antenna array.
We use a clock distributor, i.e. OctoClock, for the purpose of time and frequency synchronization between multiple USRPs.

Each sensor node consists of an energy harvesting board, a sensor board, and an energy storage device that are connected in parallel.
The received RF energy signal is converted to DC power by the rectifier inside the energy harvesting board (i.e., Powercast P1110 evaluation board).
The rectified DC power is used to charge the energy storage device (i.e., Samxon DDL series supercapacitor).
The sensor board is a Zolertia Z1 mote with a Contiki operating system as a software platform.
The active components of a Zolertia Z1 mote include the MCU (i.e., TI MSP430) and the RF transceiver (i.e., CC2420).
The CC2420 is an IEEE 802.15.4-compliant RF transceiver that operates on the 2.4 GHz frequency band.

\begin{figure}
	\centering
    \includegraphics[width=8cm, bb=2.2in 8.1in 7.3in 11.5in]{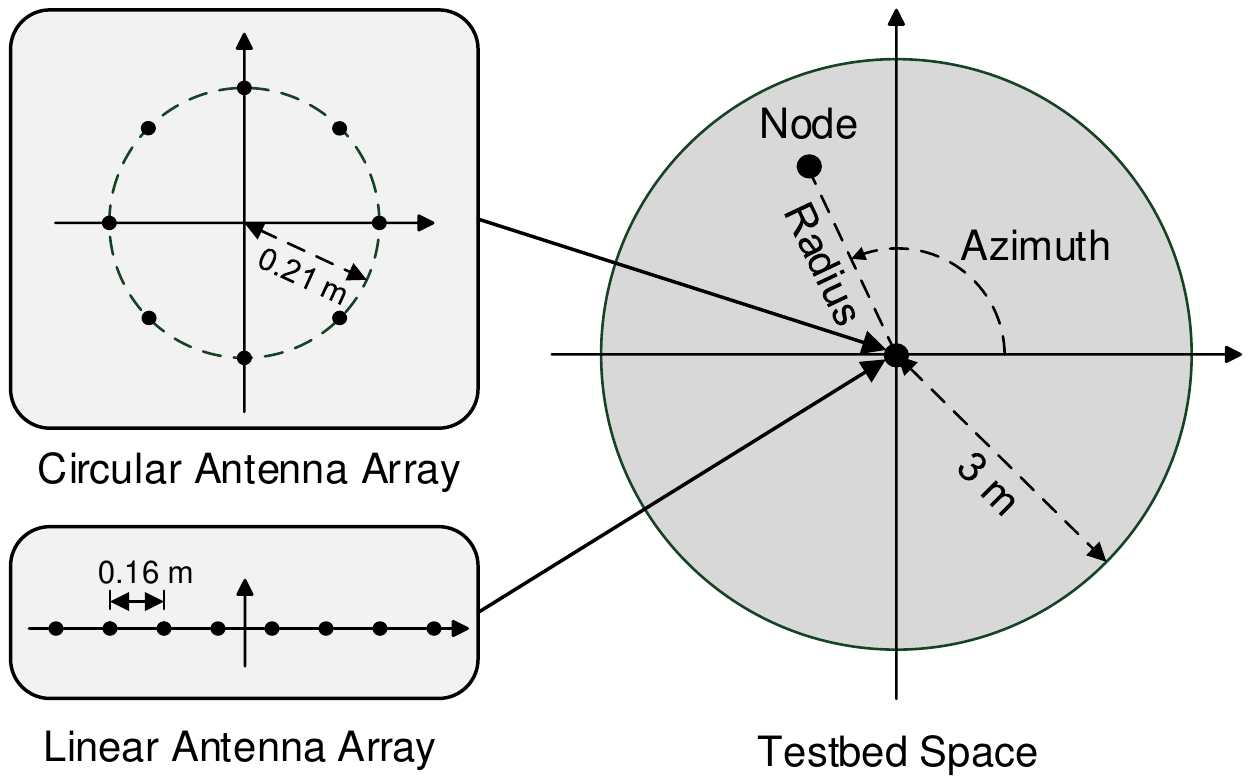}
    \caption{Testbed configuration.}
    \label{fig:testbed}
\end{figure}

The spatial dimension of the testbed space is a circular area with a 3 m radius as shown in Fig.~\ref{fig:testbed}.
The reference point of the antenna array of the power beacon is centered in this testbed space.
We use eight antenna elements in the antenna array.
If the linear antenna array is used, the antenna spacing is 0.16 m, which is the half-wavelength at the frequency of 920 MHz.
If the circular antenna array is used, the antenna elements are placed on a circle with the radius of 0.21 m.
Each sensor node is placed on a horizontal plane, and therefore the elevation of node $k$ is fixed to $\theta^\text{no}_k=\pi/2$.
Thus, the spherical coordinate of node $k$ is represented by $(d^\text{no}_k,\phi^\text{no}_k)$, where $d^\text{no}_k$ is the radius and $\phi^\text{no}_k$ is the azimuth of node $k$.

\section{Beam-Splitting and Time-Sharing Beamforming Techniques for Simultaneously Charging Multiple Sensor Nodes}\label{section:beamsplit}

\subsection{Receive Power Region and Beamforming Techniques}

When the channel gains are given, we can obtain various receive power vectors by controlling the beamforming weight vector.
Let us define the instantaneous receive power region as the set of all possible receive power vectors obtained by varying the beamforming weight vector under the transmit power constraints.
Then, the instantaneous receive power region is
\begin{align}\label{eq:instrxpowregion}
\begin{split}
{\mathcal R} = &\{{\bold r}=(r_1,\ldots,r_K)^T|\\
&\qquad r_k = |{\bold h}_k^T {\bold w}|^2\text{ for all $k=1,\ldots,K$},\\
&\qquad |w_n|^2\le P_\text{ant} \text{ for all }n=1,\ldots,N,\\
&\qquad {\bold w}^H {\bold w}\le P_\text{tot}\}.
\end{split}
\end{align}
If the power beacon alternates between different beamforming weight vectors, we can average the receive power vectors obtained by those beamforming weight vectors.
Such average receive power vectors constitute the average receive power region.
The average receive power region is the convex hull of the receive power region such that
\begin{align}\label{eq:avgrpr}
\overline{\mathcal R} = \convhull({\mathcal R}),
\end{align}
where $\convhull({\mathcal X})$ denotes the convex hull of ${\mathcal X}$.

The beamforming technique picks up one operating point out of the average receive power region $\overline{\mathcal R}$ so that the RF power can be properly distributed among sensor nodes.
In this paper, we consider two types of beamforming techniques, one is the time-sharing beamforming technique in Section \ref{section:timesharing} and the other is the beam-splitting beamforming technique in Section \ref{section:beamsplitting}.

\subsection{Time-Sharing Beamforming Technique}\label{section:timesharing}

The time-sharing beamforming technique generates a single energy beam, which is focused only on a single sensor node, at a time.
The power beacon alternately uses these single microwave beams over time in order to to supply energy to multiple sensor nodes.
Let ${\bold w}^{\text{TS},k}=(w^{\text{TS},k}_1,\ldots,w^{\text{TS},k}_N)^T$ denote the beamforming weight vector that focuses a single energy beam towards node $k$.

Since the beamforming weight vector ${\bold w}^{\text{TS},k}$ maximizes the receive power of node $k$, we can calculate ${\bold w}^{\text{TS},k}$ by solving the following optimization problem:
\begin{align}
&\text{maximize}\qquad |{\bold h}_k^T {\bold w}|^2\label{eq:optts1}\\
&\text{subject to}\qquad |w_n|^2\le P_\text{ant} \text{ for all }n=1,\ldots,N,\label{eq:optts2}\\
&\qquad\qquad\qquad{\bold w}^H {\bold w}\le P_\text{tot}.\label{eq:optts3}
\end{align}

The solution to the optimization problem in \eqref{eq:optts1}--\eqref{eq:optts3} is the phase conjugate of the channel gain vector.
A channel gain is represented by its magnitude and phase as $h_{k,n} = |h_{k,n}|\exp(j\angle h_{k,n})$.
Let $x_n$ and $\omega_n$ denote the magnitude and the phase of the beamforming weight $w_n$, respectively.
Then, we have $w_n = x_n \exp(j\omega_n)$.
We can rewrite the optimization target \eqref{eq:optts1} as $|\sum_{n=1}^N |h_{k,n}|x_n\exp(j(\angle h_{k,n}+\omega_n))|^2$, the per-antenna power constraint \eqref{eq:optts2} as $x_n^2\le P_\text{ant}$ for all $n=1,\ldots,N$, and the total power constraint \eqref{eq:optts3} as $\sum_{n=1}^N x_n^2\le P_\text{tot}$.
Since $\omega_n$ is not included in the constraints, the optimization target is maximized when $\omega_n = -\angle h_{k,n}$ for all $n=1,\ldots,N$.

Now, the optimization problem in \eqref{eq:optts1}--\eqref{eq:optts3} is reformulated as
\begin{align}
&\text{maximize}\qquad \mbox{$\sum_{n=1}^N$} |h_{k,n}|x_n\label{eq:optts4}\\
&\text{subject to}\qquad x_n^2\le P_\text{ant} \text{ for all }n=1,\ldots,N,\label{eq:optts5}\\
&\qquad\qquad\qquad\mbox{$\sum_{n=1}^N$} x_n^2\le P_\text{tot}.\label{eq:optts6}
\end{align}
The optimization problem in \eqref{eq:optts4}--\eqref{eq:optts6} is actually a water-filling problem.
Let us release the total power constraint \eqref{eq:optts6} by the Lagrange multiplier $\lambda\ge 0$.
Then, the optimization target \eqref{eq:optts4} becomes $\sum_{n=1}^N |h_{k,n}|x_n + \lambda (P_\text{tot} - \sum_{n=1}^N x_n^2)$.
The maximizer of this optimization target with the per-antenna power constraint \eqref{eq:optts5} is
\begin{align}
x_n(\lambda) = \min\big\{|h_{k,n}|/(2\lambda),\sqrt{P_\text{ant}}\big\},
\end{align}
for $n=1,\ldots,N$.
Let $\lambda^*$ denote the dual optimal $\lambda$, and let $P(\lambda)$ denote the total power given $\lambda$ such that $P(\lambda)=\sum_{n=1}^N x_n(\lambda)^2$.
For finding $\lambda^*$, we can increase $\lambda$ until the total power constraint is met (i.e., $P(\lambda) \le P_\text{tot}$).
Let us sort $|h_{k,n}|$ for $n=1,\ldots,N$ in an increasing order, and let $n(i)$ denote the index of the antenna with the $i$th smallest $|h_{k,n}|$.
Let us define $\lambda_i = |h_{k,n(i)}|/(2\sqrt{P_\text{ant}})$ for $i=1,\ldots,N$ and $\lambda_0 = 0$.
Then, we have $P(\lambda_i) = \sum_{l=1}^{i-1} |h_{k,n(l)}|^2/(2\lambda_i)^2 + (N-i+1)P_\text{ant}$.
Since $P(\lambda_i)$ is decreasing with $i$, we can define $i^*$ as the smallest $i$ that satisfies $P(\lambda_i)\le P_\text{tot}$ over $i=1,\ldots,N$.
If there is no such $i$, we have $i^* = N+1$.
Then, we have $\lambda_{i^*-1} < \lambda^* \le \lambda_{i^*}$ and $P(\lambda) = \sum_{l=1}^{i^*-1} |h_{k,n(l)}|^2/(2\lambda)^2 + (N-i^*+1)P_\text{ant}$ for $\lambda_{i^*-1} < \lambda \le \lambda_{i^*}$.
Since $P(\lambda^*)=P_\text{tot}$, we can calculate $\lambda^*$ as
\begin{align}
\lambda^* = \frac{1}{2}\sqrt{\frac{\sum_{l=1}^{i^*-1} |h_{k,n(l)}|^2}{P_\text{tot}-(N-i^*+1)P_\text{ant}}}.
\end{align}

Finally, the solution to the optimization problem in \eqref{eq:optts1}--\eqref{eq:optts3} is calculated as
\begin{align}
w_n^{\text{TS},k} = x_n(\lambda^*) \exp(-j\angle h_{k,n}),
\end{align}
for all $n=1,\ldots,N$.
This solution can be simplified below, in the special case that only the per-antenna power constraints are in effect (i.e., $P_\text{tot} \ge NP_\text{ant}$).
\begin{align}
w_n^{\text{TS},k} = \frac{h_{k,n}^*}{|h_{k,n}|} \sqrt{P_\text{ant}}.
\end{align}
On the other hand, if only the total power constraint is in effect (i.e., $P_\text{tot} \le P_\text{ant}$), we can simplify the optimal solution to
\begin{align}
w_n^{\text{TS},k} = \frac{h_{k,n}^*}{\|{\bold h}_k\|_2}\sqrt{P_\text{tot}}.
\end{align}

Let $r^{\text{TS},i}_k$ denote the receive power at node $k$ when ${\bold w}^{\text{TS},i}$ is used.
Then, we have
\begin{align}\label{eq:rxts}
r^{\text{TS},i}_k=|{\bold h}_k^T {\bold w}^{\text{TS},i}|^2.
\end{align}
We also define the receive power vector when ${\bold w}^{\text{TS},i}$ is used as ${\bold r}^{\text{TS},i}=(r^{\text{TS},i}_1,\ldots,r^{\text{TS},i}_K)^T$.

In the time-sharing beamforming technique, the power beacon alternates between the beamforming weight vectors ${\bold w}^{\text{TS},i}$ for $i=1,\ldots,K$ in a predefined time-sharing proportion.
Let $\tau_i$ denote the time-sharing proportion during which ${\bold w}^{\text{TS},i}$ is used in the power beacon.
It is satisfied that $\sum_{i=1}^K \tau_i = 1$ and $\tau_i \ge 0$ for all $i=1,\ldots,K$.
Then, the average receive power vector of the time-sharing beamforming technique is given by
\begin{align}\label{eq:rtsavg}
{\bold r}^\text{TS,avg} = \mbox{$\sum_{i=1}^K$} \tau_i {\bold r}^{\text{TS},i}.
\end{align}

\subsection{Beam-Splitting Beamforming Technique}\label{section:beamsplitting}

The beam-splitting beamforming technique splits the energy beam towards more than one nodes at the same time to charge multiple sensor nodes.
The beam-splitting beamforming technique makes use of the beamforming weight vectors that achieve the Pareto optimal points in the average receive power region $\overline{\mathcal R}$ in \eqref{eq:avgrpr}.
The Pareto frontier, denoted by $\overline{\mathcal R}^\text{PF}$, is defined as the set of all Pareto optimal points in $\overline{\mathcal R}$.
That is,
\begin{align}
\overline{\mathcal R}^\text{PF} = \{{\bold x}\in \overline{\mathcal R}\ |\ \text{there is no ${\bold r}\in \overline{\mathcal R}$ such that ${\bold x}\prec {\bold r}$}\},
\end{align}
where '$\prec$' is an element-wise inequality.

Since $\overline{\mathcal R}$ is a convex set, the elements in the Pareto frontier of $\overline{\mathcal R}$ can be obtained by finding the maximum weighted sum of the components of a receive power vector in $\overline{\mathcal R}$.
The optimization problem to find such receive power vectors is
\begin{align}
&\text{maximize}\qquad {\boldsymbol \alpha}^T {\bold x}\label{eq:optweight1}\\
&\text{subject to}\qquad {\bold x}\in \overline{\mathcal R},\label{eq:optweight2}
\end{align}
where ${\boldsymbol \alpha} = (\alpha_1,\ldots,\alpha_K)^T$ is a receive power weight vector.
The receive power weight vector ${\boldsymbol \alpha}$ should satisfy that $\sum_{k=1}^K \alpha_k \le 1$ and $\alpha_k \ge 0$ for all $k=1,\ldots,K$.

Recall that ${\bold S}$ is a positive semidefinite matrix with rank one, which is defined as ${\bold S}={\bold w}{\bold w}^H$.
Then, we can rewrite the receive power at node $k$ as $\tr({\bold G}_k {\bold S})$, the total transmit power as $\tr({\bold S})$, and the transmit power of antenna $n$ as $\tr({\bold B}_n{\bold S})$.
Here, we define an $N\times N$ matrix ${\bold B}_n$, of which only the $(n,n)$th element is one and all other elements are zero.
Then, the optimization problem in \eqref{eq:optweight1} and \eqref{eq:optweight2} is equivalent to the following optimization problem:
\begin{align}
&\text{maximize}\qquad \mbox{$\sum_{k=1}^K$} \alpha_k \tr({\bold G}_k {\bold S})\label{eq:optmulti1}\\
&\text{subject to}\qquad \tr({\bold B}_n{\bold S}) \le P_\text{ant},\text{ for $n=1,\ldots,N$},\label{eq:optmulti2}\\
&\qquad\qquad\qquad \tr({\bold S})\le P_\text{tot},\label{eq:optmulti3}\\
&\qquad\qquad\qquad {\bold S}\succeq {\bold 0},\label{eq:optmulti4}\\
&\qquad\qquad\qquad \rank({\bold S})=1\label{eq:optmulti5},
\end{align}
where ${\bold S}\succeq {\bold 0}$ means that ${\bold S}$ is a positive semidefinite matrix and $\rank({\bold S})$ is the rank of ${\bold S}$.
Let ${\bold S}^\text{BS}({\boldsymbol \alpha}) = {\bold w}^\text{BS}({\boldsymbol \alpha}){\bold w}^\text{BS}({\boldsymbol \alpha})^H$ denote the optimal solution of \eqref{eq:optmulti1}--\eqref{eq:optmulti5}.
The receive power of node $k$ with the optimal beamforming weight vector ${\bold w}^\text{BS}({\boldsymbol \alpha})$ is given by
\begin{align}
r_k^\text{BS}({\boldsymbol \alpha})=|{\bold h}_k^T {\bold w}^\text{BS}({\boldsymbol \alpha})|^2.
\end{align}
The receive power vector with ${\bold w}^\text{BS}({\boldsymbol \alpha})$ is given by ${\bold r}^\text{BS}({\boldsymbol \alpha}) = (r_1^\text{BS}({\boldsymbol \alpha}),\ldots,r_K^\text{BS}({\boldsymbol \alpha}))^T$.

The optimization problem \eqref{eq:optmulti1}--\eqref{eq:optmulti5} can be solved in a closed form when the per-antenna power constraints are not in effect (i.e., $P_\text{tot} \le P_\text{ant}$).
We first solve the optimization problem without the per-antenna power constraints in \eqref{eq:optmulti2}.
The objective function \eqref{eq:optmulti1} can be rewritten as $\tr({\bold V}({\boldsymbol \alpha}){\bold S})$, where
\begin{align}
{\bold V}({\boldsymbol \alpha}) = \mbox{$\sum_{k=1}^K$} \alpha_k {\bold G}_k.
\end{align}
The eigenvalue decomposition of ${\bold V}({\boldsymbol \alpha})$ is ${\bold V}({\boldsymbol \alpha}) = {\bold U}({\boldsymbol \alpha})^H {\bold Z}({\boldsymbol \alpha}) {\bold U}({\boldsymbol \alpha})$, where ${\bold U}({\boldsymbol \alpha})$ is a unitary matrix such that ${\bold U}({\boldsymbol \alpha}) = ({\bold u}_1({\boldsymbol \alpha}),\ldots,{\bold u}_N({\boldsymbol \alpha}))^T$ and ${\bold Z}({\boldsymbol \alpha})$ is a diagonal matrix such that ${\bold Z}({\boldsymbol \alpha}) = \diag(z_1({\boldsymbol \alpha}),\ldots,z_N({\boldsymbol \alpha}))$.
The diagonal elements in ${\bold Z}({\boldsymbol \alpha})$ is sorted in a descending order.
Therefore, $z_1({\boldsymbol \alpha})$ is the principal eigenvalue and ${\bold u}_1({\boldsymbol \alpha})$ is the principal eigenvector of ${\bold V}({\boldsymbol \alpha})$.
In \cite{Choi:2015}, it is shown that the optimal beamforming weight vector without the per-antenna power constraints is given by
\begin{align}\label{eq:wbs1}
{\bold w}^\text{BS}({\boldsymbol \alpha}) = \sqrt{P_\text{tot}}\cdot {\bold u}_1^*({\boldsymbol \alpha}).
\end{align}

Now, we consider a general optimization problem with both per-antenna and total power constraints.
The optimization problem in \eqref{eq:optmulti1}--\eqref{eq:optmulti5} is actually a quadratically constrained quadratic problem (QCQP).
The QCQP can be approximately solved by the semidefinite relaxation (SDR) \cite{Luo:2010}.
For the SDR, the rank-one constraint in \eqref{eq:optmulti5} is removed, and the resulting semidefinite programming (SDP) is solved by a convex optimization technique such as the interior-point method.
Then, the rank-one approximation of the solution from the interior-point method can be derived by calculating the principal eigenvector of the solution.
However, the complexity of the SDR can be high due to the interior-point method when many transmit antennas are used.

Since our target is to develop an algorithm which is deployed in a real-time testbed, we propose a low-complexity approximate method for solving the optimization problem \eqref{eq:optmulti1}--\eqref{eq:optmulti5} rather than using the SDR.
By using the eigenvalues and eigenvectors of ${\bold V}({\boldsymbol \alpha})$, we can rewrite the objective function \eqref{eq:optmulti1} as $\sum_{n=1}^N z_n({\boldsymbol \alpha})|{\bold u}_n({\boldsymbol \alpha})^T {\bold w}|^2$.
To simplify this objective function, we can only consider the term with the largest eigenvalue.
Then, the simplified objective function is $z_1({\boldsymbol \alpha})|{\bold u}_1({\boldsymbol \alpha})^T {\bold w}|^2$.
By using this objective function, we can formulate the optimization problem as
\begin{align}
&\text{maximize}\qquad |{\bold u}_1({\boldsymbol \alpha})^T {\bold w}|^2\label{eq:optbs1}\\
&\text{subject to}\qquad |w_n|^2\le P_\text{ant} \text{ for all }n=1,\ldots,N,\label{eq:optbs2}\\
&\qquad\qquad\qquad{\bold w}^H {\bold w}\le P_\text{tot}.\label{eq:optbs3}
\end{align}
This optimization problem is identical to the optimization problem in \eqref{eq:optts1}--\eqref{eq:optts3} for the time-sharing beamforming technique if we replace ${\bold u}_1({\boldsymbol \alpha})$ with ${\bold h}_k$.
Therefore, we can apply the same optimization method to solve \eqref{eq:optbs1}--\eqref{eq:optbs3} as the one to solve \eqref{eq:optts1}--\eqref{eq:optts3} in Section \ref{section:timesharing}.

\subsection{Comparison between Time-Sharing and Beam-Splitting Beamforming Techniques}

In this subsection, we compare the performance of the time-sharing and beam-splitting beamforming techniques.
We propose a metric called a beam-splitting gain to quantify the gain of using the beam-splitting beamforming technique over the time-sharing beamforming technique.
The beam-splitting gain is defined as the ratio of the maximum weighted sum of the components of the receive power vectors achieved by the beam-splitting and time-sharing beamforming techniques, that is,
\begin{align}
\Gamma = \frac{{\boldsymbol \beta}^T{\bold r}^\text{BS}({\boldsymbol \beta})}{\max_{i=1,\ldots,K}{\boldsymbol \beta}^T{\bold r}^{\text{TS},i}},
\end{align}
where ${\boldsymbol \beta} = (\beta_1,\ldots,\beta_K)^T$ is some receive power weight vector.

We choose ${\boldsymbol \beta}$ such that a hyperplane ${\boldsymbol \beta}^T {\bold x} = 1$ includes all ${\bold r}^{\text{TS},i}$ for $i=1,\ldots,K$ (i.e., ${\boldsymbol \beta}^T {\bold r}^{\text{TS},i} = 1$ for all $i=1,\ldots,K$).
Then, we can calculate ${\boldsymbol \beta}$ as
\begin{align}\label{eq:beta}
{\boldsymbol \beta} = ({\bold r}^{\text{TS}})^{-1}{\boldsymbol 1},
\end{align}
where ${\bold r}^{\text{TS}} = ({\bold r}^{\text{TS},1},\ldots,{\bold r}^{\text{TS},K})^T$ and ${\bold 1}$ is a vector of all ones.
For such ${\boldsymbol \beta}$ in \eqref{eq:beta}, the beam-splitting gain is reduced to
\begin{align}
\Gamma = {\boldsymbol \beta}^T{\bold r}^\text{BS}({\boldsymbol \beta}).
\end{align}
The beam-splitting gain is a function of the channel gain matrix ${\bold H}$, and it is equal to or higher than one.

\subsection{Experimental Results for Beamforming Techniques}

In this subsection, we present experimental results regarding the beamforming techniques.
Figs.~\ref{fig:pr2dcirc}--\ref{fig:pr3dlin} show the receive power region and the Pareto frontier, experimentally derived in our testbed.
The receive power region, which is the instantaneous one in \eqref{eq:instrxpowregion}, is depicted as a cloud of receive power vectors obtained by using random beamforming weights.
The Pareto frontier is obtained by the receive power vectors of the beam-splitting beamforming technique (i.e., ${\bold r}^\text{BS}({\boldsymbol \alpha})$).
To derive various points in the Pareto frontier, we use random ${\boldsymbol \alpha}$'s for the beam-splitting beamforming technique.
In addition, these figures also show the receive power vectors of the time-sharing beamforming technique (i.e., ${\bold r}^{\text{TS},i}$).

We have tested the beamforming techniques in various parameter and environment settings.
Two nodes are tested in Figs.~\ref{fig:pr2dcirc} and \ref{fig:pr2dlin} while three nodes are tested in Figs.~\ref{fig:pr3dcirc} and \ref{fig:pr3dlin}.
We use a circular antenna array in Figs.~\ref{fig:pr2dcirc} and \ref{fig:pr3dcirc} and a linear antenna array in Figs.~\ref{fig:pr2dlin} and \ref{fig:pr3dlin}.
We also vary the azimuth of each node (i.e., $\phi^\text{no}_k$) and the maximum total transmit power (i.e., $P_\text{tot}$).
In the caption below each figure, we specify the azimuth of each node and the maximum total transmit power in sequence.
The distance from the center of the antenna array is set to $d^\text{no}_k = 2$ m for all nodes.

In all Figs.~\ref{fig:pr2dcirc}--\ref{fig:pr3dlin}, we can clearly see that the Pareto frontier obtained by the proposed beam-splitting beamforming technique correctly forms the upper bound of the receive power region.
In addition, we can see that each time-sharing receive power vector maximizes the receive power of its corresponding node.
Therefore, these figures demonstrate that the proposed beamforming techniques work well in a real-life testbed with various antenna arrays, node locations, and transmit power constraints.

The shape of the receive power region greatly depends on the locations of the nodes.
As seen in all Figs.~\ref{fig:pr2dcirc}--\ref{fig:pr3dlin}, the receive powers of all nodes are highly correlated when the nodes are placed in the similar direction (e.g., 0$^\circ$ and 10$^\circ$ in the two node case).
We also observe that the beam-splitting beamforming technique is able to achieve better receive power vectors than the time-sharing beamforming technique does in some scenarios (e.g., Fig.~\ref{fig:pr2dcirc10deg2m0_5}, \ref{fig:pr2dcirc10deg2m1}, \ref{fig:pr2dlin10deg2m0_5}, and \ref{fig:pr2dlin10deg2m1}).

\begin{figure}
    \centering
    \subfigure[0$^\circ$, 10$^\circ$, 560 mW]{
        ~~\label{fig:pr2dcirc10deg2m0_5}\includegraphics[width=4.0cm, bb=1.5in 0.3in 9.9in 7.4in] {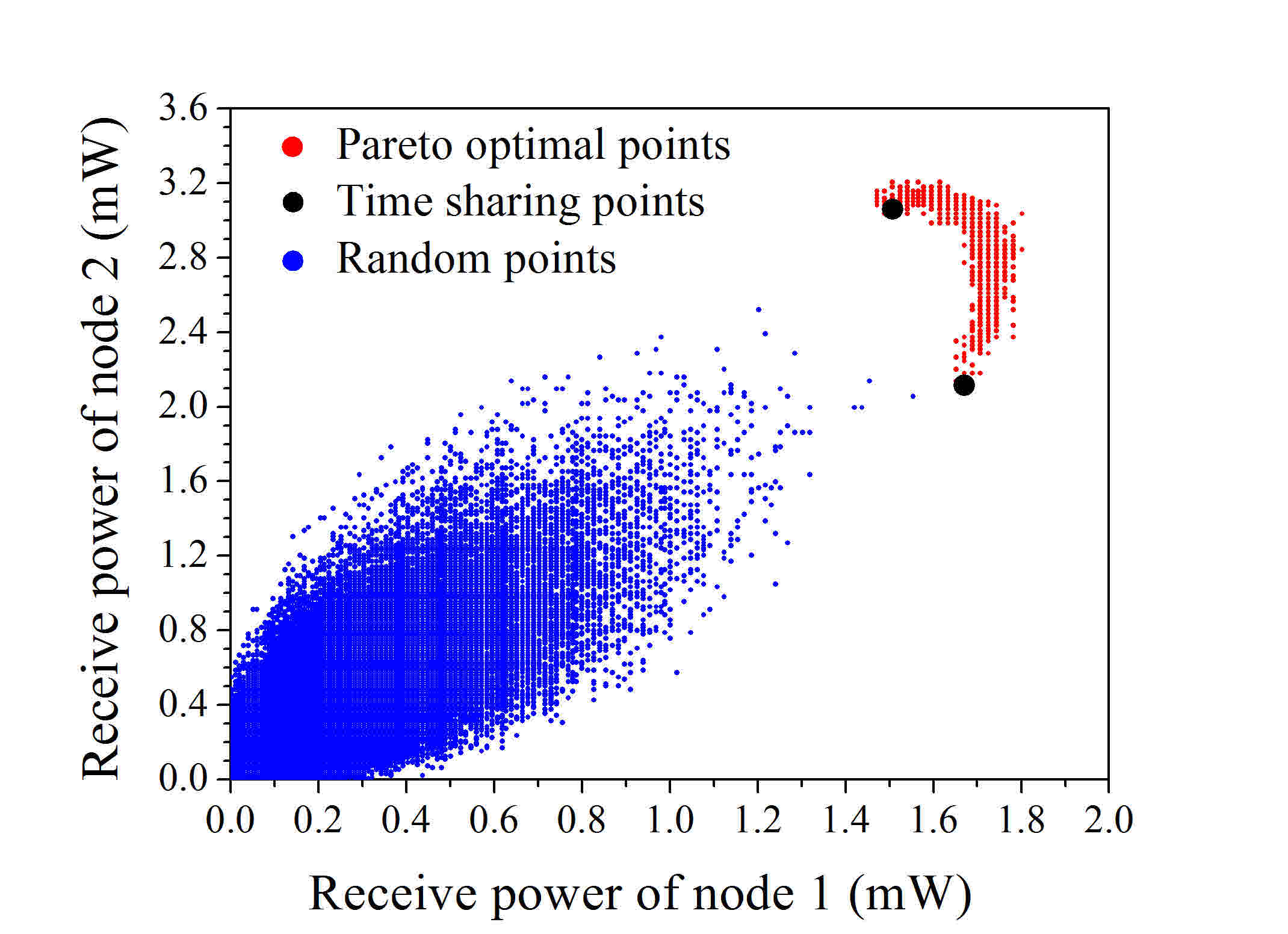}
        }~~~
    \subfigure[0$^\circ$, 10$^\circ$, 1120 mW]{
        \label{fig:pr2dcirc10deg2m1}\includegraphics[width=4.0cm, bb=1.5in 0.3in 9.8in 7.4in] {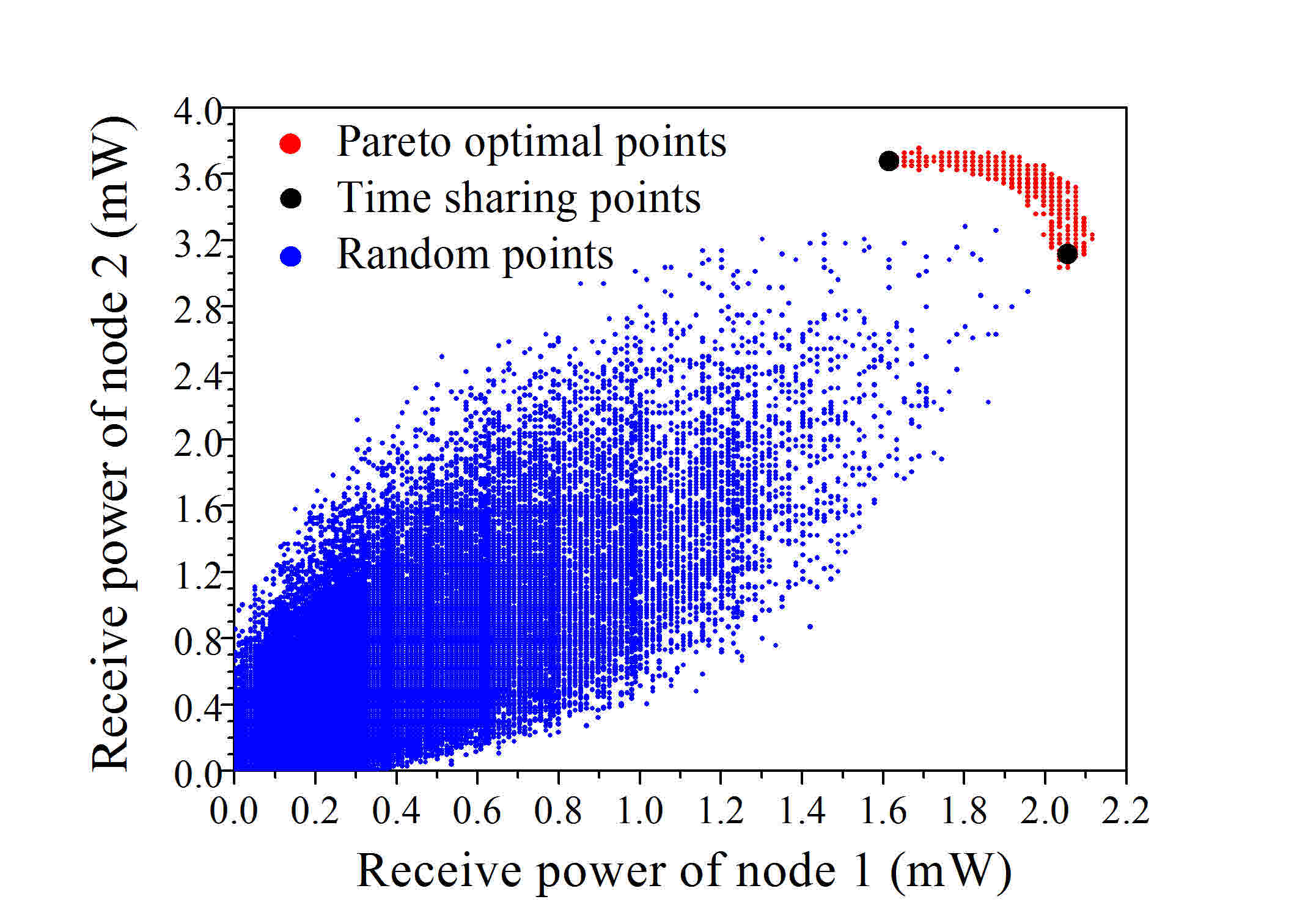}
        }\\
    \subfigure[0$^\circ$, 90$^\circ$, 560 mW]{
        ~~\label{fig:pr2dcirc90deg2m0_5}\includegraphics[width=4.0cm, bb=1.5in 0.3in 9.8in 7.4in] {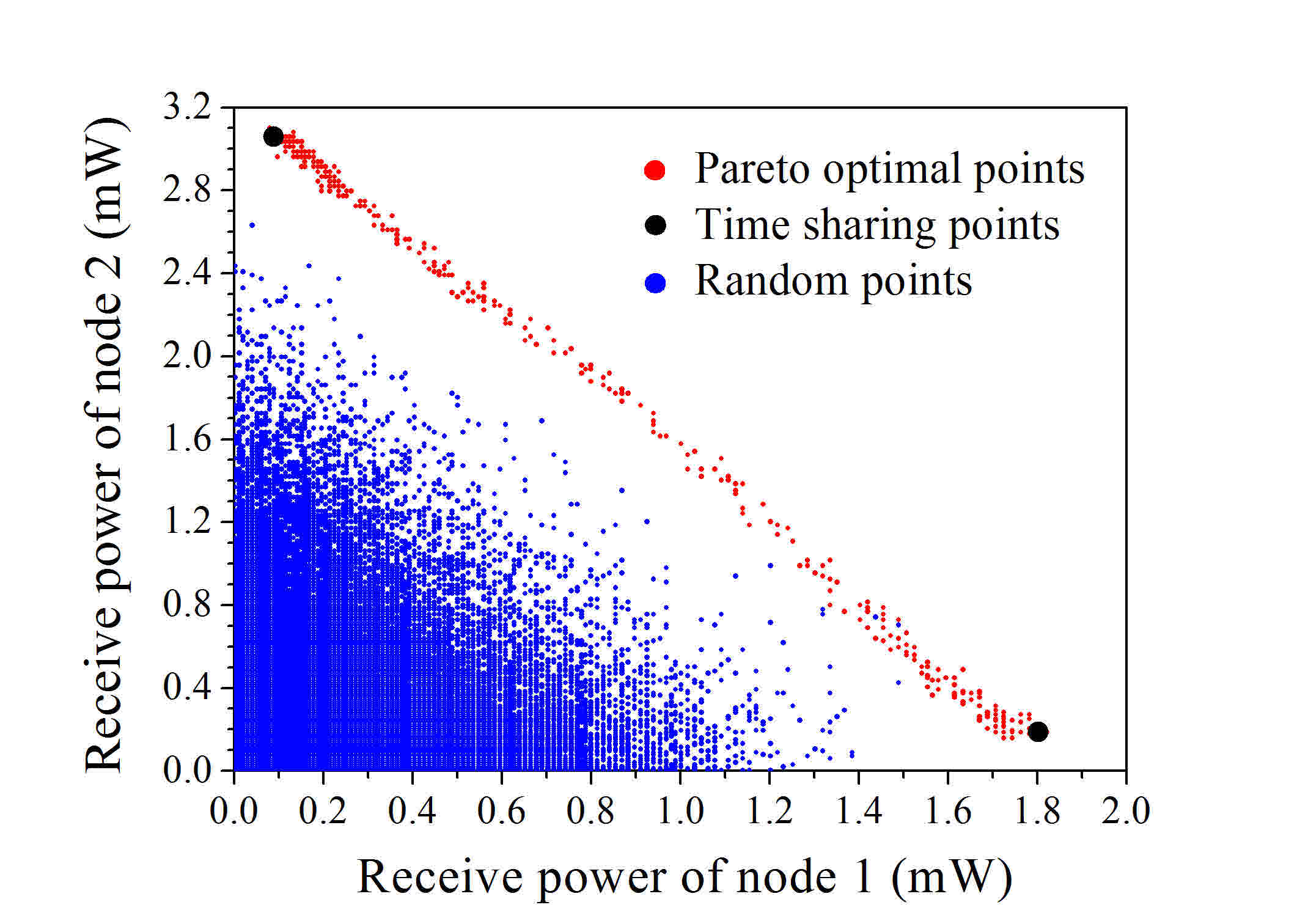}
        }~~~
    \subfigure[0$^\circ$, 90$^\circ$, 1120 mW]{
        \label{fig:pr2dcirc90deg2m1}\includegraphics[width=4.0cm, bb=1.5in 0.3in 9.8in 7.4in] {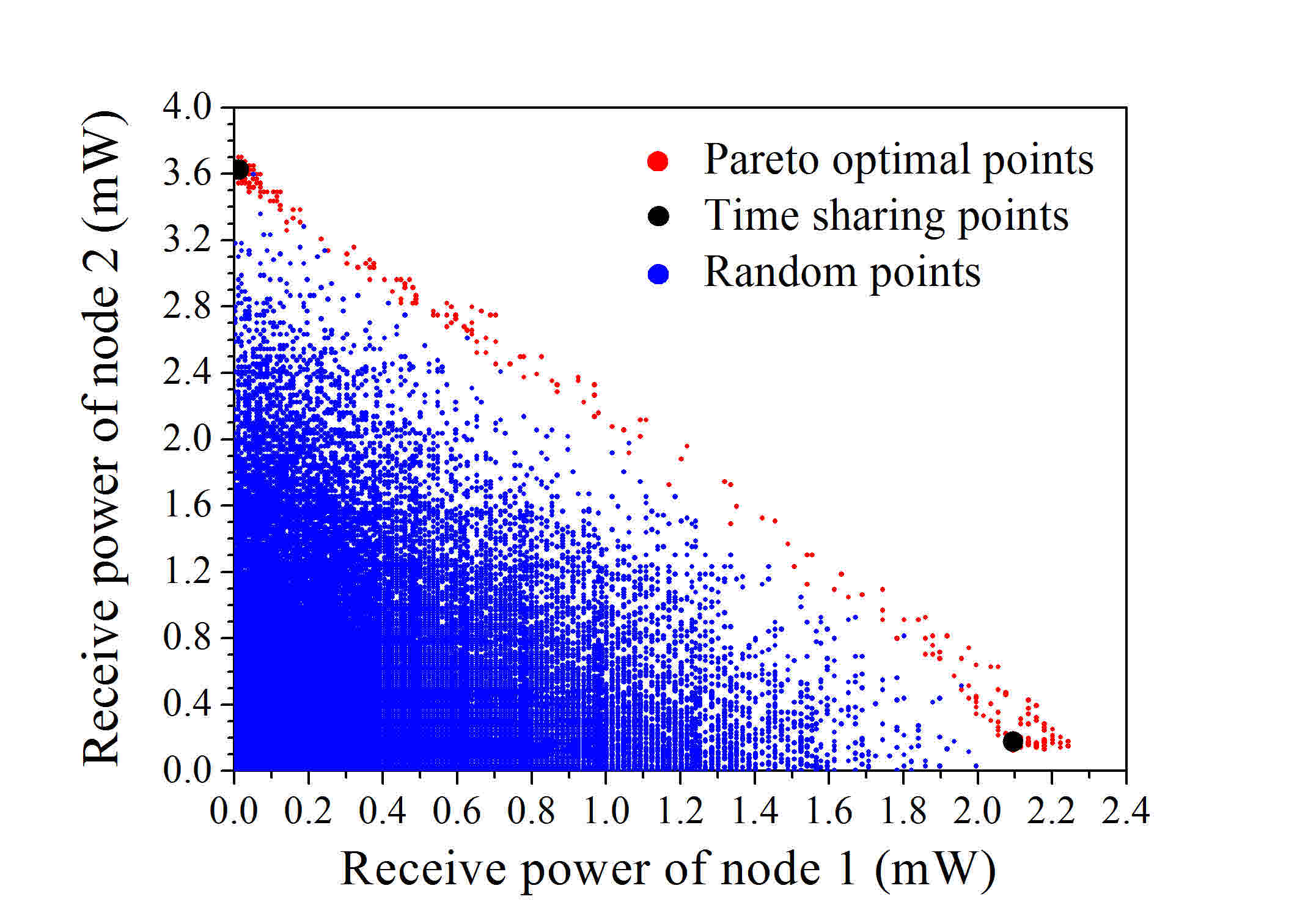}
        }
    \caption{Receive power region and Pareto frontier of two nodes with a circular antenna array.}
    \label{fig:pr2dcirc}
\end{figure}

\begin{figure}
    \centering
    \subfigure[0$^\circ$, 10$^\circ$, 560 mW]{
        ~~\label{fig:pr2dlin10deg2m0_5}\includegraphics[width=4.0cm, bb=1.5in 0.3in 9.9in 7.4in] {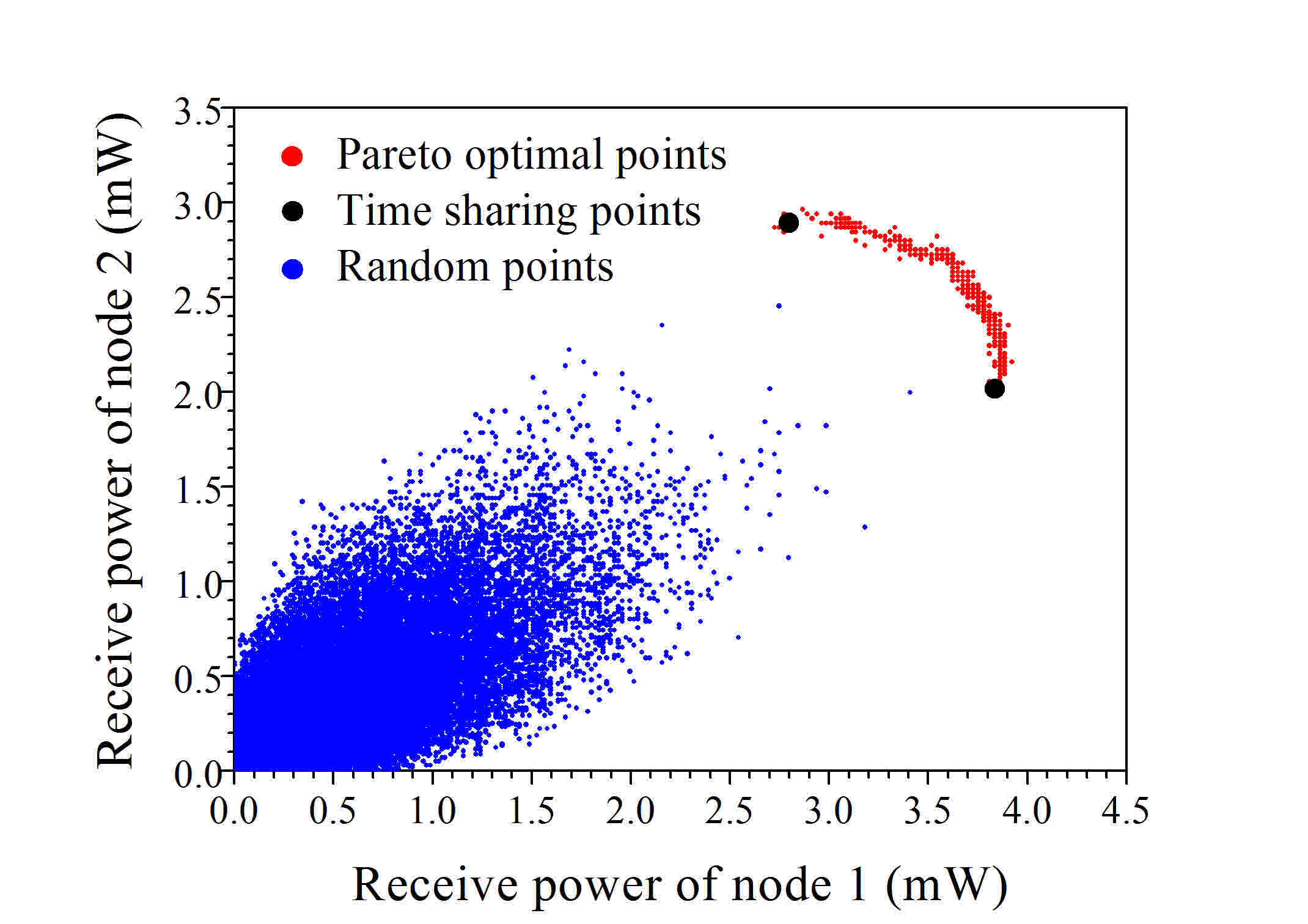}
        }~~~
    \subfigure[0$^\circ$, 10$^\circ$, 1120 mW]{
        \label{fig:pr2dlin10deg2m1}\includegraphics[width=4.0cm, bb=1.5in 0.3in 9.8in 7.4in] {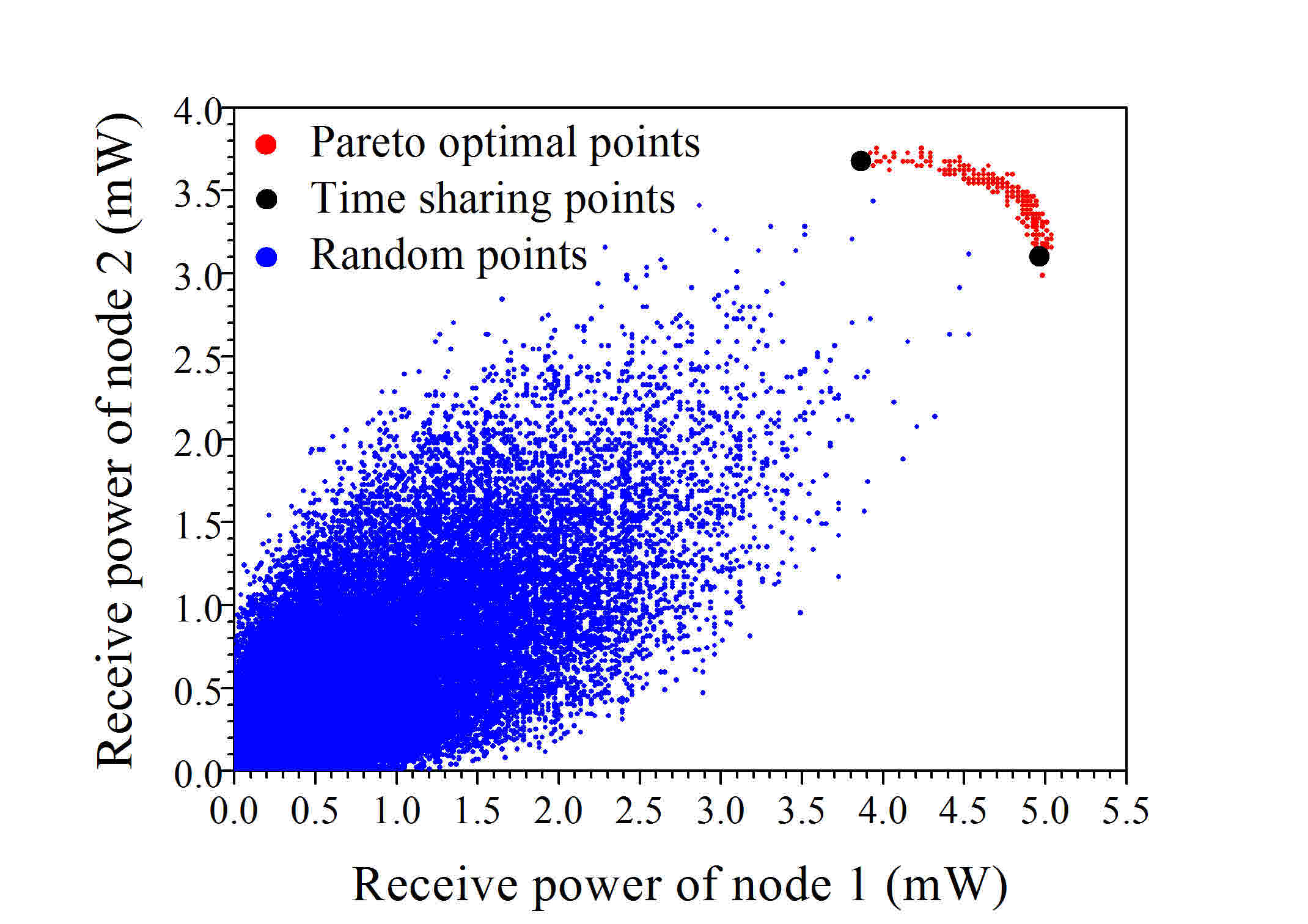}
        }\\
    \subfigure[0$^\circ$, 90$^\circ$, 560 mW]{
        ~~\label{fig:pr2dlin90deg2m0_5}\includegraphics[width=4.0cm, bb=1.5in 0.3in 9.8in 7.4in] {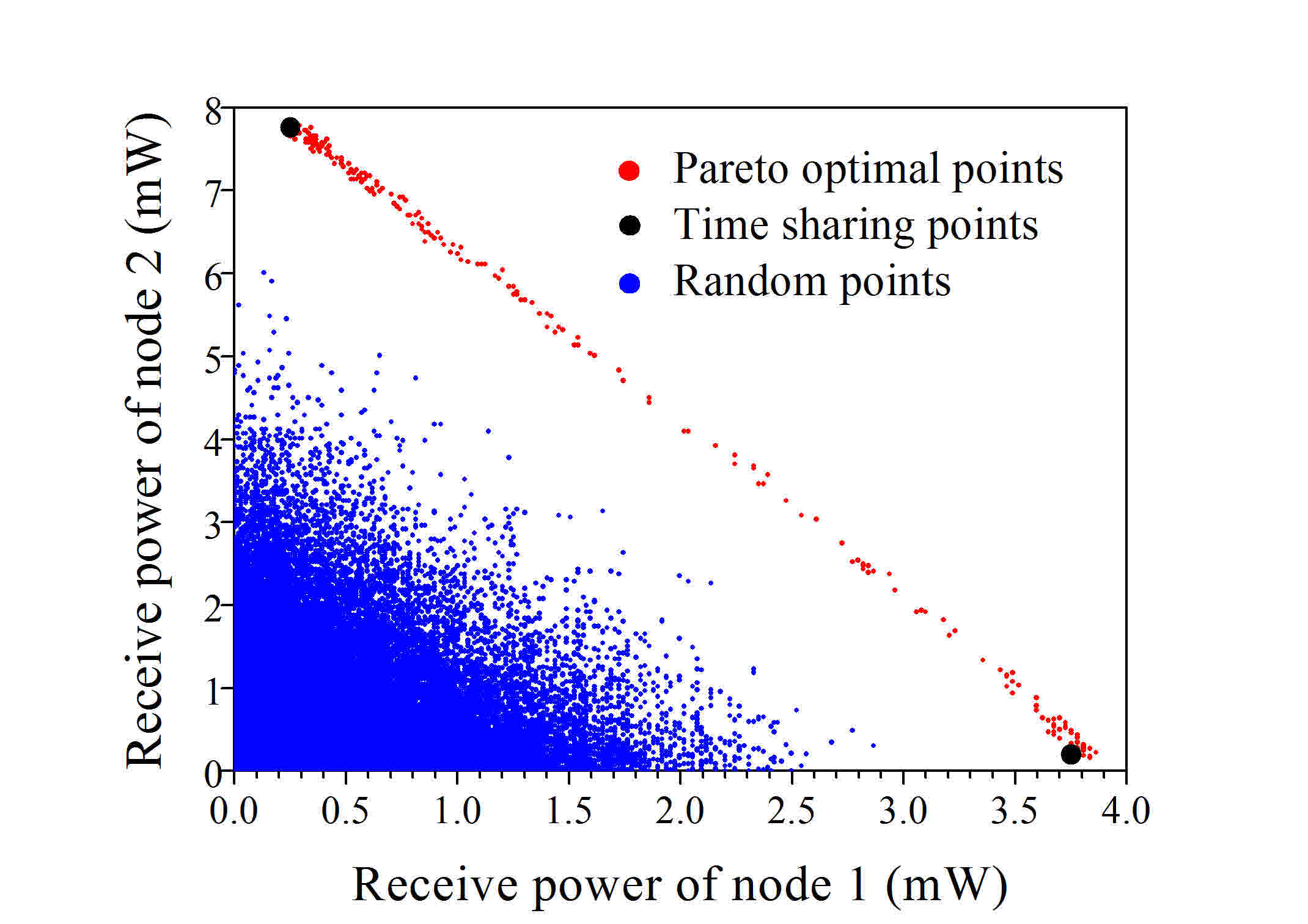}
        }~~~
    \subfigure[0$^\circ$, 90$^\circ$, 1120 mW]{
        \label{fig:pr2dlin90deg2m1}\includegraphics[width=4.0cm, bb=1.5in 0.3in 9.8in 7.4in] {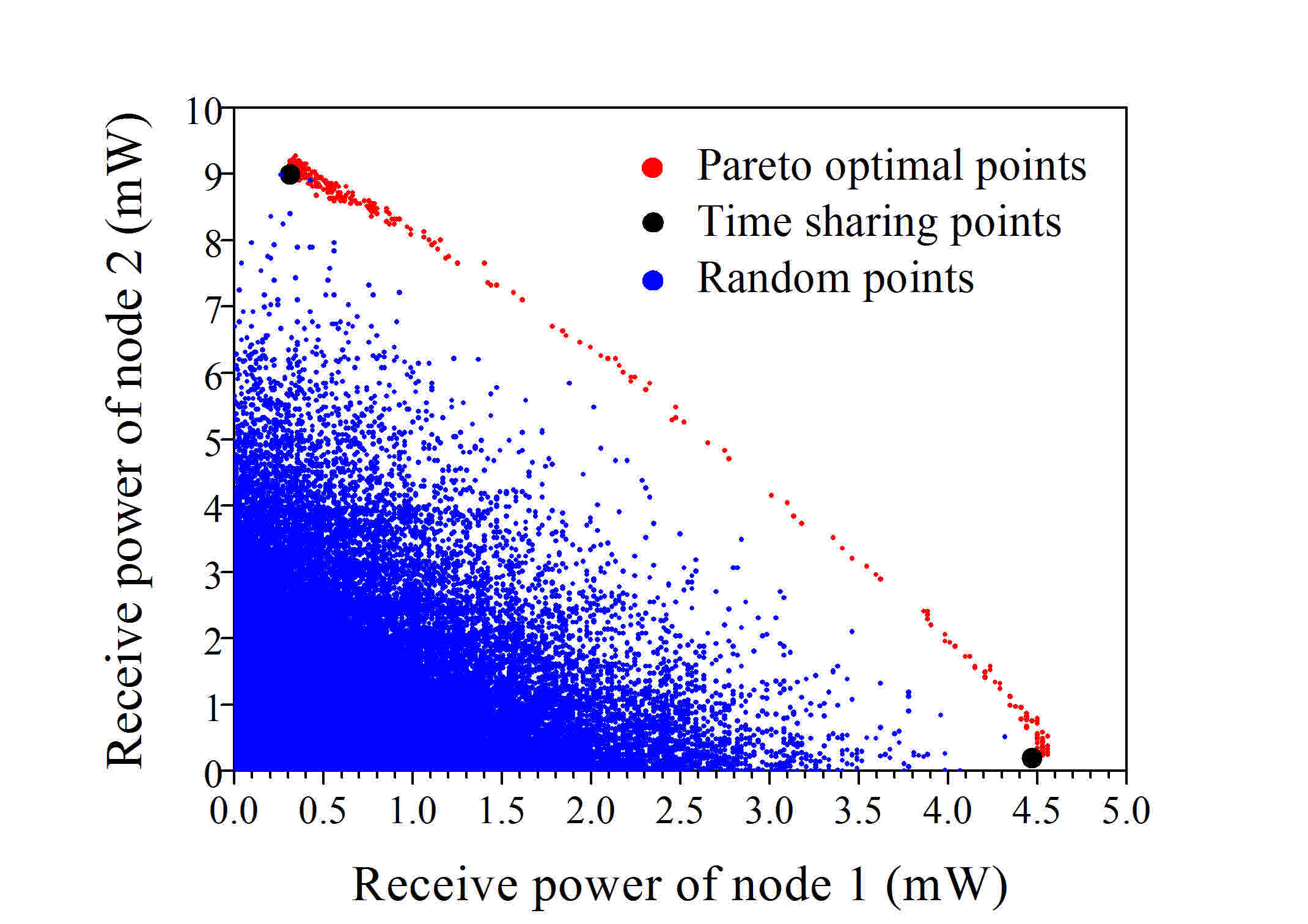}
        }
    \caption{Receive power region and Pareto frontier of two nodes with a linear antenna array.}
    \label{fig:pr2dlin}
\end{figure}

\begin{figure}
    \centering
    \subfigure[0$^\circ$, 10$^\circ$, 20$^\circ$, 560 mW]{
        ~~\label{fig:pr3dcirc10deg2m0_5}\includegraphics[width=4.0cm, bb=1.5in 0.3in 9in 7.4in] {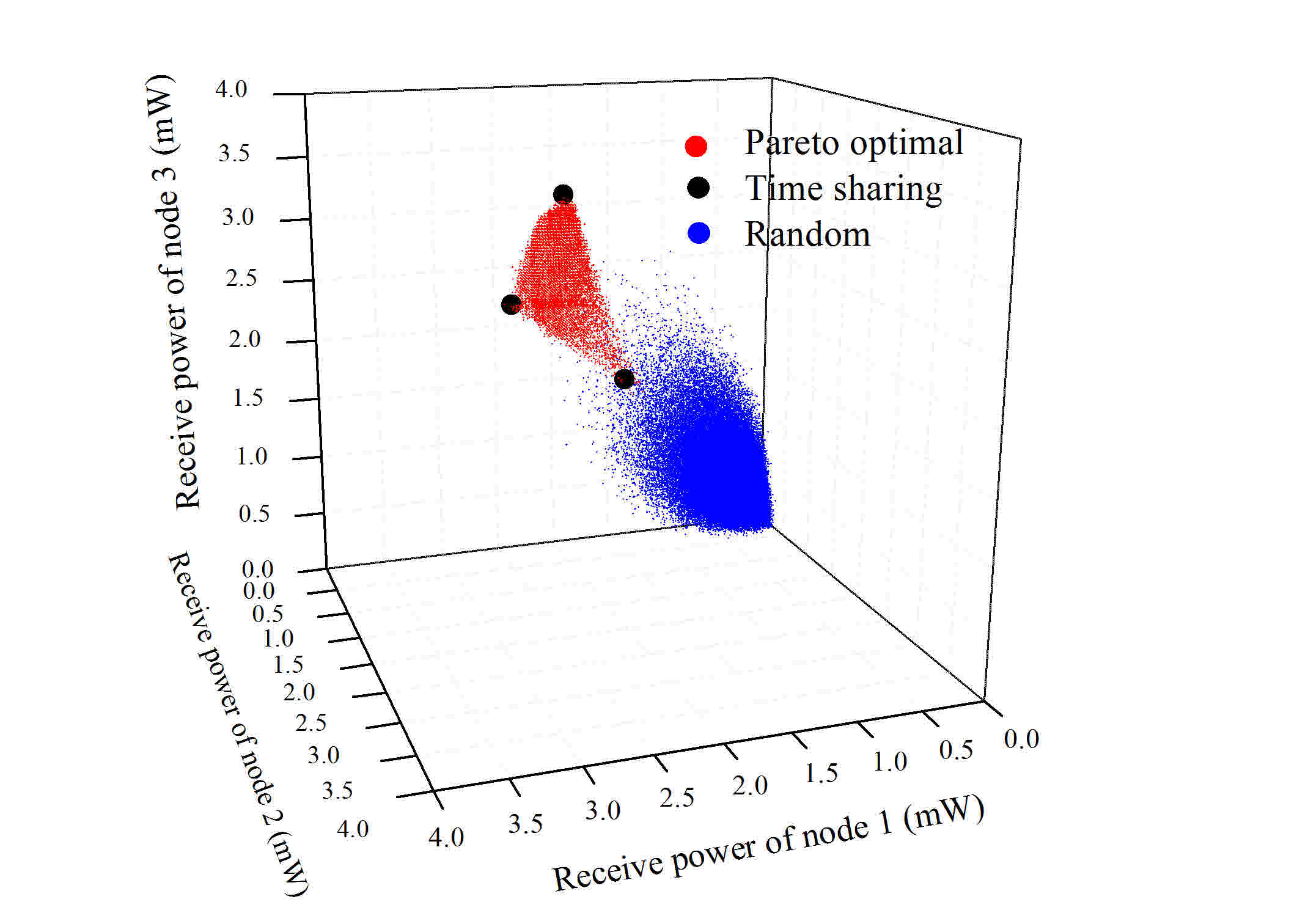}
        }~~~
    \subfigure[0$^\circ$, 10$^\circ$, 20$^\circ$, 1120 mW]{
        \label{fig:pr3dcirc10deg2m1}\includegraphics[width=4.0cm, bb=1.5in 0.3in 9in 7.4in] {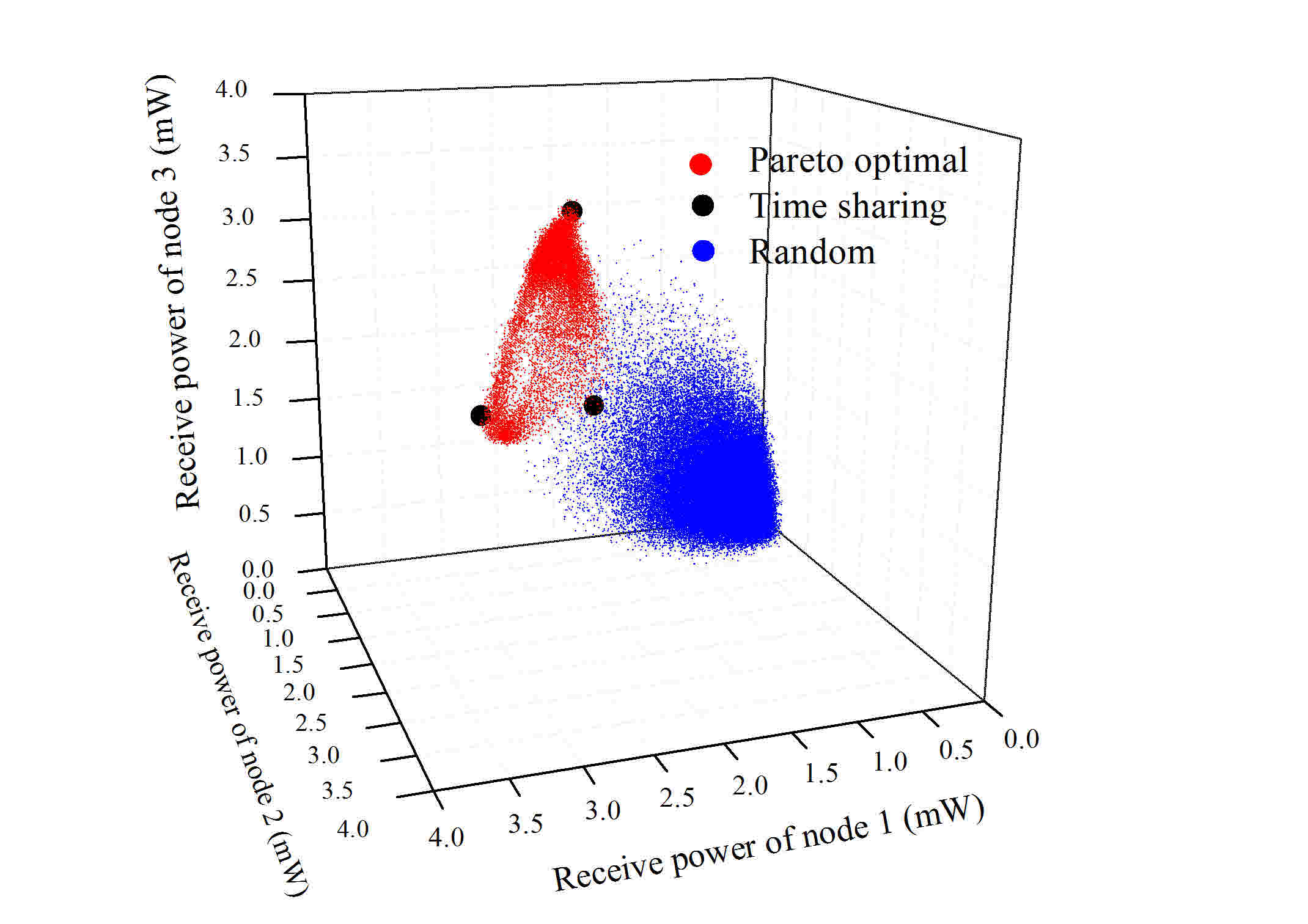}
        }\\
    \subfigure[0$^\circ$, 120$^\circ$, 240$^\circ$, 560 mW]{
        ~~\label{fig:pr3dcirc120deg2m0_5}\includegraphics[width=4.0cm, bb=1.5in 0.3in 9in 7.4in] {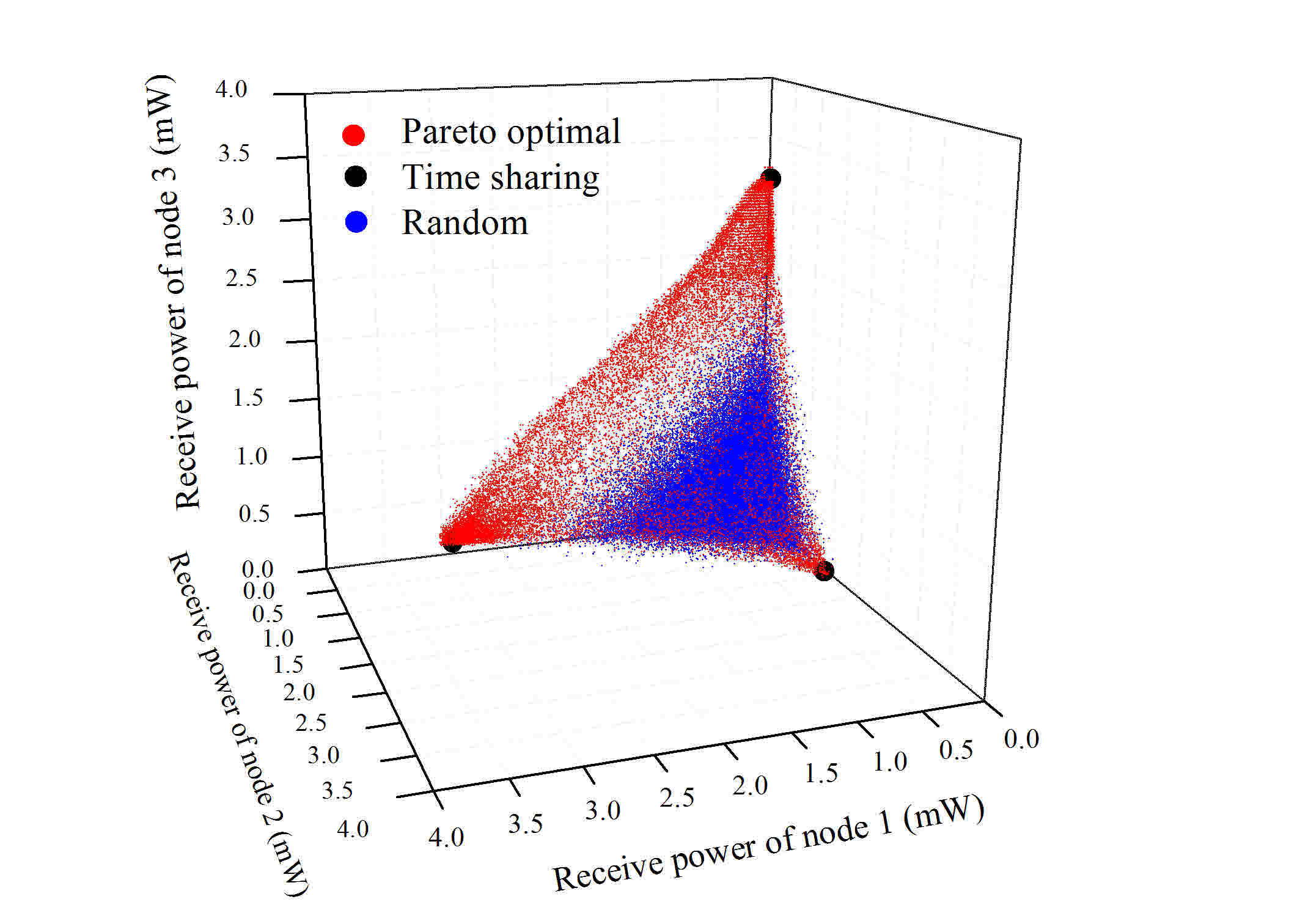}
        }~~~
    \subfigure[0$^\circ$, 120$^\circ$, 240$^\circ$, 1120 mW]{
        \label{fig:pr3dcirc120deg2m1}\includegraphics[width=4.0cm, bb=1.5in 0.3in 9in 7.4in] {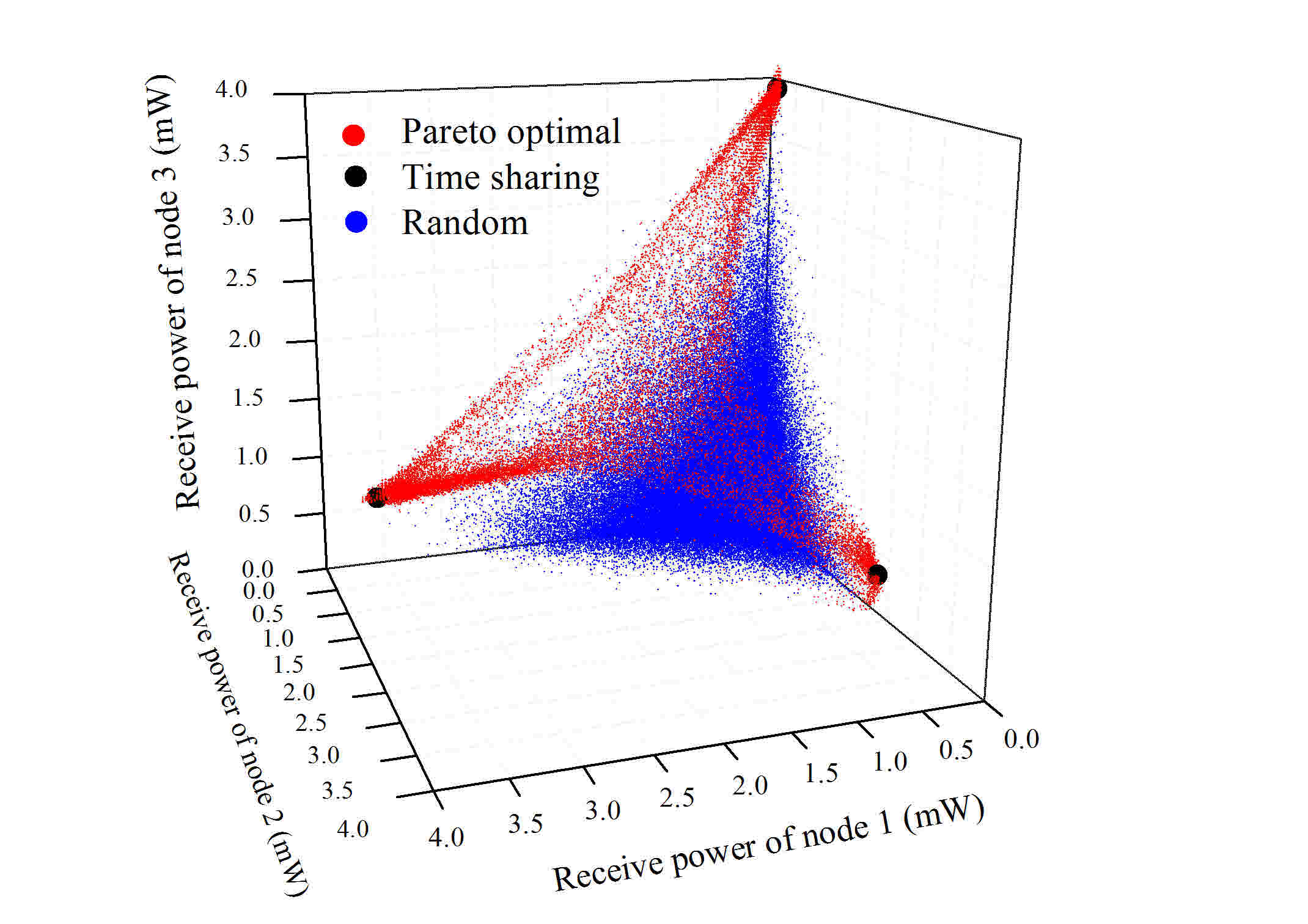}
        }
    \caption{Receive power region and Pareto frontier of three nodes with a circular antenna array.}
    \label{fig:pr3dcirc}
\end{figure}

\begin{figure}
    \centering
    \subfigure[0$^\circ$, 10$^\circ$, 20$^\circ$, 560 mW]{
        ~~\label{fig:pr3dlin10deg2m0_5}\includegraphics[width=4.0cm, bb=1.5in 0.3in 9in 7.4in] {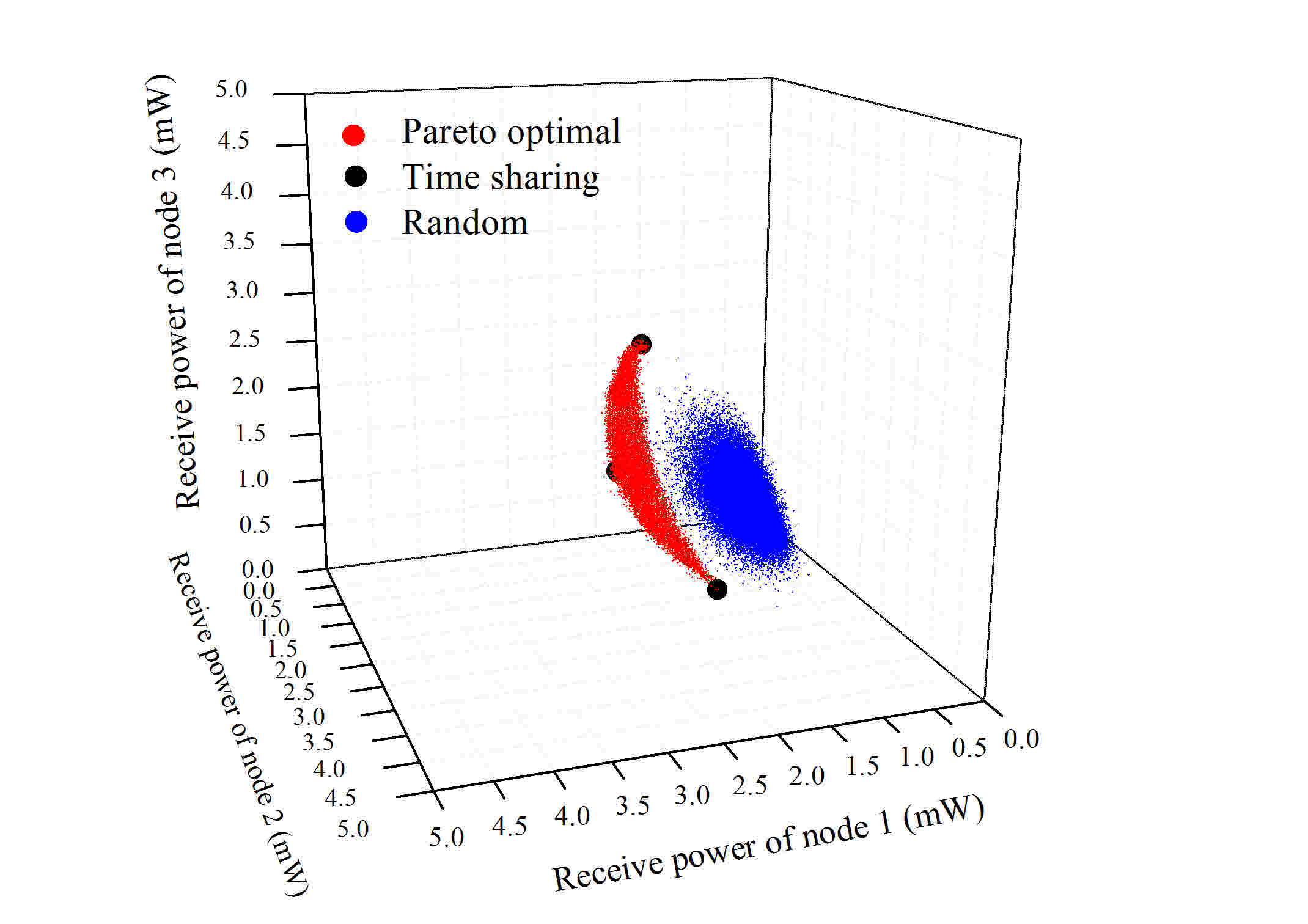}
        }~~~
    \subfigure[0$^\circ$, 10$^\circ$, 20$^\circ$, 1120 mW]{
        \label{fig:pr3dlin10deg2m1}\includegraphics[width=4.0cm, bb=1.5in 0.3in 9in 7.4in] {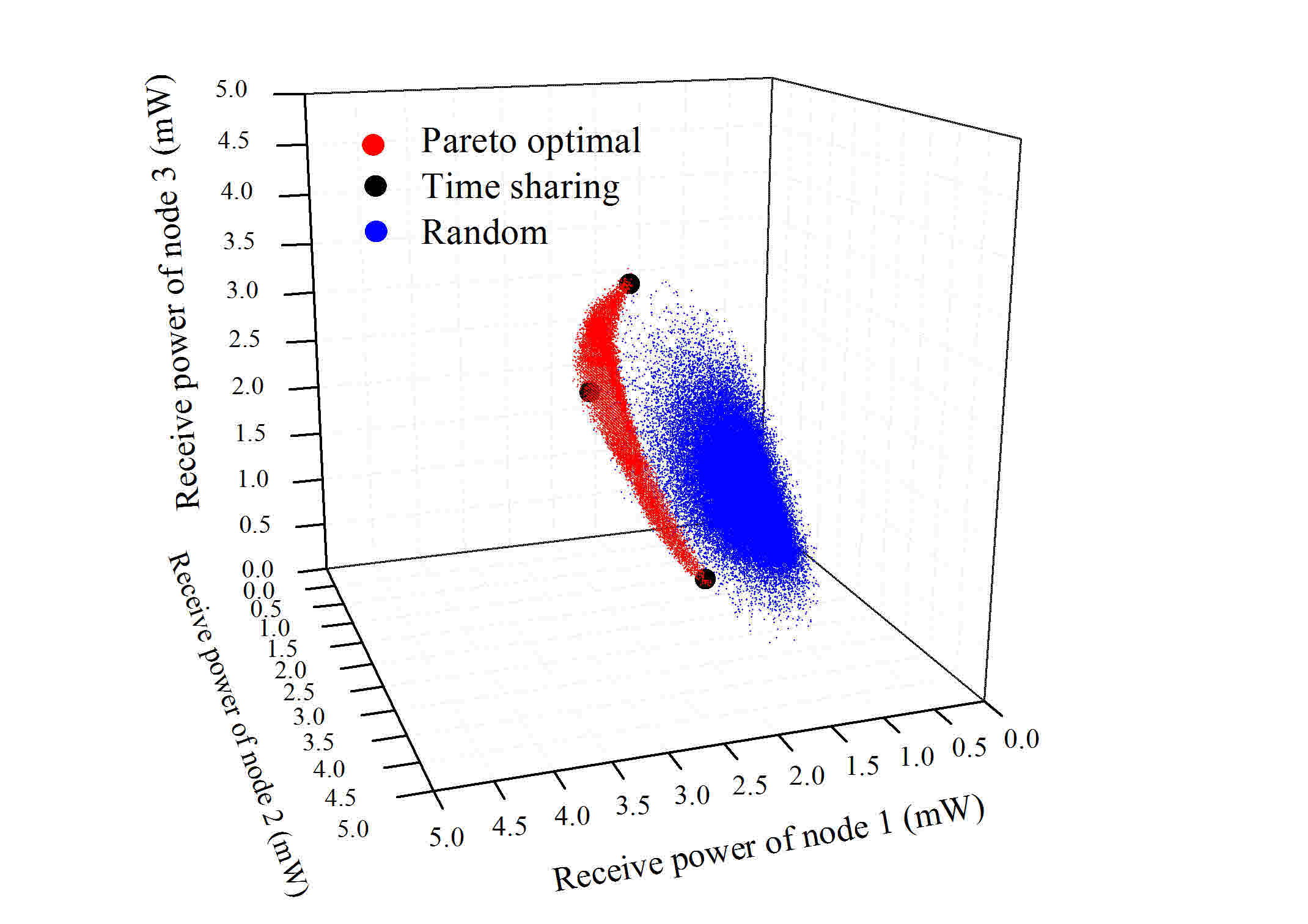}
        }\\
    \subfigure[0$^\circ$, 120$^\circ$, 240$^\circ$, 560 mW]{
        ~~\label{fig:pr3dlin120deg2m0_5}\includegraphics[width=4.0cm, bb=1.5in 0.3in 9in 7.4in] {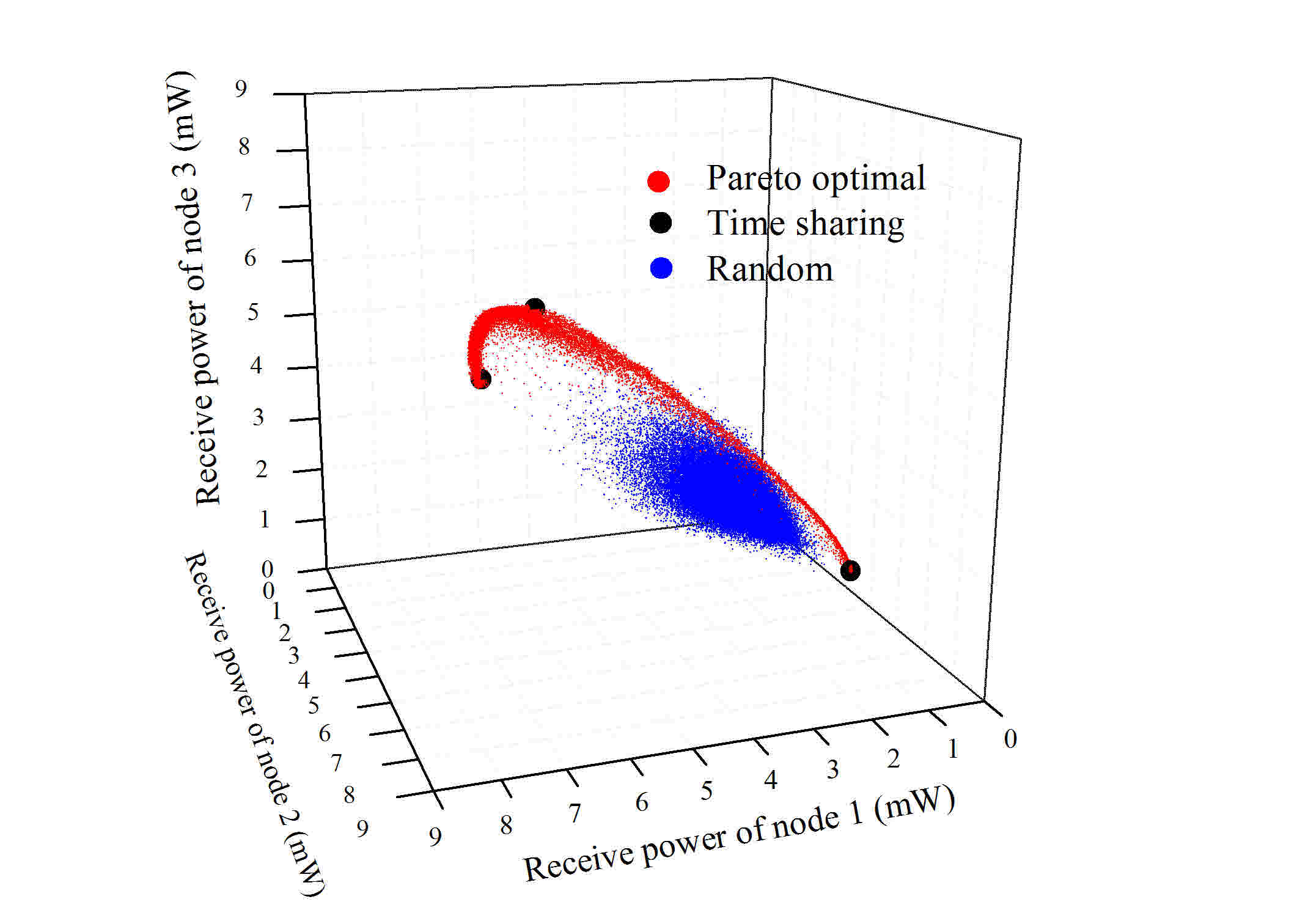}
        }~~~
    \subfigure[0$^\circ$, 120$^\circ$, 240$^\circ$, 1120 mW]{
        \label{fig:pr3dlin120deg2m1}\includegraphics[width=4.0cm, bb=1.5in 0.3in 9in 7.4in] {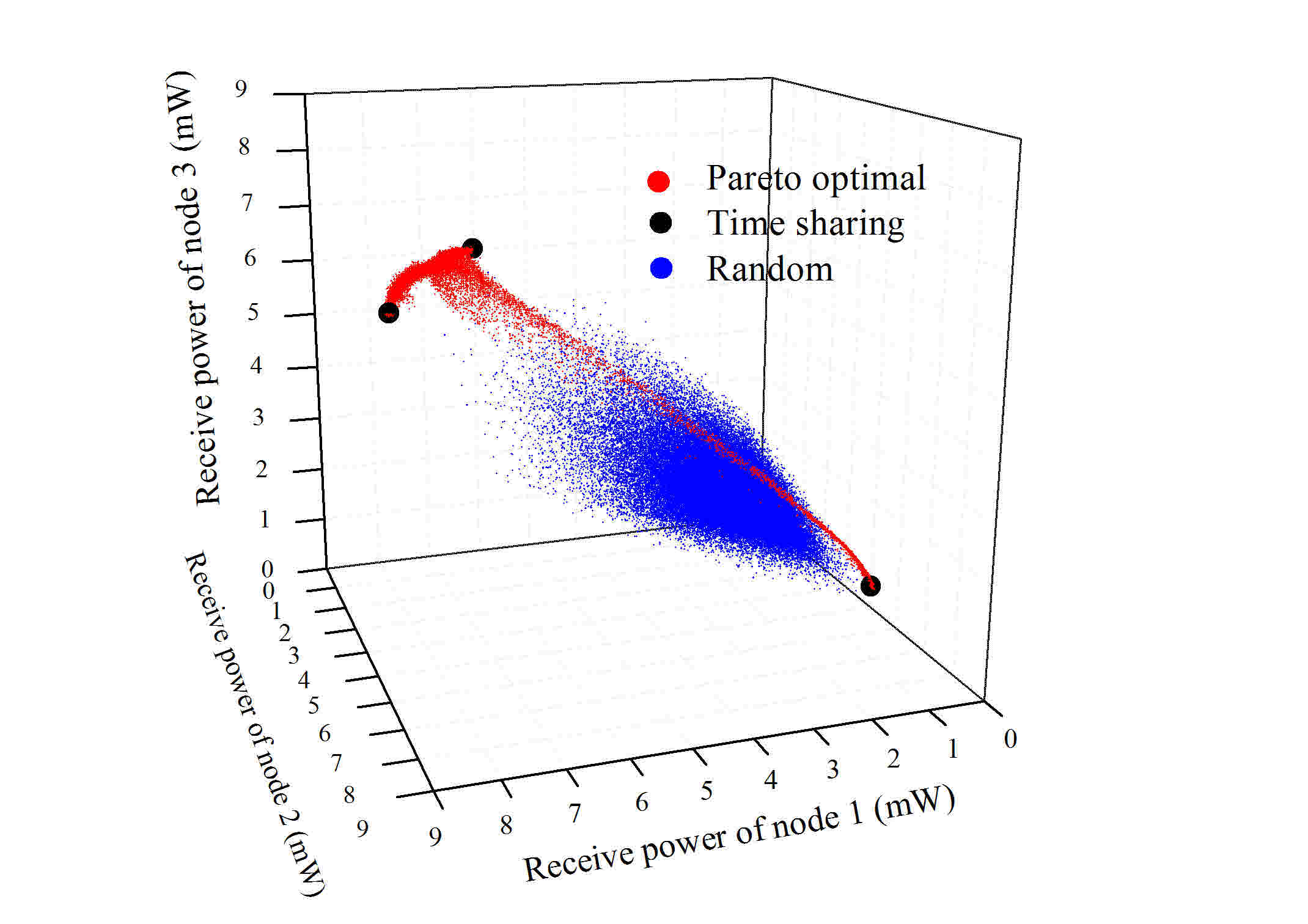}
        }
    \caption{Receive power region and Pareto frontier of three nodes with a linear antenna array.}
    \label{fig:pr3dlin}
\end{figure}

We show the beam-splitting gain in Fig.~\ref{fig:bsg} to compare the beam-splitting and time-sharing beamforming techniques.
In Fig.~\ref{fig:bsgloc2n} and \ref{fig:bsgloc3n}, we show the beam-splitting gain according to the node location.
In Fig.~\ref{fig:bsgloc2n}, there are two nodes with locations (2m, 0$^\circ$) and (2m, $x^\circ$), where $x$ is the x-axis of the graph.
In Fig.~\ref{fig:bsgloc3n}, there are three nodes with locations (2m, 0$^\circ$), (2m, $x^\circ$), and (2m, $2x^\circ$), where $x$ is the x-axis of the graph.
For Fig.~\ref{fig:bsgloc2n} and \ref{fig:bsgloc3n}, the maximum total transmit power is fixed to $P_\text{tot}=1120$ mW.
In Fig.~\ref{fig:bsgmaxpow2n} and \ref{fig:bsgmaxpow3n}, we show the beam-splitting gain according to the maximum total transmit power.
There are two nodes with locations (2m, 0$^\circ$) and (2m, $90^\circ$) in Fig.~\ref{fig:bsgmaxpow2n}, and there are three nodes with locations (2m, 0$^\circ$), (2m, 120$^\circ$), and (2m, 240$^\circ$) in Fig.~\ref{fig:bsgmaxpow3n}.
In all graphs in Fig.~\ref{fig:bsg}, we can see that the beamforming gain ranges from 1 to 1.45, which means up to 45\% gain can be achieved by the beam-splitting beamforming technique.

\begin{figure}
    \centering
    \subfigure[Two nodes, according to node locations]{
        \label{fig:bsgloc2n}\includegraphics[width=4.0cm, bb=1in 0.3in 9in 7.4in] {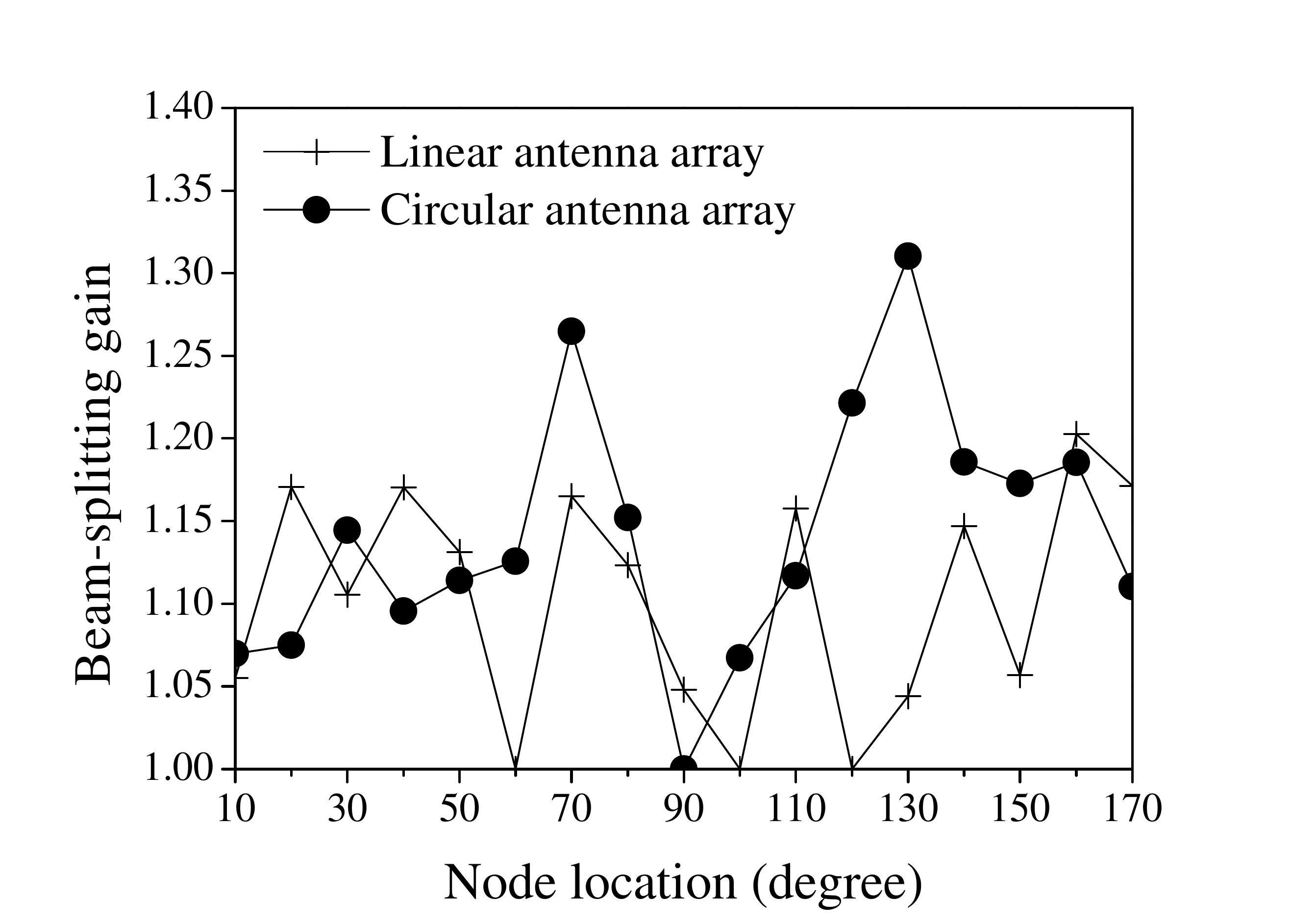}
        }~
    \subfigure[Three nodes, according to node locations]{
        \label{fig:bsgloc3n}\includegraphics[width=4.0cm, bb=1in 0.3in 9in 7.4in] {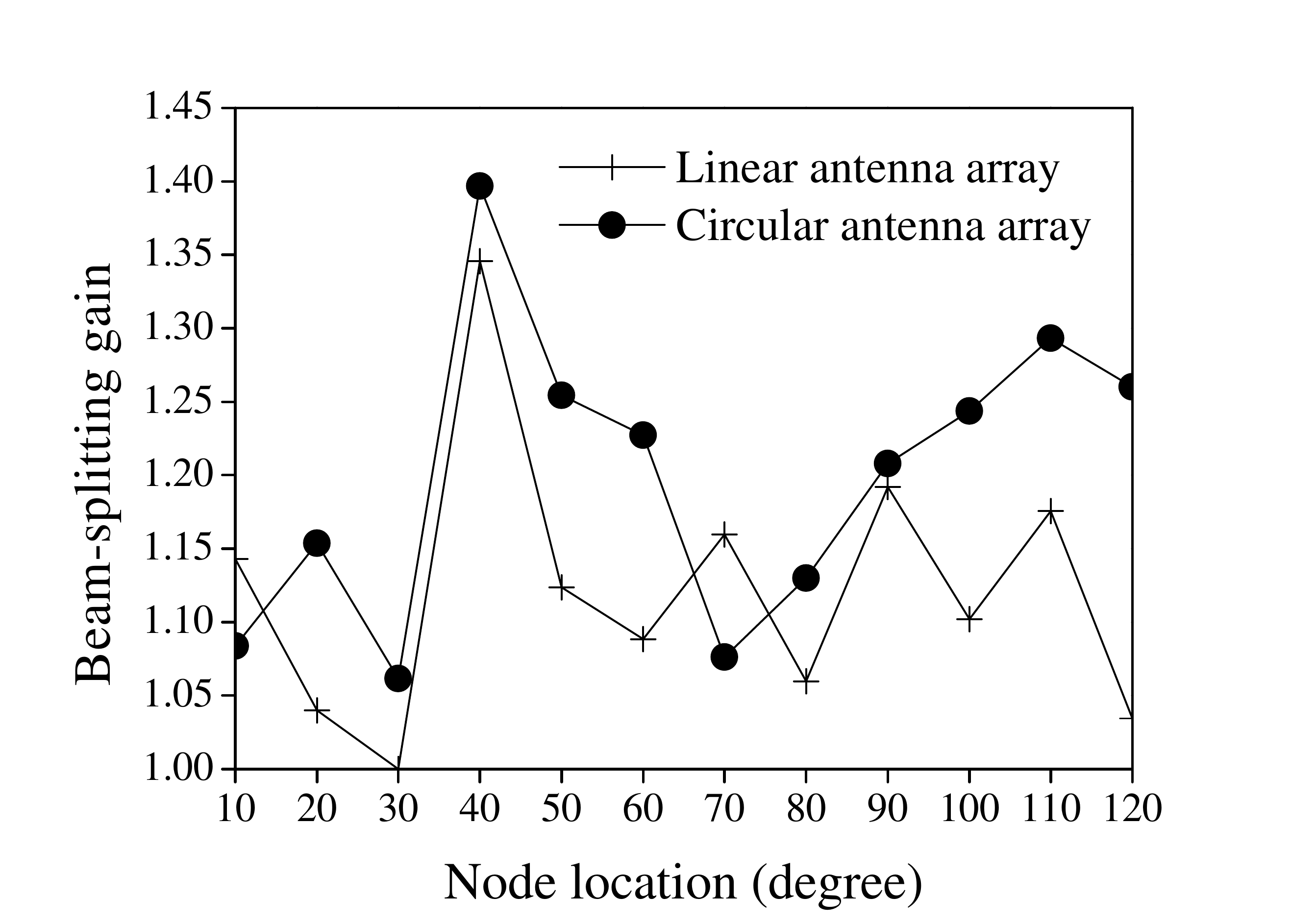}
        }\\
    \subfigure[Two nodes, according to maximum total transmit power]{
        \label{fig:bsgmaxpow2n}\includegraphics[width=4.0cm, bb=1in 0.3in 9in 7.4in] {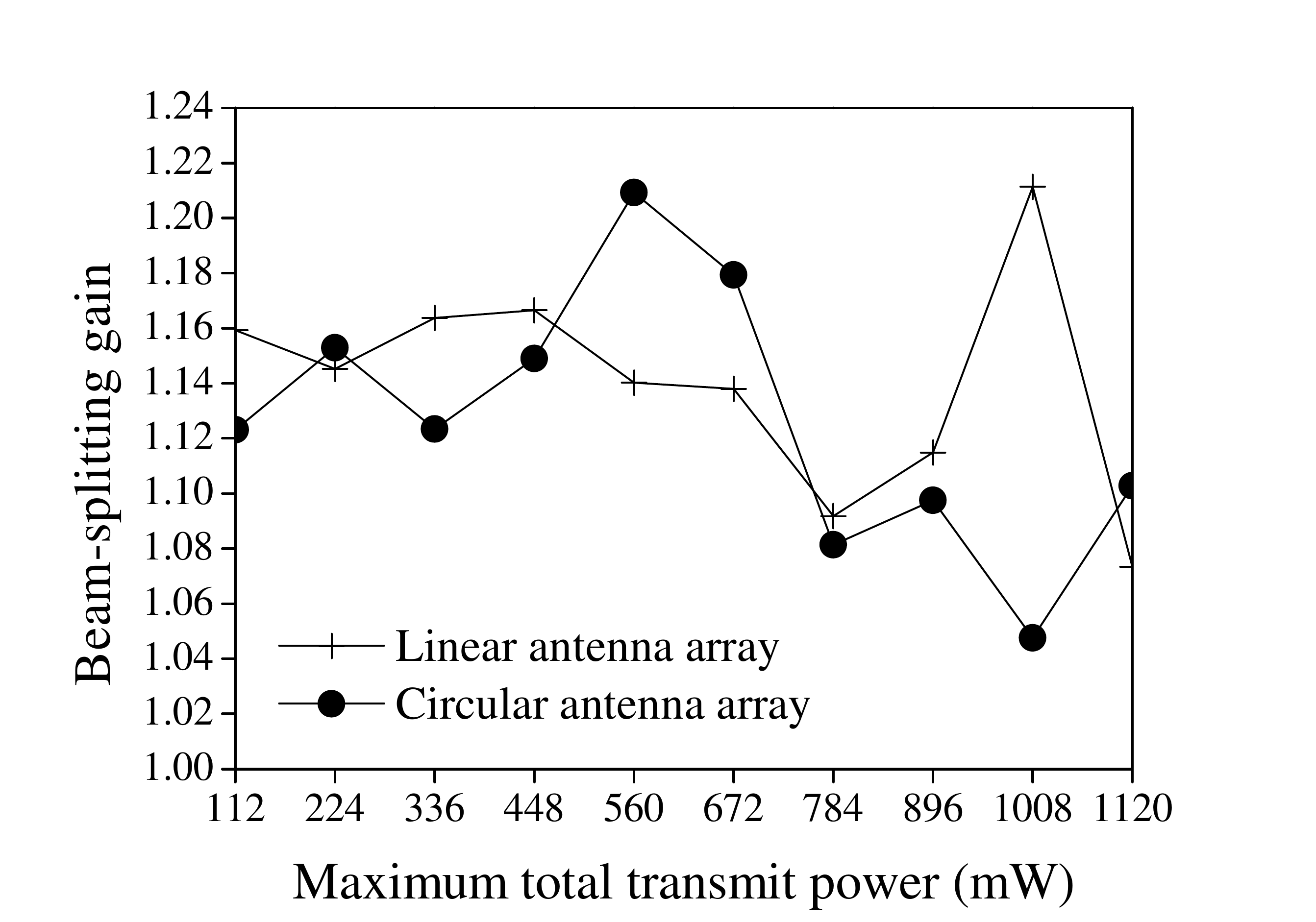}
        }~
    \subfigure[Three nodes, according to maximum total transmit power]{
        \label{fig:bsgmaxpow3n}\includegraphics[width=4.0cm, bb=1in 0.3in 9in 7.4in] {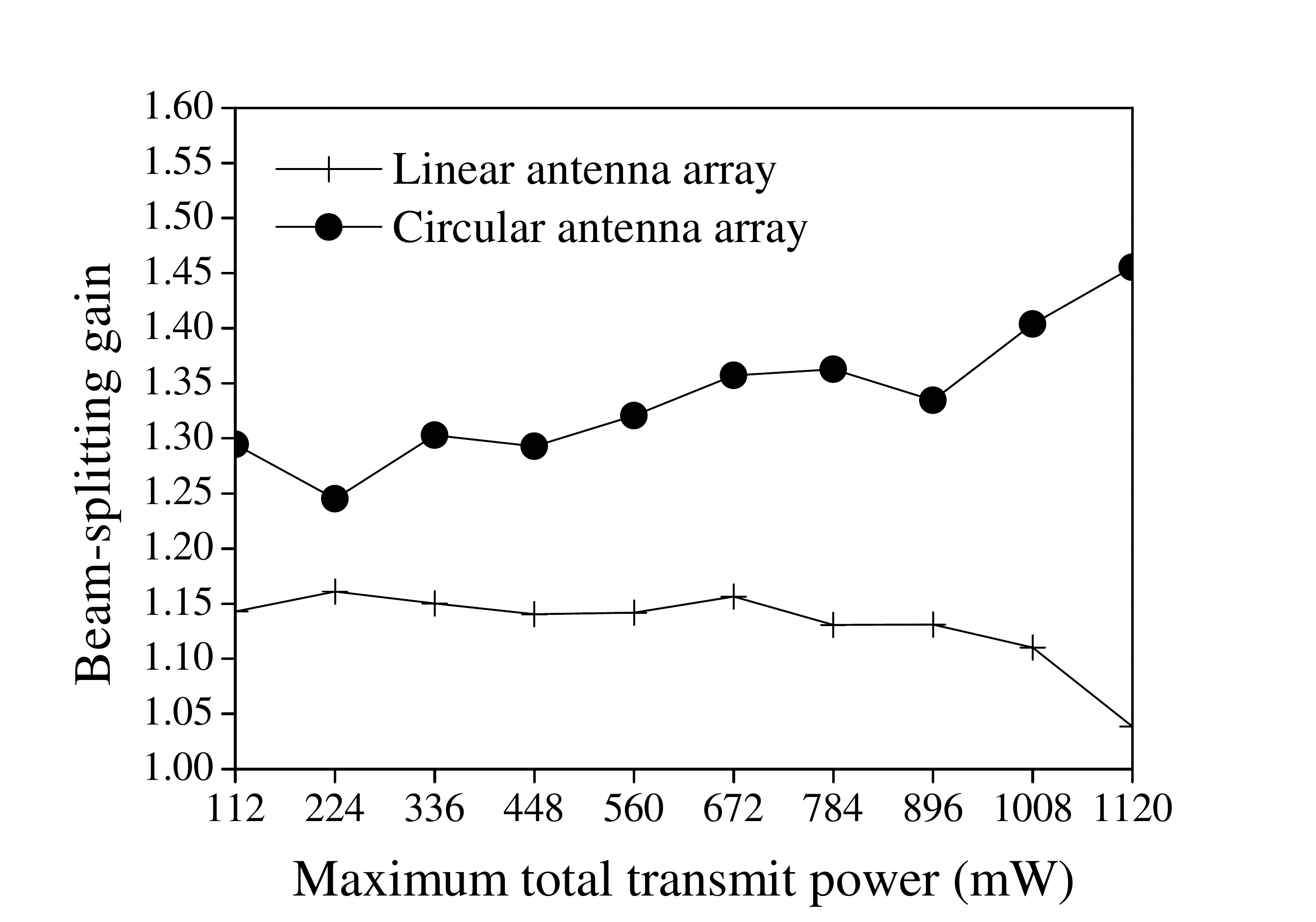}
        }
    \caption{Beam-splitting gain.}
    \label{fig:bsg}
\end{figure}

\section{Joint Beam-Splitting and Energy Neutral Control Method}\label{section:jointmethod}

\subsection{Protocol Description}

In this subsection, we introduce a multi-node multi-antenna WPSN protocol.
Note that this protocol is an extended version of the multi-antenna WPSN protocol for a single node in our previous work \cite{Choi:2017}.
The proposed protocol targets to distribute the RF power to each sensor node in a fair manner by using a beam-splitting beamforming technique.
In addition, the proposed protocol balances the consumed power and harvested power of each node so that a node is not blacked out.
To this end, the proposed protocol should be able to obtain the necessary information (e.g., channel gain and stored energy), and to adaptively control the beamforming weights and the consumed power based on the obtained information.

\begin{figure}
	\centering
    \includegraphics[width=8cm, bb=1.2in 3.7in 6.5in 7.8in]{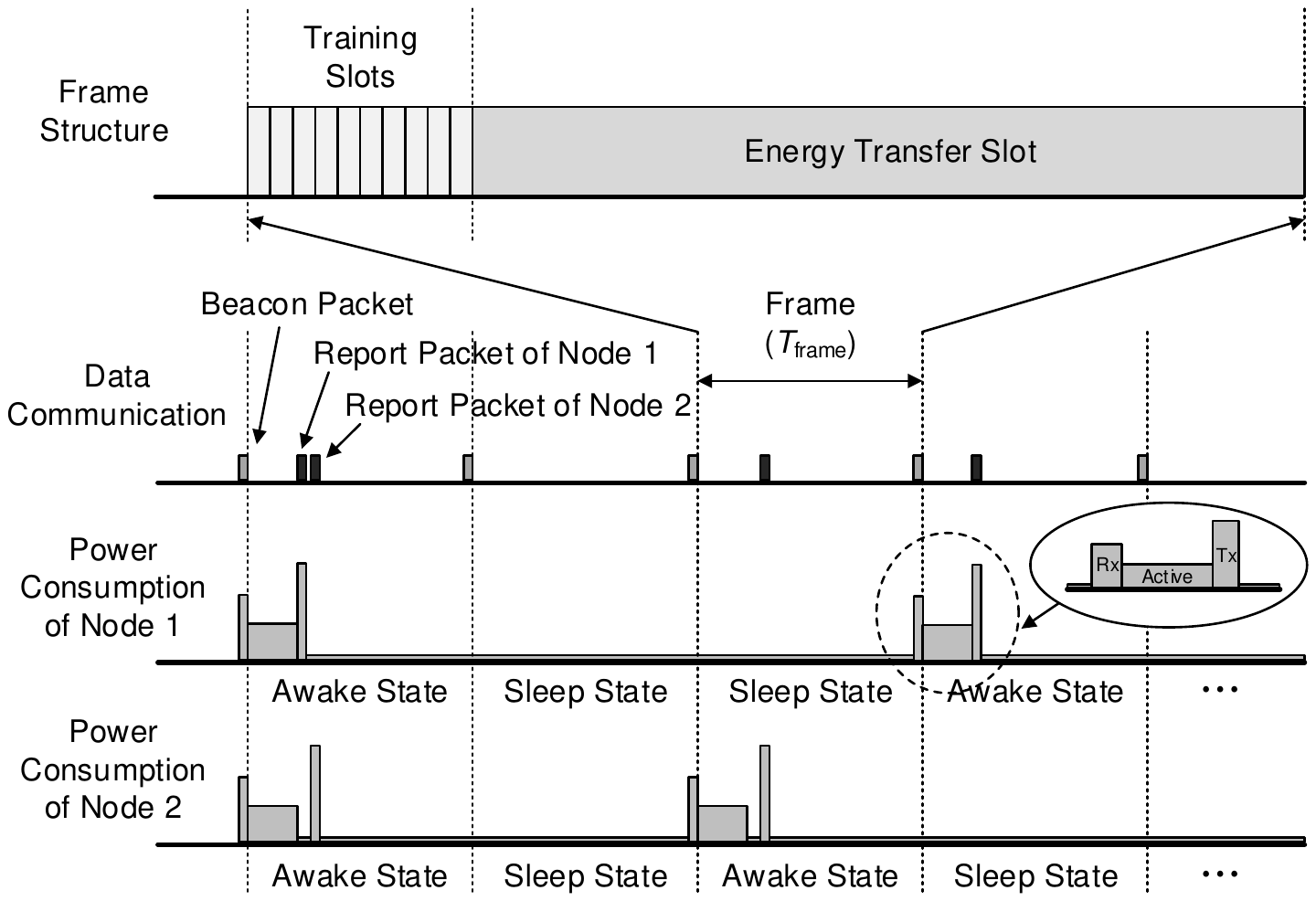}
    \caption{Multi-node multi-antenna WPSN protocol.}
    \label{fig:protocol}
\end{figure}

In the proposed protocol, time is divided into frames, each of which is indexed by $t$, as illustrated in Fig.~\ref{fig:protocol}.
The length of one frame is denoted by $T_\text{frame}$.
Within a frame, multiple training slots are followed by an energy transfer slot.
The training slots are used for estimating the channel gains $h_{k,n}$.
In \cite{Choi:2017} and \cite{Choi:2017_2}, we have proposed a channel estimation method, which sends RF signals with training beamforming weights during the training slots and estimates the channel gains based on the receive power at a node during the training slots.
In this paper, we adopt the channel estimation method in \cite{Choi:2017} and \cite{Choi:2017_2}, which is indispensable for implementing a real-life multi-node multi-antenna WPSN testbed here.

The power beacon forms energy beams to transfer RF energy to nodes during an energy transfer slot of each frame.
Let $T_\text{es}$ denote the length of an energy transfer slot.
During the energy transfer slot of frame $t$, the power beacon uses $w_n(t)$ as the beamforming weight of antenna $n$.
The beamforming weight vector in frame $t$ is denoted by ${\bold w}(t) = (w_1(t),\ldots,w_N(t))^T$.
The per-antenna and total power constraints are satisfied for the beamforming weight vector during the energy transfer slot, i.e., $|w_n(t)|^2\le P_\text{ant}$ for all $n=1,\ldots,N$ and ${\bold w}(t)^H {\bold w}(t)\le P_\text{tot}$.
We assume that the channel gain $h_{k,t}$ does not change over frame.
Therefore, the receive power at node $k$ during the energy transfer slot of frame $t$ is $r_k(t) = |{\bold h}_k^T {\bold w}(t)|^2$.

Each node controls its consumed power by adaptively switching between a sleep state and an awake state.
A node can be either in the sleep state or in the awake state during a frame.
If a sensor node is in the sleep state in frame $t$, the sensor node is in the idle mode during the whole frame $t$.
Therefore, very small power is consumed by a node in the sleep state.
On the other hand, if a sensor node is in the awake state in frame $t$, the sensor node is put into the receive, active, transmit, and idle modes during frame $t$.
The time duration in mode $m$ is denoted by $T_m$ during a frame in which a sensor node is in the awake state.
Then, it is satisfied that $\sum_{m\in {\mathcal M}} T_m = T_\text{frame}$.

Let $a_k(t)$ denote the activity of node $k$ in frame $t$.
We have $a_k(t)=0$ if node $k$ is in the sleep state in frame $t$, and $a_k(t)=1$ if node $k$ is in the awake state in frame $t$.
The average consumed power is reduced by decreasing the ratio of frames in the awake state.
The awake frame ratio $\sigma_k(t)$ is defined as the probability that node $k$ is in the awake state in frame $t$.
The awake frame ratio $\sigma_k(t)$ is determined by the power beacon, which is then notified to node $k$.
Then, node $k$ randomly decides the activity $a_k(t)$ according to the awake frame ratio so that $\sigma_k(t) = \ex[a_k(t)]$ is satisfied.

The power beacon sends a beacon packet at the start of each frame.
A beacon packet contains the awake frame ratios for all nodes.
At the start of a frame, all nodes in the awake state are put into the receive mode during $T_\text{rx}$ for receiving the beacon packet.
After the beacon packet is received, a node in the awake state changes its mode to the active mode during $T_\text{act}$.
While a node is in the active mode, the node measures the receive power in all training slots for the channel estimation.
During this period, the node can also perform sensing and computation pertaining to its own mission.
After the active mode is over, the node transmits a report packet to the power beacon, which contains the receive power measurements obtained during the training slots, the stored energy level, and other sensing results.
For transmitting the report packet, the node is put into a transmit mode during $T_\text{tx}$.
To avoid collisions between the report packets from multiple nodes, the nodes transmit the report packets after different back-off intervals, during which the node is in the idle mode.
After the report packet is sent, the node is set to the idle mode during the rest of the frame.

The power beacon receives the report packets from all the nodes in the awake state in each frame.
The power beacon estimates the channel gains from the receive power measurements in the report packet according to the channel estimation algorithm in \cite{Choi:2017}.
Then, the power beacon decides the beamforming weights and the awake frame ratios for the next frame, based on the estimated channel gains and the reported stored energy level.

\subsection{Stored Energy Evolution Model}

In this subsection, we explain the evolution of the stored energy in each node over frames.
Let $E_k(t)$ denote the stored energy in node $k$ at the start of frame $t$.
The stored energy evolution of node $k$ is governed by the following equation:
\begin{align}\label{eq:enevol}
\begin{split}
E_k(t+1) &= \min\big\{E_k(t) + \Delta^+(r_k(t)) \\
&\qquad\qquad- \Delta^-(a_k(t),E_k(t)),\ E_\text{max}\big\},
\end{split}
\end{align}
where $\Delta^+(r_k(t))$ and $\Delta^-(a_k(t),E_k(t))$ are the harvested energy and consumed energy in node $k$ during frame $t$, respectively.

From \eqref{eq:rhok}, the harvested energy during one frame when the receive power is $r$, is given by
\begin{align}
\Delta^+(r) = \eta T_\text{es}\cdot r.
\end{align}
From \eqref{eq:deltame}, the consumed energy during one frame when the activity is $a$ and the stored energy is $E$, is given by
\begin{align}\label{eq:consume}
\begin{split}
\Delta^-(a,E)
= \begin{cases}
\delta(\text{idle},E)\cdot T_\text{frame}, & \text{if }a=0\\
\sum_{m\in {\mathcal M}}\delta(m,E)\cdot T_m, & \text{if }a=1.
\end{cases}
\end{split}
\end{align}
Note that the consumed energy in the awake state is much larger than that in the sleep state, that is, $\Delta^-(1,E) \gg \Delta^-(0,E)$.

We can rewrite \eqref{eq:consume} as follows:
\begin{align}
\begin{split}
\Delta^-(a,E)& = (\Delta^-(1,E)-\Delta^-(0,E))\cdot a + \Delta^-(0)\\
& = \kappa(E)\cdot a + \varphi(E),
\end{split}
\end{align}
where $\kappa(E) = \Delta^-(1,E)-\Delta^-(0,E)$ and $\varphi(E) = \Delta^-(0,E)$.
In the above equation, $\varphi(E)$ is the energy that is always consumed due to the idle power consumption of a sensor node and the leakage power of a supercapacitor.
On the other hand, $\kappa(E)$ is the energy that is additionally consumed to activate a sensor node for a frame.

The stored energy is stabilized if the expected stored energy variation is equal to or larger than zero.
The expected stored energy variation when the receive power is $r$, the awake frame ratio is $\sigma$, and the stored energy is $E$, is
\begin{align}\label{eq:envar}
\begin{split}
&\ex[\Delta^+(r) - \Delta^-(a,E)|r,\sigma,E]\\
&\qquad= \eta T_\text{es}\cdot r - \kappa(E)\cdot \sigma - \varphi(E).
\end{split}
\end{align}
The energy neutral operation makes the expected stored energy variation in \eqref{eq:envar} no less than zero by controlling the receive power and awake frame ratio.

\subsection{Optimal Control Problem Formulation}

In this subsection, we formulate an optimization problem that maximizes the sum utility of all nodes while satisfying the energy neutrality constraint in each node.
The utility of a node is an increasing function of the awake frame ratio since sensing and reporting can be done more frequently with the higher awake frame ratio.
We define a utility function that maps the awake frame ratio to the utility obtained by each node.
We use the log utility function such that
\begin{align}\label{eq:utility}
\mu(\sigma) = \frac{\sigma^\psi - 1}{\psi},
\end{align}
where $\psi$ is a parameter for the utility function and $\sigma$ is the awake frame ratio.
If $\psi$ goes to zero, the utility function is $\mu(\sigma) = \ln(\sigma)$.
The utility obtained by node $k$ at frame $t$ is $\mu(\sigma_k(t))$.
Note that we have $\mu(\sigma)\le 0$ since $0\le \sigma \le 1$.

To maximize the sum utility of all nodes while satisfying the energy neutrality constraint, we formulate the optimization problem as
\begin{align}\label{eq:optimization}
\begin{split}
&\text{maximize}\qquad \mbox{$\sum_{k=1}^K$} \mu(\sigma_k)\\
&\text{subject to}\\
&\quad \eta T_\text{es}\cdot r_k - \kappa(E_k)\cdot \sigma_k - \varphi(E_k) \ge \epsilon,\\
&\quad\qquad \text{for all $k=1,\ldots,K$ and $E_\text{min}\le E_k\le E_\text{max}$},\\
&\quad 0\le \sigma_k \le 1,\quad\text{for all $k=1,\ldots,K$},\\
&\quad {\bold r} \in \overline{\mathcal R},
\end{split}
\end{align}
where $\epsilon \ge 0$ is a small margin for the energy neutrality constraints.
Let ${\bold r}^*(\epsilon)=(r_1^*(\epsilon),\ldots,r_K^*(\epsilon))^T$ and $\sigma_k^*(\epsilon)$ for $k=1,\ldots,K$ as the optimal solutions to \eqref{eq:optimization}.
With these optimal solutions, the expected stored energy variation satisfies that $\eta T_\text{es}\cdot r_k^*(\epsilon) - \kappa(E_k)\cdot \sigma_k^*(\epsilon) - \varphi(E_k) \ge \epsilon$ for all $k=1,\ldots,K$ and $E_\text{min}\le E_k\le E_\text{max}$.
The optimal value of the optimization problem \eqref{eq:optimization} is denoted by $U^*(\epsilon) = \sum_{k=1}^K \mu(\sigma_k^*(\epsilon))$.

\subsection{Lyapunov Optimization for Joint Beam-Splitting and Energy Neutral Control}\label{section:lyapunov}

We design the optimal control algorithm to find the optimal solution of \eqref{eq:optimization} based on the Lyapunov optimization technique.
The Lyapunov function is defined as
\begin{align}
L(t) = \frac{1}{2}\sum_{k=1}^K (E_\text{max}-E_k(t))^2,
\end{align}
where we will call $(E_\text{max}-E_k(t))$ the stored energy deficiency.
The Lyapunov function is the sum of the squared stored energy deficiencies.
The Lyapunov drift is the expected variation of the Lyapunov function such that
\begin{align}
D(t) = \ex[L(t+1)-L(t)|{\bold E}(t)],
\end{align}
where ${\bold E}(t) = (E_1(t),\ldots,E_K(t))^T$.
The proposed algorithm minimizes the following drift-plus-penalty function to maximize the sum utility while stabilizing the stored energy:
\begin{align}\label{eq:lambda}
D(t) - \lambda \sum_{k=1}^K \mu(\sigma_k(t)).
\end{align}

The following lemma holds for the drift-plus-penalty function.
\begin{lemma}\label{lemma:driftpluspenalty}
The drift-plus-penalty function satisfies the following inequality:
\begin{align}\label{eq:driftpluspenalty}
\begin{split}
&D(t) - \lambda \mbox{$\sum_{k=1}^K$} \mu(\sigma_k(t))\\
&\le -\eta T_\text{es}\mbox{$\sum_{k=1}^K$} (E_\text{max}-E_k(t))\cdot r_k(t)\\
&\quad - \lambda\mbox{$\sum_{k=1}^K$}\big\{\mu(\sigma_k(t))-(\kappa(E_k(t))/\lambda)(E_\text{max}-E_k(t)) \sigma_k(t)\big\}\\
&\quad + \mbox{$\sum_{k=1}^K$}(E_\text{max}-E_k(t))\varphi(E_k(t)) + \Upsilon,
\end{split}
\end{align}
where
\begin{align}
\begin{split}
&\Upsilon = \max_{{\bold r},\sigma_k,E_k}\Big\{\mbox{$\frac{1}{2}\sum_{k=1}^K$} \ex[(\Delta^+(r_k)-\Delta^-(a_k,E_k))^2|E_k]\Big\}\\
&=  \frac{1}{2} \max_{{\bold r},\sigma_k,E_k}
\Big\{ \mbox{$\sum_{k=1}^K$}\big\{(\eta T_\text{es}r_k)^2 - 2\eta T_\text{es}r_k(\kappa(E_k) \sigma_k+\varphi(E_k))\\
&\qquad\qquad\qquad + (\kappa(E_k)^2+2\kappa(E_k)\varphi(E_k))\sigma_k + \varphi(E_k)^2\big\}\Big\}.
\end{split}
\end{align}
\end{lemma}
\begin{proof}
See Appendix~\ref{proof:driftpluspenalty}.
\end{proof}

The proposed algorithm minimizes the right-hand side of \eqref{eq:driftpluspenalty} in each frame.
Let $\widetilde{\bold w}(t) = (\widetilde{w}_1(t),\ldots,\widetilde{w}_N(t))^T$, $\widetilde{\bold r}(t) = (\widetilde{r}_1(t),\ldots,\widetilde{r}_K(t))^T$, and $\widetilde{\sigma}_k(t)$ denote the beamforming weight vector, receive power vector, and awake frame ratio that minimize the right-hand side of \eqref{eq:driftpluspenalty}.
We solve the following optimization problem to obtain $\widetilde{\bold r}(t)$:
\begin{align}\label{eq:optrxpow}
\begin{split}
&\text{maximize}\qquad \mbox{$\sum_{k=1}^K$} (E_\text{max}-E_k(t))\cdot r_k\\
&\text{subject to}\qquad {\bold r}\in \overline{\mathcal R}.
\end{split}
\end{align}
Let us define ${\boldsymbol \Omega}(t) = (E_\text{max}-E_1(t),\ldots,E_\text{max}-E_K(t))^T$ as the stored energy deficiency vector.
Then, the optimal beamforming weight vector for the optimization problem \eqref{eq:optrxpow} is equal to the beamforming weight vector for the beam-splitting beamforming technique with the receive power weight vector ${\boldsymbol \Omega}(t)$.
Therefore, we have
\begin{align}\label{eq:wtopt}
\widetilde{\bold w}(t) = {\bold w}^\text{BS}({\boldsymbol \Omega}(t)).
\end{align}
The corresponding optimal receive power vector is
\begin{align}
\widetilde{\bold r}(t) = |{\bold h}_k^T {\bold w}^\text{BS}({\boldsymbol \Omega}(t))|^2.
\end{align}

To obtain $\widetilde{\sigma}_k(t)$, we solve the following optimization problem:
\begin{align}\label{eq:afr}
\begin{split}
&\text{maximize}\qquad \mu(\sigma)-(\kappa(E_k(t))/\lambda)(E_\text{max}-E_k(t)) \sigma\\
&\text{subject to}\qquad 0\le \sigma \le 1.
\end{split}
\end{align}
The solution to the optimization problem \eqref{eq:afr} is
\begin{align}\label{eq:sigmaktopt}
\widetilde{\sigma}_k(t) = \min\{((\kappa(E_k(t))/\lambda)(E_\text{max}-E_k(t)))^{\frac{1}{\psi-1}},1\}.
\end{align}
In frame $t$, the proposed algorithm sets the beamforming weight to $\widetilde{\bold w}(t)$ in \eqref{eq:wtopt} and the awake frame ratio for node $k$ to $\widetilde{\sigma}_k(t)$ in \eqref{eq:sigmaktopt}.

Now, we prove that the proposed algorithm achieves the optimality in terms of the sum utility while stabilizing the stored energy.
The expected sum utility is defined as
\begin{align}
\limsup_{\tau\rightarrow\infty} \frac{1}{\tau}\sum_{t=1}^\tau \sum_{k=1}^K \ex[\mu(\sigma_k(t))].
\end{align}
In addition, the expected stored energy deficiency of node $k$ is defined as
\begin{align}
\limsup_{\tau\rightarrow\infty} \frac{1}{\tau}\sum_{t=1}^\tau \ex[E_\text{max}-E_k(t)].
\end{align}
The following theorem holds for the expected sum utility and the expected stored energy deficiency.
\begin{theorem}\label{theorem:optimality}
If the beamforming weight vector in frame $t$ is set to $\widetilde{\bold w}(t)$ and the awake frame ratio for node $k$ in frame $t$ is set to $\widetilde{\sigma}_k(t)$ for all $k=1,\ldots,K$, the lower bound of the expected sum utility is given by
\begin{align}\label{eq:optutility}
\limsup_{\tau\rightarrow\infty} \frac{1}{\tau}\sum_{t=1}^\tau \sum_{k=1}^K \ex[\mu(\sigma_k(t))]
\ge U^*(\epsilon) - \frac{\Upsilon}{\lambda}.
\end{align}
In addition, the sum of the expected stored energy deficiencies of all nodes has an upper bound as
\begin{align}\label{eq:optstoredenergy}
\sum_{k=1}^K \limsup_{\tau\rightarrow\infty} \frac{1}{\tau}\sum_{t=1}^\tau \ex[E_\text{max}-E_k(t)]
\le \frac{\Upsilon-\lambda U^*(\epsilon)}{\epsilon}.
\end{align}
\end{theorem}
\begin{proof}
See Appendix~\ref{proof:optimality}.
\end{proof}

From \eqref{eq:optutility}, the expected sum utility converges to the optimal value of the optimization problem \eqref{eq:optimization} (i.e., $U^*(\epsilon)$) as $\lambda$ increases.
With a very small $\epsilon$, the proposed algorithm achieves the optimality in terms of the expected sum utility.
From \eqref{eq:optutility}, we can see that the expected stored energy deficiency is bounded.
As $\lambda$ decreases, the upper bound of the expected stored energy deficiency is reduced.
Therefore, the stored energy can be kept above $E_\text{min}$ by using a sufficiently small $\lambda$.

\section{Experimental Result}\label{section:result}

In this section, we present the experimental results that show the performance of the proposed joint beam-splitting and energy neutral control method.
Besides the proposed method, we have tested the joint time-sharing and energy neutral control method for the purpose of comparison.
This time-sharing control method is easily designed by slightly modifying the beam-splitting control method.
The time-sharing control method maximizes the target function in \eqref{eq:optrxpow} over the time-sharing receive power vectors ${\bold r}^{\text{TS},i}$.
For all results in this section, we use the circular antenna array, and the maximum total transmit power is set to $P_\text{tot}=1120$ mW.
Unless noted otherwise, the parameters for the control methods are set to $\psi=0$, $\lambda = 5\times 10^{-6}$ J$^2$, and $\kappa(E) = 2.77\times 10^{-4}$ J.

In Figs.~\ref{fig:timebs} and \ref{fig:timets}, we show the operation of the beam-splitting and time-sharing control methods over time.
For these figures, nodes 1, 2, and 3 are located at (1.5m, 0$^\circ$), (1.5m, 120$^\circ$), and (1.5m, 240$^\circ$), respectively, at the start.
At 20 minutes after the start, we move away node 3 from the transmit antenna to the location (2m, 240$^\circ$) to see how the proposed method adapts to the location change.

Fig.~\ref{fig:timebs} shows the beamforming weight, receive power, stored energy deficiency, and awake frame ratio of the beam-splitting control method.
Until 20 minutes, the stored energy deficiency is stably maintained thanks to the adaptive control of the beamforming weight and the awake frame ratio.
At 20 minutes, the system is agitated due to the movement of node 3 and it starts to adapt to the new channel condition.
The receive power of node 3 is reduced since node 3 is moved away from the transmit antenna.
The beamforming weights are changed as well to the values suitable for the new channel gains of node 3.
The receive powers of nodes 1 and 2 are slightly reduced to focus more power on node 3.
Due to the reduced receive power, the stored energy deficiency of node 3 increases over time.
However, the control method lowers the awake frame ratio of node 3 to reduce the power consumption, and the system is stabilized within a few minutes.

\begin{figure}
    \centering
    \subfigure[Beamforming weight]{
        \label{fig:timeweightbs}\includegraphics[width=4.0cm, bb=1in 0.3in 9in 7.4in] {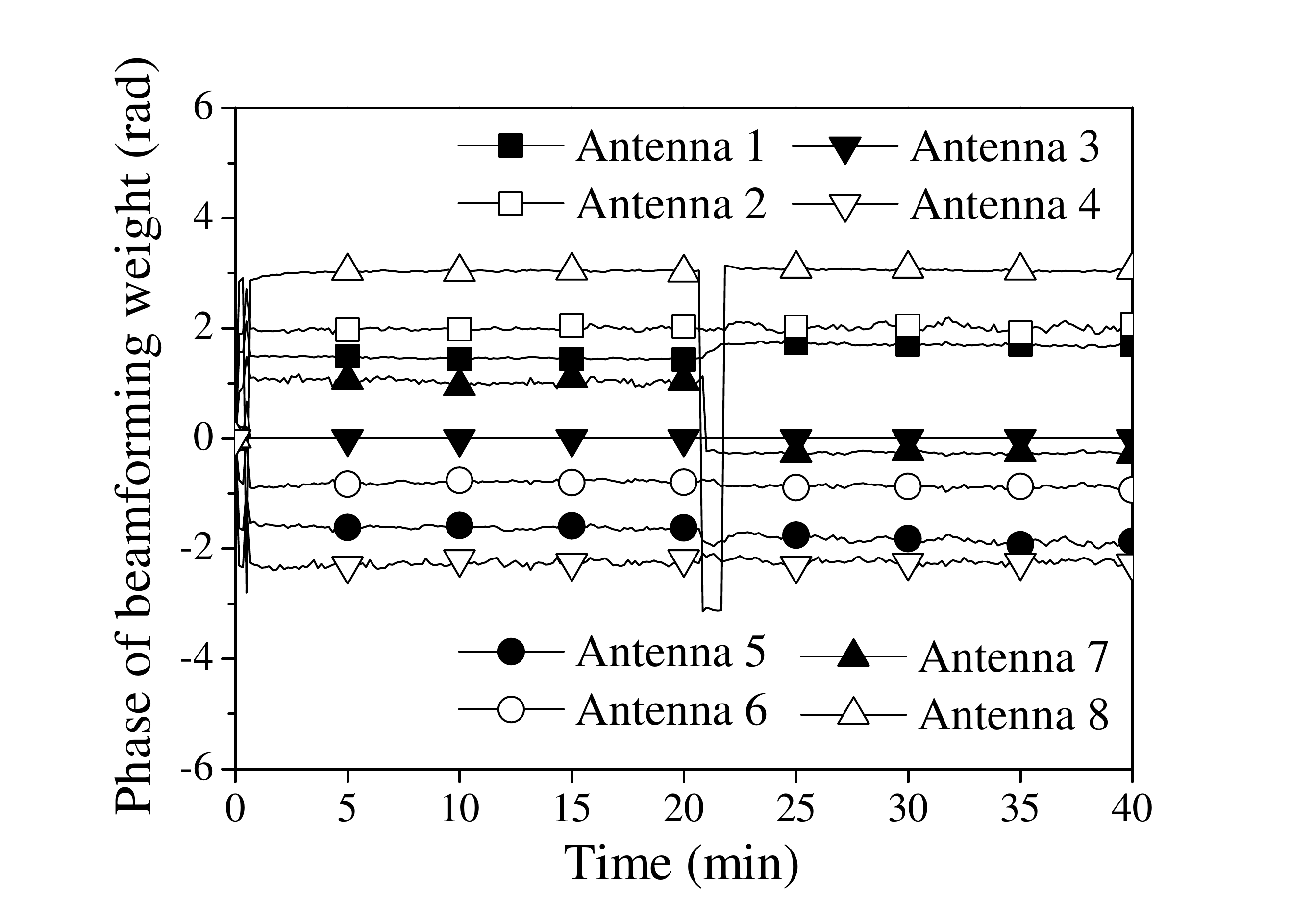}
        }~
    \subfigure[Receive power]{
        \label{fig:timerxpowbs}\includegraphics[width=4.0cm, bb=1in 0.3in 9in 7.4in] {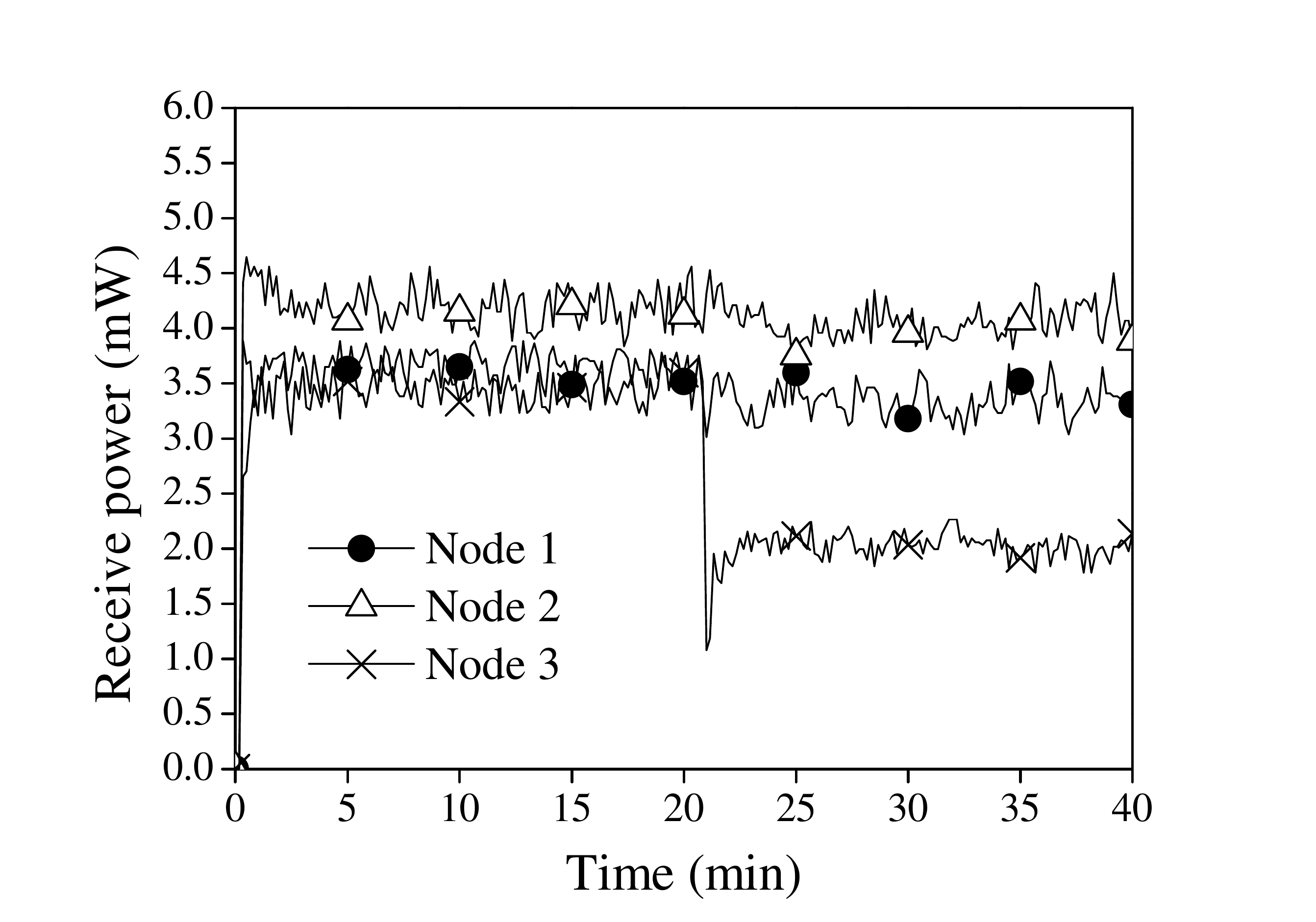}
        }\\
    \subfigure[Stored energy deficiency]{
        \label{fig:timedeficiencybs}\includegraphics[width=4.0cm, bb=1in 0.3in 9in 7.4in] {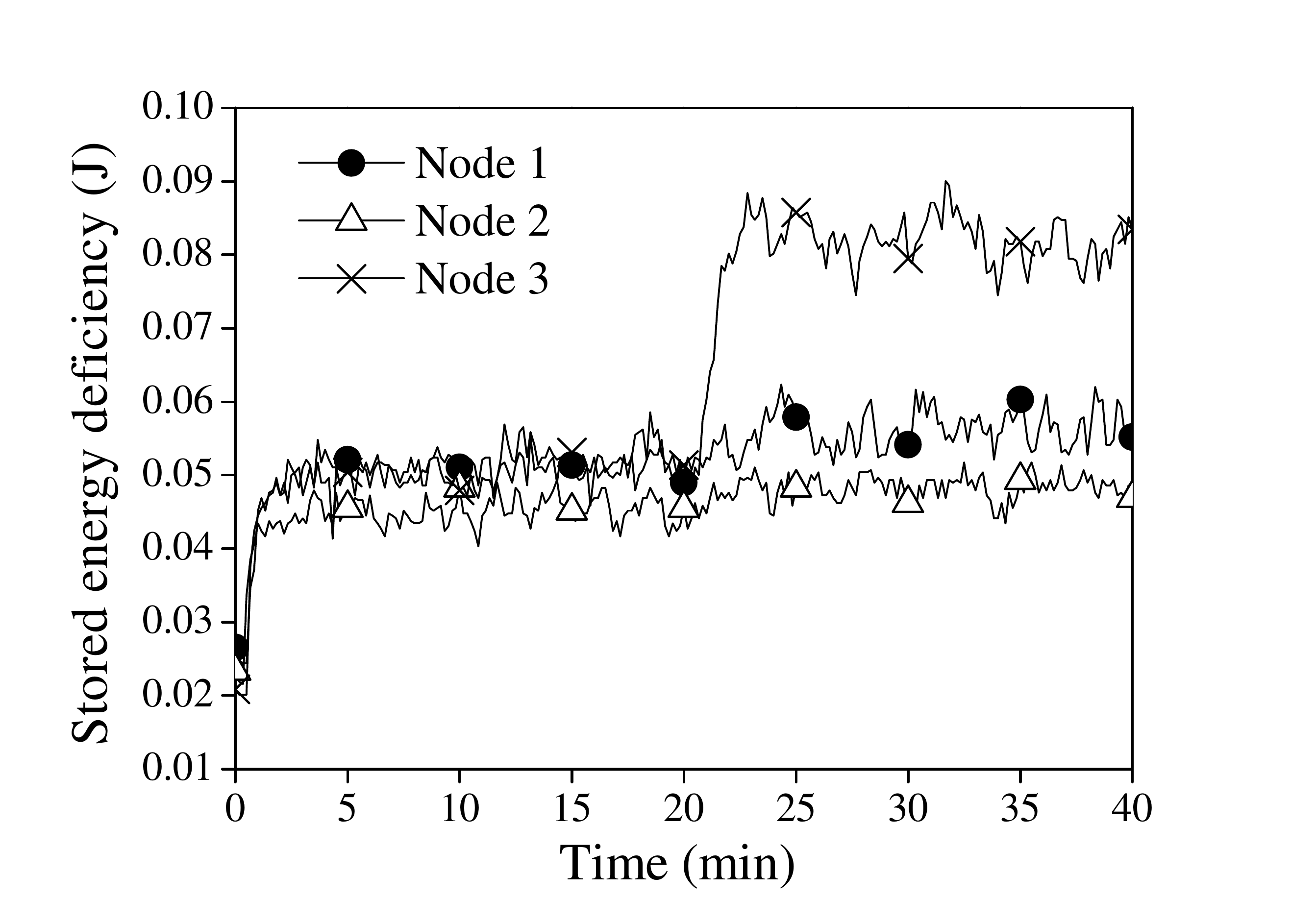}
        }~
    \subfigure[Awake frame ratio]{
        \label{fig:timeafrbs}\includegraphics[width=4.0cm, bb=1in 0.3in 9in 7.4in] {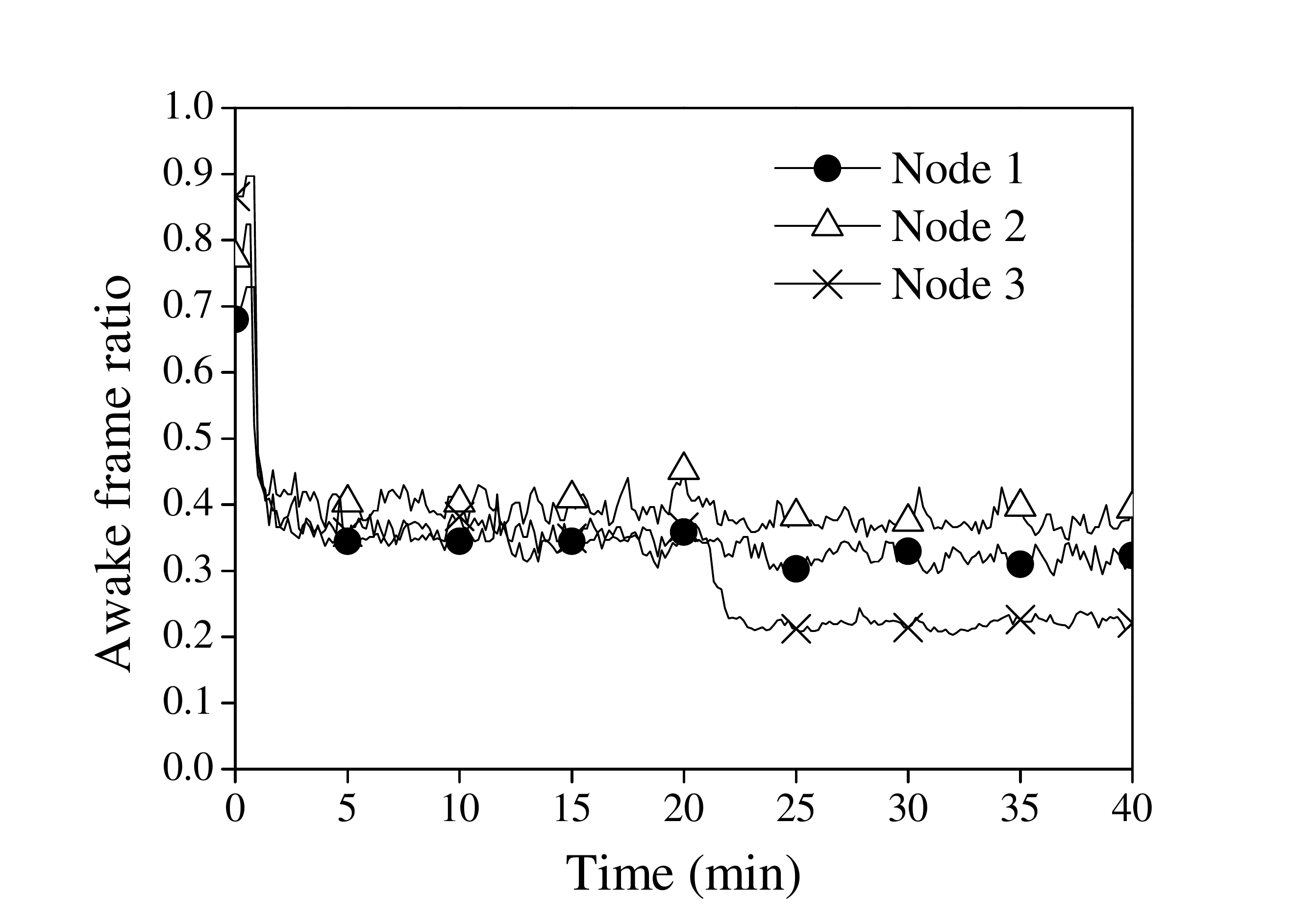}
        }
    \caption{Operation of the joint beam-splitting and energy neutral control method.}
    \label{fig:timebs}
\end{figure}

The operation of the time-sharing control method in Fig.~\ref{fig:timets} is very similar to the one of the beam-splitting control method, except that the performance of the time-sharing control method is slightly lower (i.e., higher stored energy deficiency and lower awake frame ratio) due to its inefficiency.
In Fig.~\ref{fig:timets}, we show the beamforming weight and receive power during a very short period (i.e., 10 seconds) since they change very rapidly.
In Fig.~\ref{fig:timerxpowts}, we can see that the power is focused on one node at a time, and each node is alternately selected in a time-sharing manner.

\begin{figure}
    \centering
    \subfigure[Beamforming weight]{
        \label{fig:timeweightts}\includegraphics[width=4.0cm, bb=1in 0.3in 9in 7.4in] {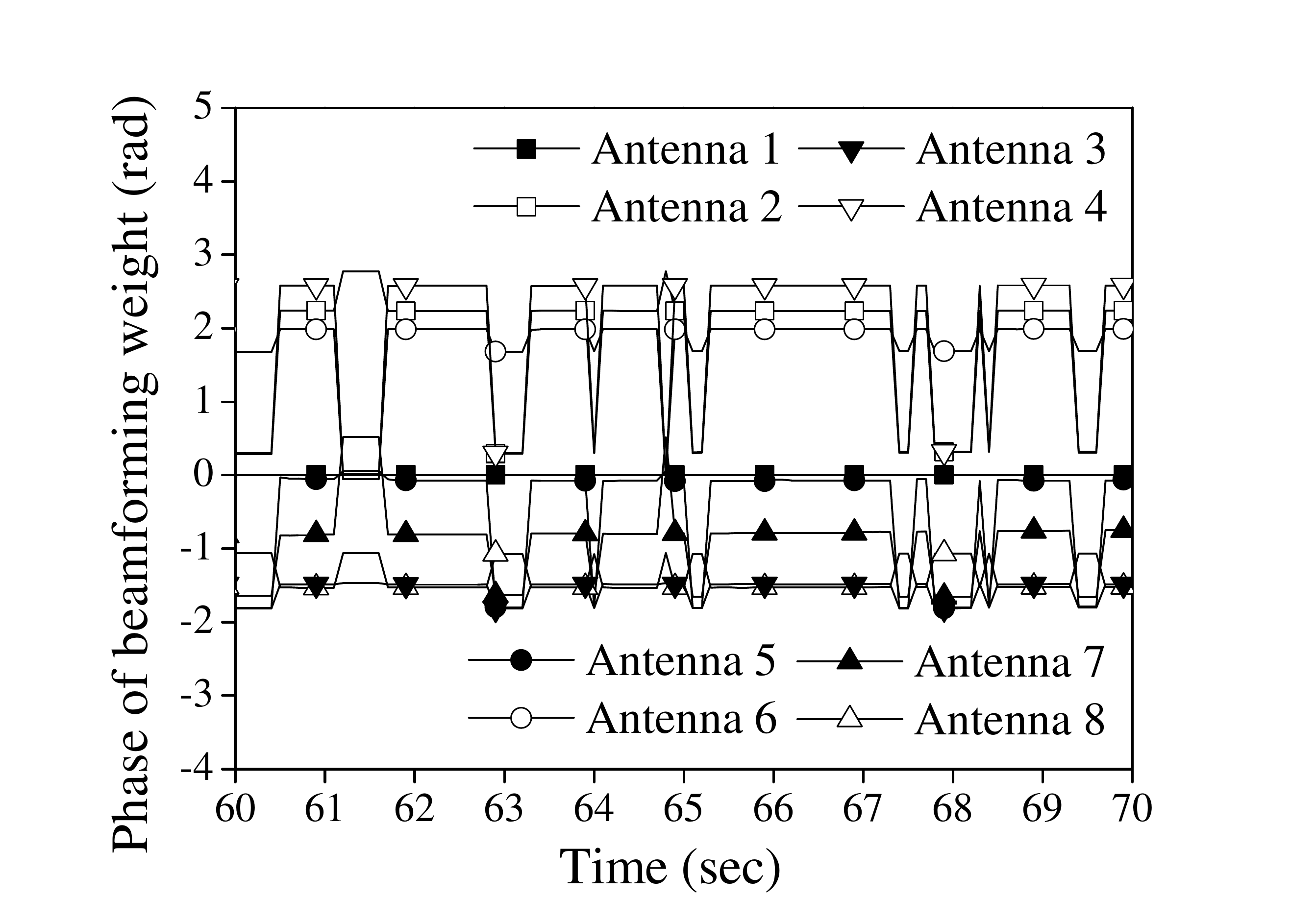}
        }~
    \subfigure[Receive power]{
        \label{fig:timerxpowts}\includegraphics[width=4.0cm, bb=1in 0.3in 9in 7.4in] {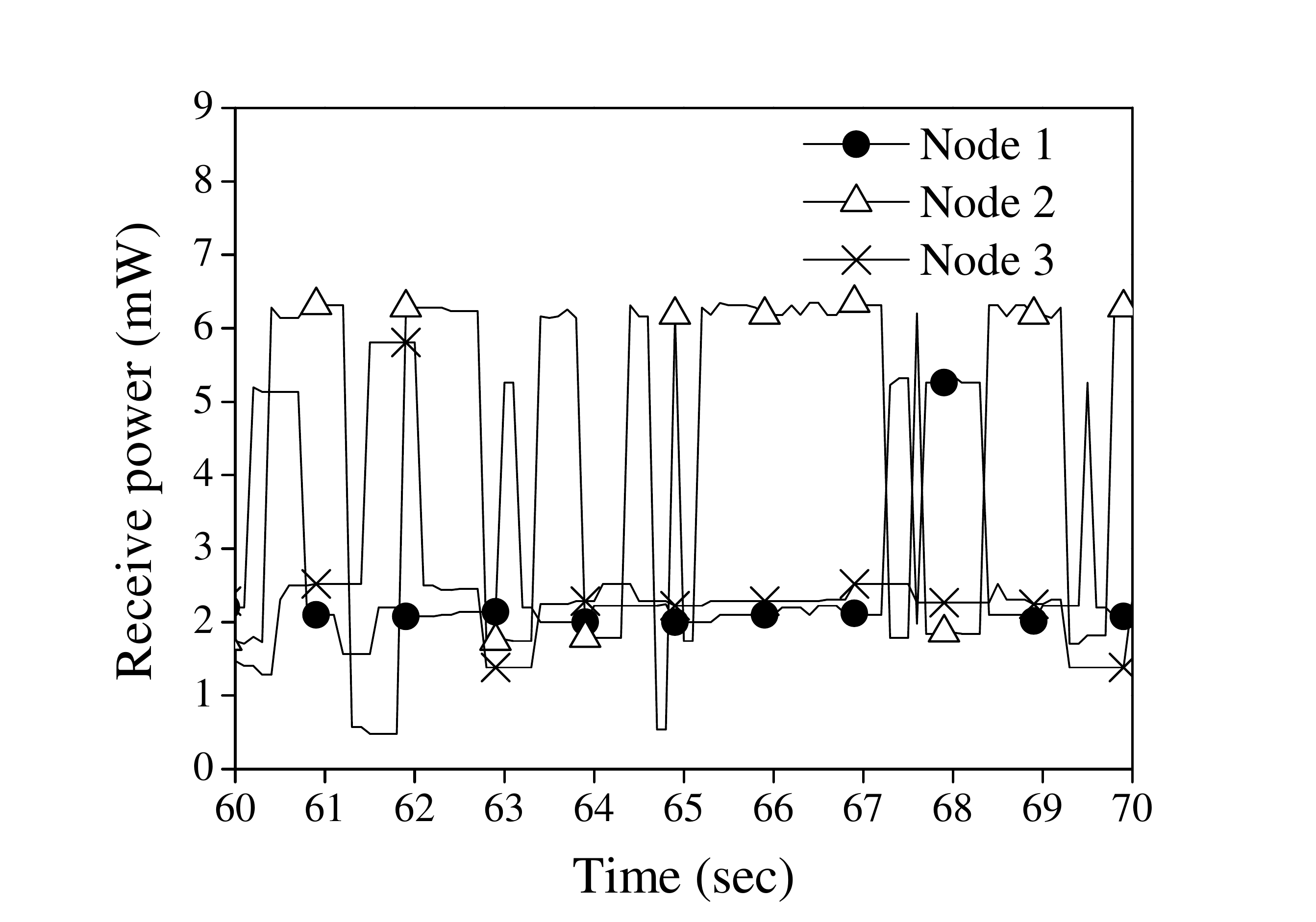}
        }\\
    \subfigure[Stored energy deficiency]{
        \label{fig:timedeficiencyts}\includegraphics[width=4.0cm, bb=1in 0.3in 9in 7.4in] {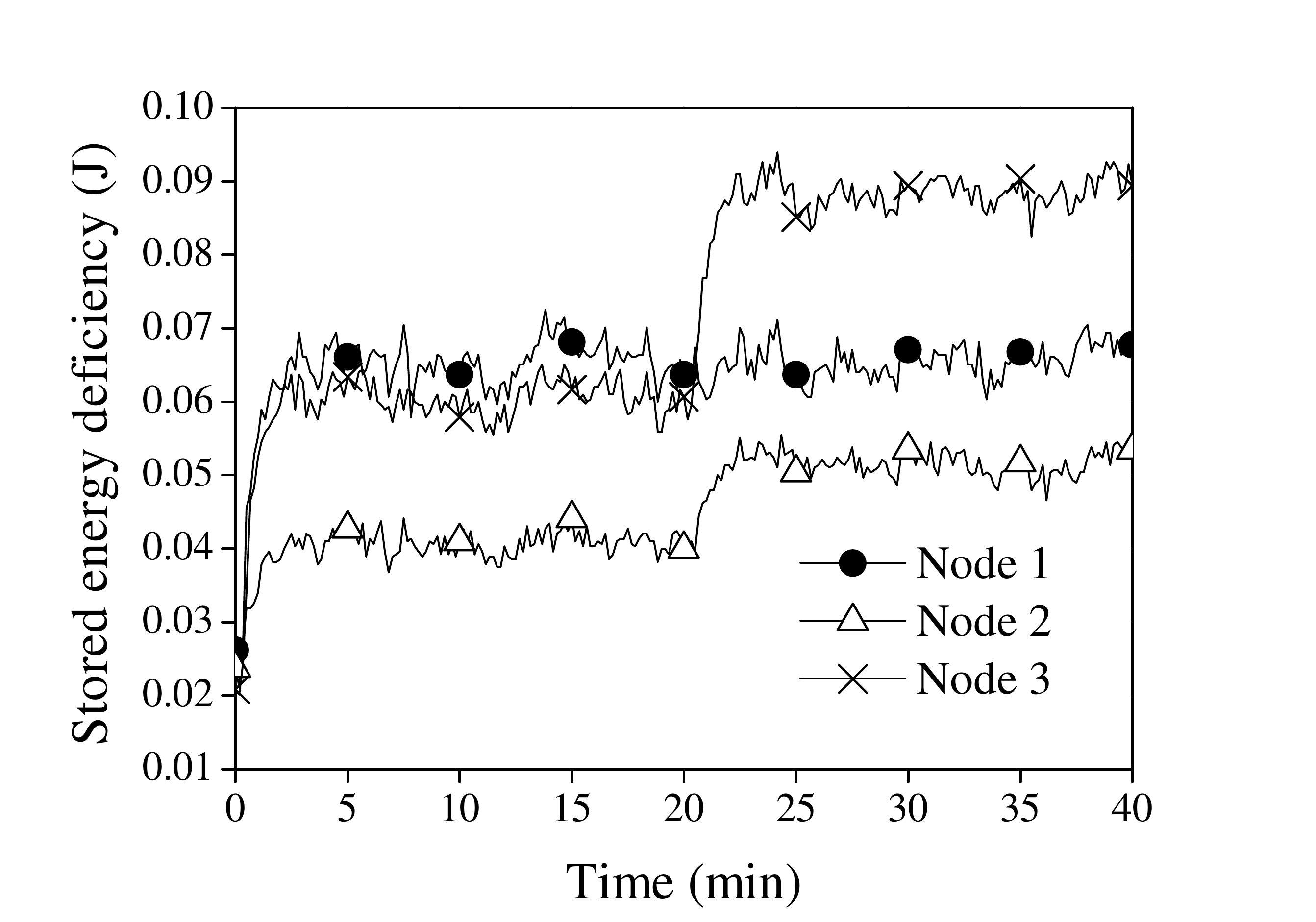}
        }~
    \subfigure[Awake frame ratio]{
        \label{fig:timeafrts}\includegraphics[width=4.0cm, bb=1in 0.3in 9in 7.4in] {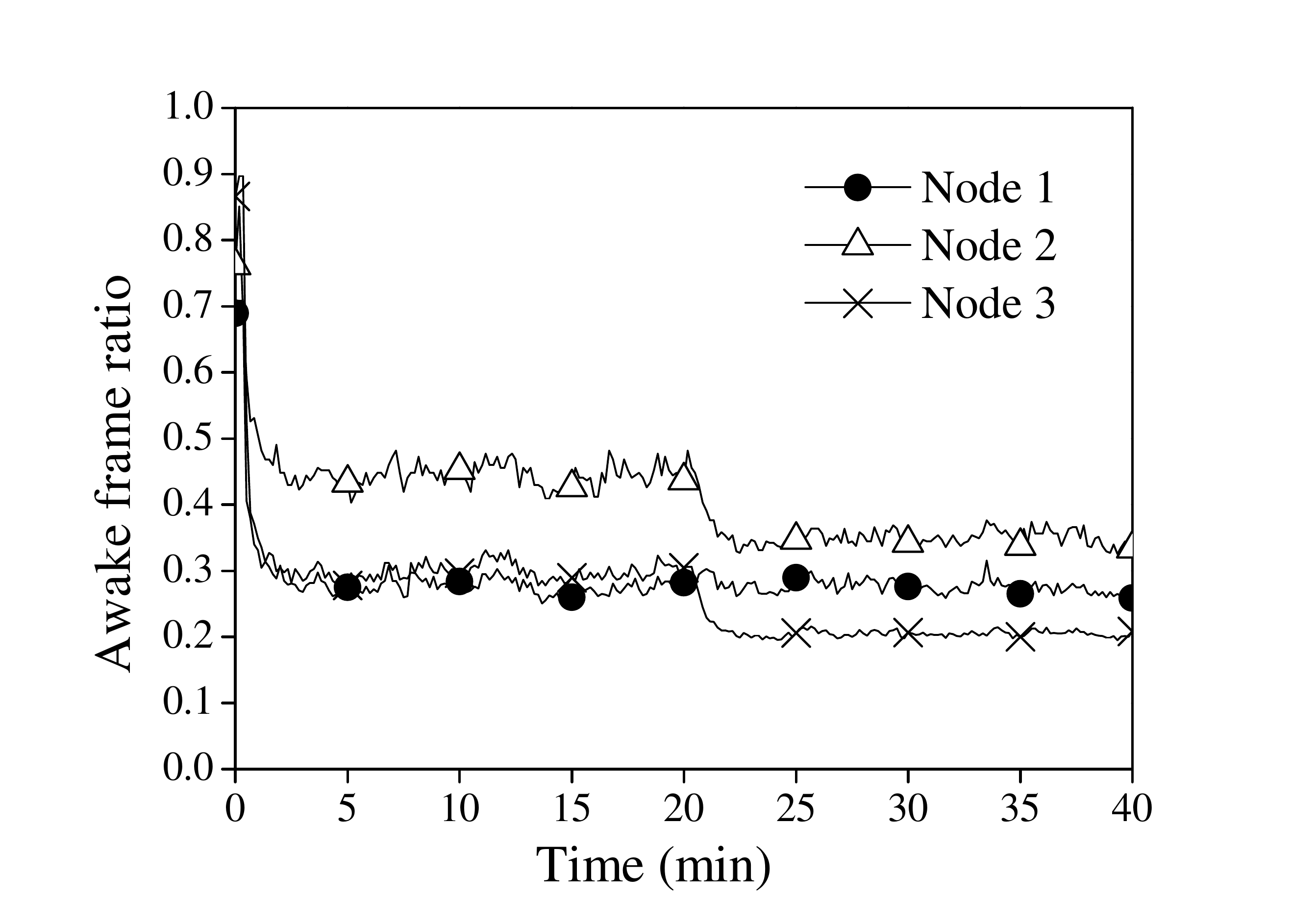}
        }
    \caption{Operation of the joint time-sharing and energy neutral control method.}
    \label{fig:timets}
\end{figure}

Fig.~\ref{fig:psi} shows the awake frame ratio and stored energy deficiency according to $\psi$.
The parameter $\psi$ is used for the utility function as given in \eqref{eq:utility}, and controls the fairness of the awake frame ratio.
In this figure, nodes 1, 2, and 3 are located at (1m, 0$^\circ$), (1.5m, 120$^\circ$), and (2m, 240$^\circ$), respectively.
To obtain the awake frame ratio and stored energy deficiency, we take an average over more than 100 seconds after they are stabilized.
In Fig.~\ref{fig:psiafr}, we can see that the awake frame ratios become less fairly controlled as we use higher $\psi$.
On the other hand, the higher $\psi$ enhances overall system efficiency in sacrifice of fairness, which results in lower stored energy deficiency as seen in Fig.~\ref{fig:psideficiency}.
We can also see that the beam-splitting control method (i.e., BS) generally has higher awake frame ratio and lower stored energy deficiency than the time-sharing control method (i.e., TS) has.

\begin{figure}
    \centering
    \subfigure[Awake frame ratio]{
        \label{fig:psiafr}\includegraphics[width=4.0cm, bb=1in 0.3in 9in 7.4in] {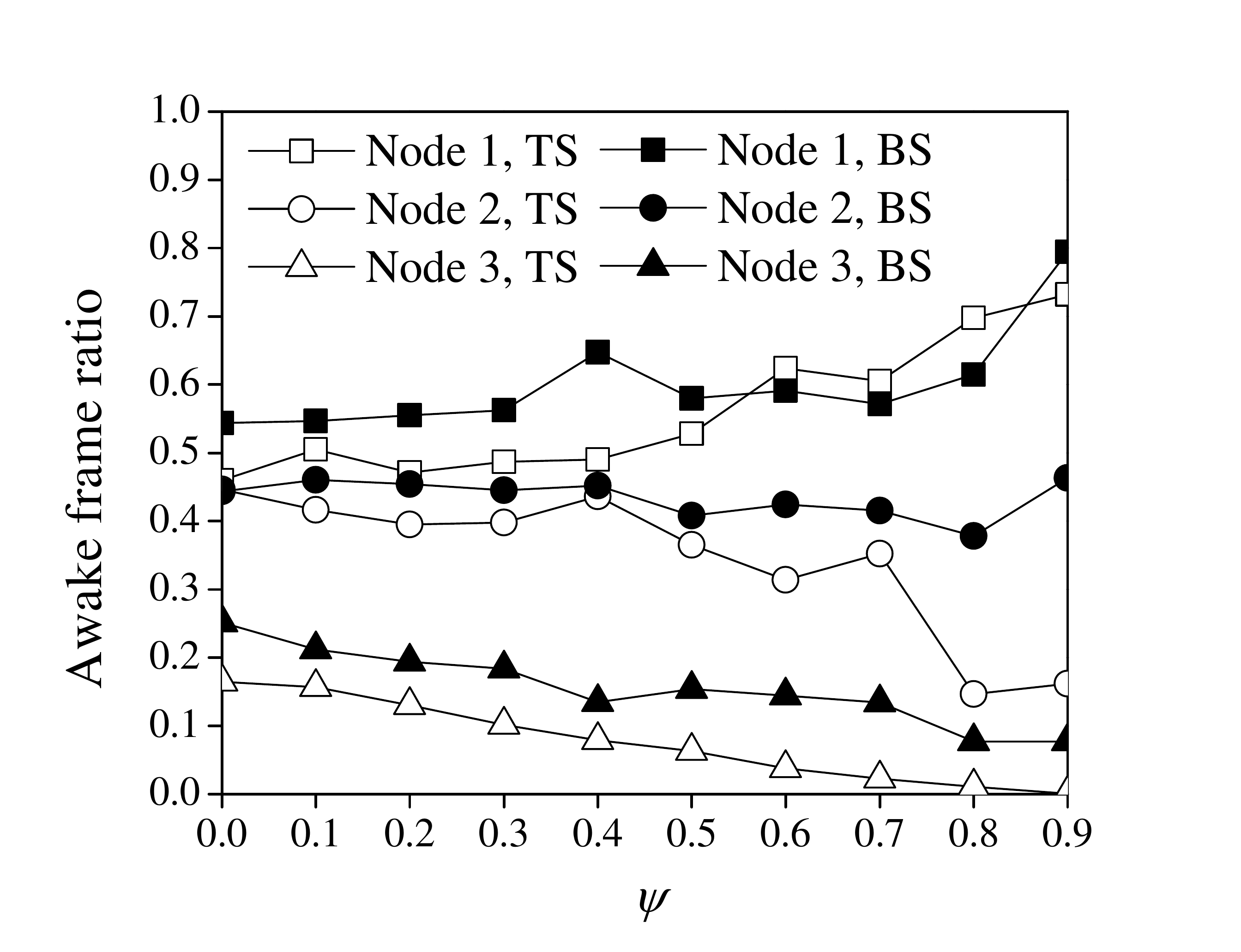}
        }~
    \subfigure[Stored energy deficiency]{
        \label{fig:psideficiency}\includegraphics[width=4.0cm, bb=1in 0.3in 9in 7.4in] {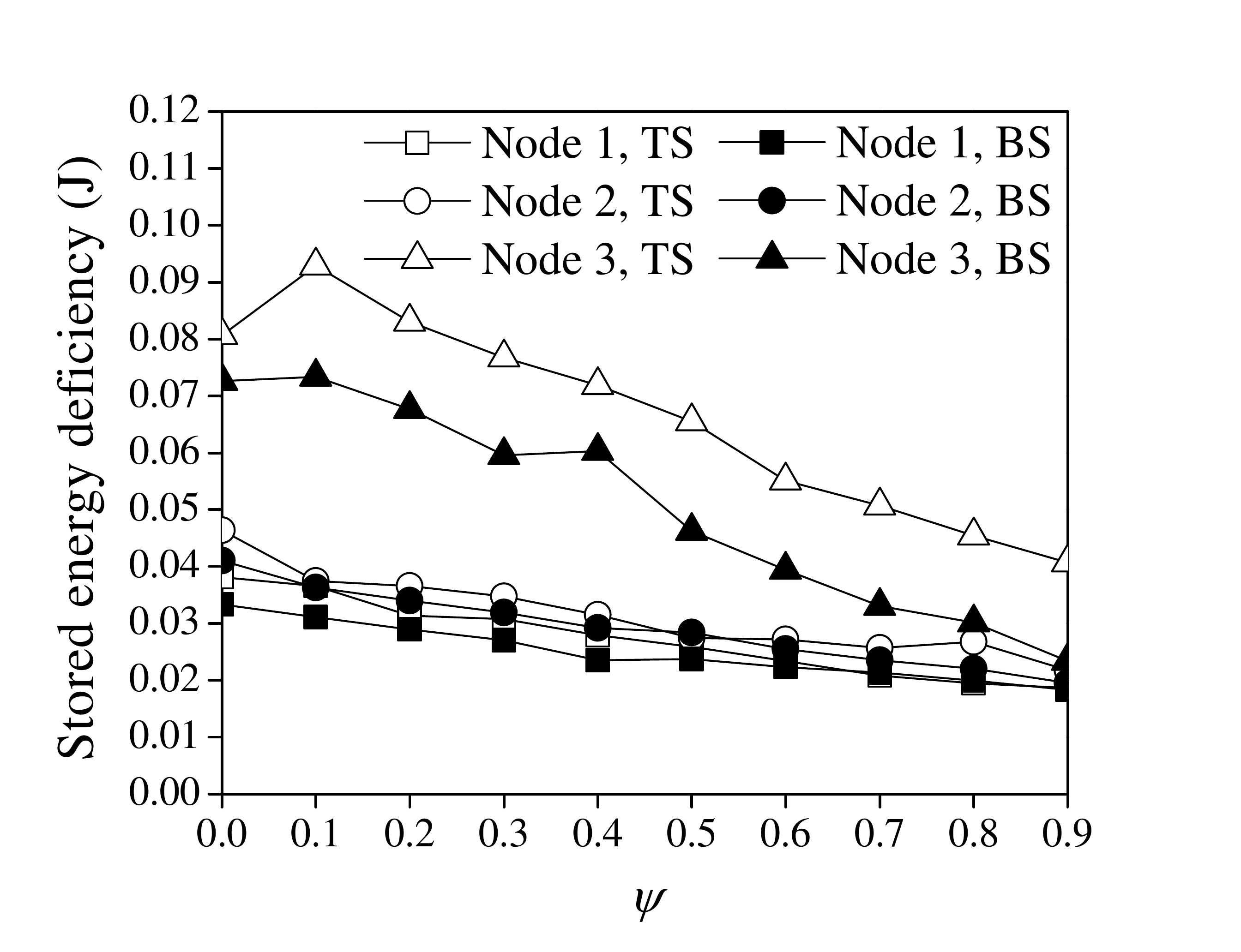}
        }
    \caption{Awake frame ratio and stored energy deficiency according to $\psi$.}
    \label{fig:psi}
\end{figure}

Fig.~\ref{fig:lambda} shows the utility and stored energy deficiency as a function of $\lambda$.
For this graph, the nodes are placed in the same locations as in Fig.~\ref{fig:psi}.
As seen in \eqref{eq:lambda}, $\lambda$ represents the importance of the sum utility in relation to the drift of the stored energy deficiency.
In Fig.~\ref{fig:lambda}, we can see that higher $\lambda$ results in higher utility at the cost of increased stored energy deficiency.
This figure also shows that the beam-splitting control method outperforms the time-sharing control method.

\begin{figure}
    \centering
    \subfigure[Utility]{
        \label{fig:lambdautility}\includegraphics[width=4.0cm, bb=1in 0.3in 9in 7.4in] {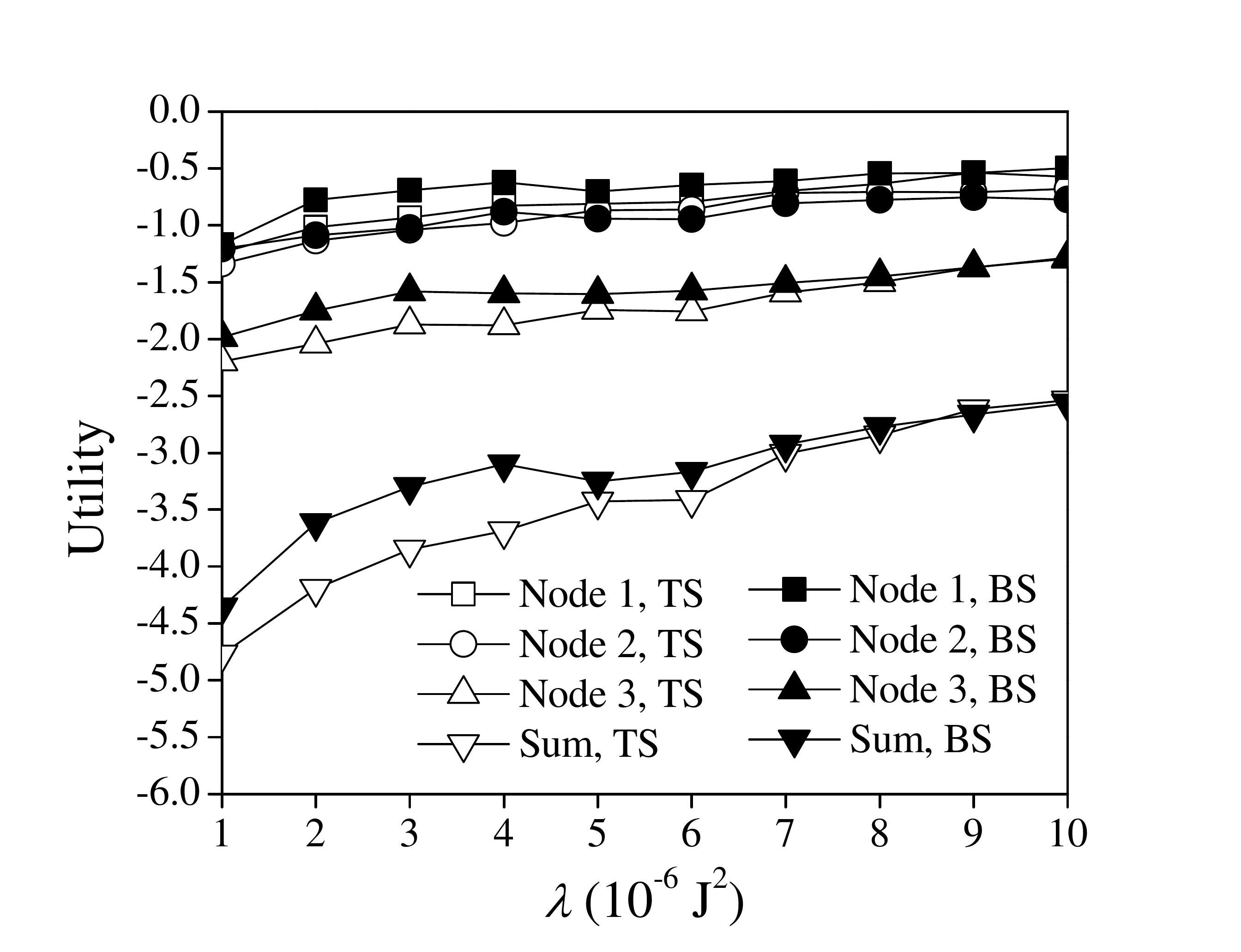}
        }~
    \subfigure[Stored energy deficiency]{
        \label{fig:lambdadeficiency}\includegraphics[width=4.0cm, bb=1in 0.3in 9in 7.4in] {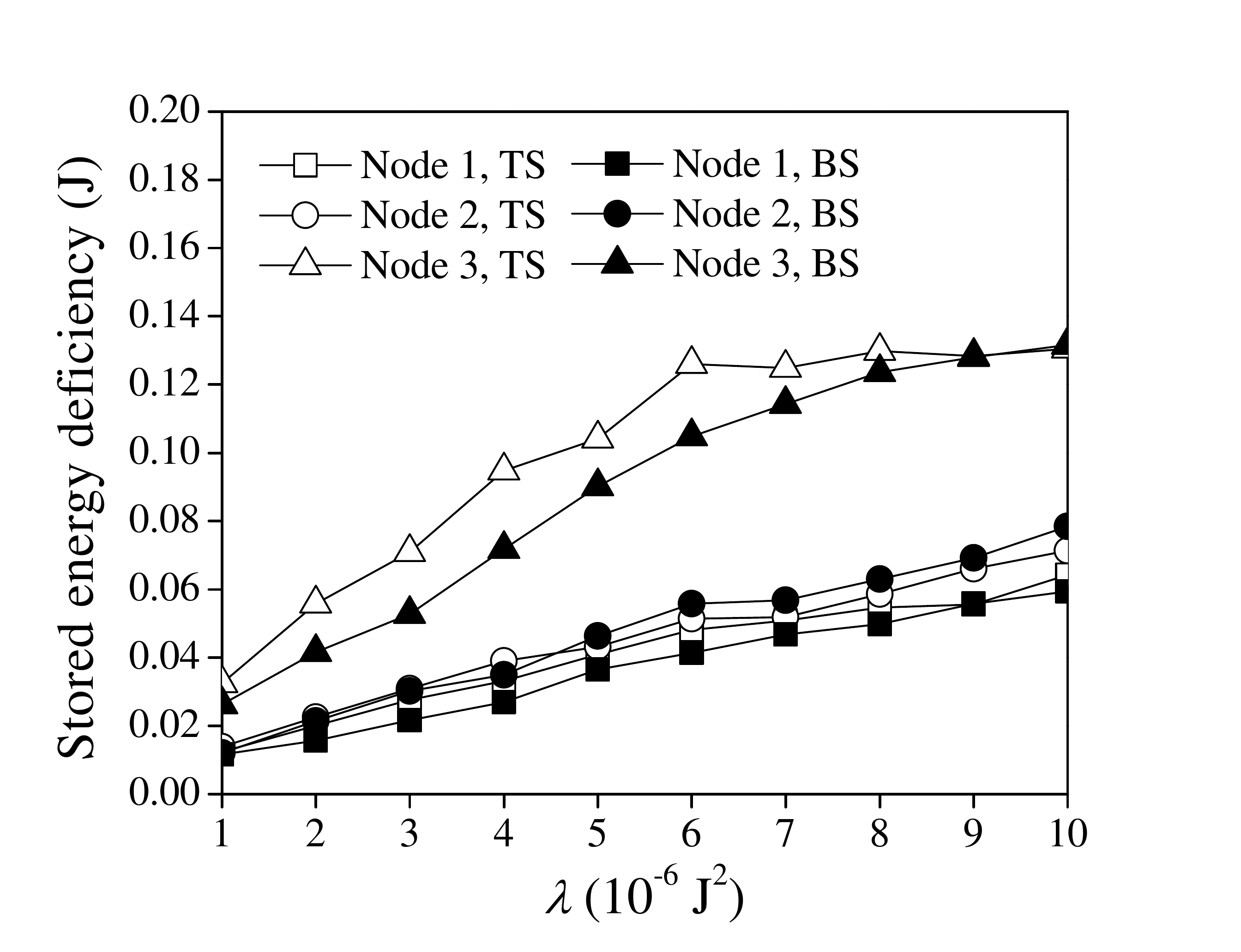}
        }
    \caption{Utility and stored energy deficiency according to $\lambda$.}
    \label{fig:lambda}
\end{figure}

\section{Conclusion}\label{section:conclusion}

In this paper, we have designed the joint beam-splitting and energy neutral control algorithm for the multi-node multi-antenna WPSN.
By using this proposed algorithm, we have solved two most important problems in the design of the multi-node multi-antenna WPSN: how to distribute RF power to multiple nodes and how to keep them alive for a perpetual operation.
The contribution of this work is distinguished from previous theoretical works in that the proposed algorithm is validated by the experiments in the real testbed.

\appendices

\section{Proof of Lemma~\ref{lemma:driftpluspenalty}}\label{proof:driftpluspenalty}

The drift-plus-penalty function satisfies that
\begin{align}\label{eq:proofdpp}
\begin{split}
&D(t) - \lambda \mbox{$\sum_{k=1}^K$} \mu(\sigma_k(t))\\
&\le\mbox{$\frac{1}{2}\sum_{k=1}^K$}\ex\big[ (E_\text{max}-\min\big\{E_k(t) + \Delta^+(r_k(t))\\
&\qquad\qquad\qquad\qquad\qquad\qquad - \Delta^-(a_k(t),E_k(t)),\ E_\text{max}\big\})^2\\
&\qquad\qquad\qquad- (E_\text{max}-E_k(t))^2\big|{\bold E}(t)] - \lambda \mbox{$\sum_{k=1}^K$} \mu(\sigma_k(t))\\
&\le -\mbox{$\sum_{k=1}^K$}(E_\text{max}-E_k(t))\\
&\qquad\times\ex[\Delta^+(r_k(t))-\Delta^-(a_k(t),E_k(t))|{\bold E}(t)]\\
&\quad+\mbox{$\frac{1}{2}\sum_{k=1}^K$}\ex[(\Delta^+(r_k(t))-\Delta^-(a_k(t),E_k(t)))^2|{\bold E}(t)]\\
&\quad- \lambda \mbox{$\sum_{k=1}^K$} \mu(\sigma_k(t))\\
&\le -\eta T_\text{es}\mbox{$\sum_{k=1}^K$} (E_\text{max}-E_k(t))\cdot r_k(t)\\
&\quad- \lambda\mbox{$\sum_{k=1}^K$}\big\{\mu(\sigma_k(t))-(\kappa(E_k(t))/\lambda)(E_\text{max}-E_k(t)) \sigma_k(t)\big\}\\
&\quad+ \mbox{$\sum_{k=1}^K$}(E_\text{max}-E_k(t))\varphi(E_k(t)) + \Upsilon.
\end{split}
\end{align}

\section{Proof of Theorem~\ref{theorem:optimality}}\label{proof:optimality}

If the proposed optimal control algorithm is used, we have the following inequality from \eqref{eq:proofdpp}:
\begin{align}\label{eq:proofopt1}
\begin{split}
&D(t) - \lambda \mbox{$\sum_{k=1}^K$} \mu(\sigma_k(t))\\
&\quad\le -\mbox{$\sum_{k=1}^K$}
(E_\text{max}-E_k(t))(\eta T_\text{es}\cdot \widetilde{r}_k(t) \\
&\quad\qquad\qquad- \kappa(E_k(t))\cdot \widetilde{\sigma}_k(t) - \varphi(E_k(t))) \\
&\quad\quad- \lambda \mbox{$\sum_{k=1}^K$} \mu(\widetilde{\sigma}_k(t)) + \Upsilon\\
&\quad\le -\mbox{$\sum_{k=1}^K$}
(E_\text{max}-E_k(t))(\eta T_\text{es}\cdot r^*_k(\epsilon) \\
&\quad\qquad\qquad- \kappa(E_k(t))\cdot \sigma^*_k(\epsilon) - \varphi(E_k(t))) \\
&\quad\quad- \lambda \mbox{$\sum_{k=1}^K$} \mu(\sigma^*_k(\epsilon)) + \Upsilon\\
&\quad\le -\mbox{$\sum_{k=1}^K$}
(E_\text{max}-E_k(t))\epsilon - \lambda U^*(\epsilon) + \Upsilon
\end{split}
\end{align}
By taking $\limsup_{\tau\rightarrow \infty}\frac{1}{\tau}\sum_{t=1}^\tau \ex$ on the both sides of \eqref{eq:proofopt1}, we have the following inequality:
\begin{align}\label{eq:proofopt2}
\begin{split}
&- \lambda \limsup_{\tau\rightarrow\infty} \frac{1}{\tau}\sum_{t=1}^\tau \sum_{k=1}^K \ex[\mu(\sigma_k(t))]\\
&\qquad + \epsilon\cdot \sum_{k=1}^K \limsup_{\tau\rightarrow\infty} \frac{1}{\tau}\sum_{t=1}^\tau
\ex[E_\text{max}-E_k(t)] \le \Upsilon - \lambda U^*(\epsilon).
\end{split}
\end{align}
From \eqref{eq:proofopt2}, we can obtain \eqref{eq:optutility} and \eqref{eq:optstoredenergy}.

\bibliographystyle{IEEEtran}
\bibliography{IEEEabrv,MWPSN_BIB}

\end{document}